\documentclass[a4paper,twocolumn,11pt,accepted=2025-04-14]{quantumarticle}
\pdfoutput=1

\usepackage[numbers, sort&compress]{natbib}

\AtBeginDocument{%
  \setlength{\textwidth}{174mm}
}

\usepackage{enumitem}

\usepackage[utf8]{inputenc}
\usepackage[english]{babel}
\usepackage{diagbox}
\usepackage{makecell}
\usepackage{textcomp}
\usepackage{amsmath,etoolbox}
\usepackage{amsthm}
\usepackage{amssymb}
\usepackage[top=1cm]{geometry}
\usepackage{dsfont}
\usepackage{xspace}
\usepackage{bm}
\usepackage{leftidx}
\usepackage{float}
\makeatletter
\let\newfloat\newfloat@ltx
\makeatother
\usepackage[makeroom]{cancel}
\usepackage{indentfirst}
\usepackage{underoverlap}
\usepackage{hyperref}
\usepackage[caption=false]{subfig}
\usepackage{graphicx}
\usepackage[capitalise]{cleveref}
\crefname{section}{Sec.}{Secs.}
\Crefname{section}{Section}{Sections}
\crefname{figure}{Fig.}{Figs.}
\Crefname{figure}{Figure}{Figures}
\crefname{appendix}{App.}{Apps.}
\Crefname{appendix}{Appendix}{Appendices}
\crefname{equation}{Eq.}{Eqs.}
\Crefname{equation}{Equation}{Equations}
\crefname{definition}{Def.}{Defs.}
\Crefname{definition}{Definition}{Definitions}
\crefname{theorem}{Thm}{Thms.}
\Crefname{theorem}{Theorem}{Theorems}
\crefname{corollary}{Corollary}{Cors.}
\crefname{remark}{Remark}{Remarks}
\crefname{proposition}{Prop.}{Props.}
\usepackage{tikz}
\usetikzlibrary{
quantikz2,
decorations.pathmorphing, patterns,
cd}
\usepackage{ulem}
\usepackage{array}
\newcolumntype{L}[1]{>{\raggedright\let\newline\\\arraybackslash\hspace{0pt}}m{#1}}
\newcolumntype{C}[1]{>{\centering\let\newline\\\arraybackslash\hspace{0pt}}m{#1}}
\newcolumntype{R}[1]{>{\raggedleft\let\newline\\\arraybackslash\hspace{0pt}}m{#1}}
\usepackage{tikz-cd}
\usepackage[titletoc,title]{appendix}

\usepackage{adjustbox}
\usepackage{mathtools}
\usepackage{multirow}
\usepackage{placeins}
\usepackage{xargs}
\usepackage{xstring}
\usepackage{rotating}
\usepackage{calc}
\usepackage{upgreek}
\usepackage{algorithm}
\usepackage{algpseudocode}
\usepackage{tabularray}
\usepackage{pdftexcmds}
\makeatletter
\let\scshape\relax % to avoid a warning
\DeclareRobustCommand\scshape{%
  \not@math@alphabet\scshape\relax
  \ifnum\pdf@strcmp{\f@family}{\familydefault}=\z@
    \fontfamily{qbk}%
  \fi
  \fontshape\scdefault\selectfont}
\makeatother
\usepackage{afterpage}

\DeclareMathAlphabet{\dutchcal}{U}{dutchcal}{m}{n}
\SetMathAlphabet{\dutchcal}{bold}{U}{dutchcal}{b}{n}
\DeclareMathAlphabet{\dutchbcal} {U}{dutchcal}{b}{n}

\usetikzlibrary{matrix,arrows,decorations.pathmorphing}
\newenvironment{sloppypar*}{\sloppy\ignorespaces}{\par}
\allowdisplaybreaks % Allow page breaks within environments like align and alignat

\newcommand {\sbra} [1] {\langle #1 |}
\newcommand {\sket} [1] {| #1 \rangle}

\newcommand{\me}{\mathrm{e}}
\newcommand {\unit} {\mathds{1}}

\newcommand{\fex}{f}
\newcommand{\Fex}{F}

\newcommand{\n}{\mathsf{n}}
\newcommand{\nq}{\mathsf{n}_{\mathsf{q}}}

\newcommand{\NLat}{{\mathsf{N_\Lambda}}}

\newcommand{\OO}{\mathcal{O}}

\newcommand{\dims}{\mathsf{d}}

\newcommand{\FT}{\widehat{\mathcal{F}}}

\newcommand{\Hgate}{\mathrm{H}}
\newcommand{\Xgate}{\mathrm{X}}
\newcommand{\Zgate}{\mathrm{Z}}

\usepackage{amsthm}

\newtheorem{theorem}{Theorem}
\newtheorem{lemma}[theorem]{Lemma}

\theoremstyle{definition}

\newcommand{\PREPARE}{{\textsc{prep}\xspace}}
\newcommand{\SELECT}{{\textsc{sel}\xspace}}

\makeatletter
\def\hlinewd#1{%
\noalign{\ifnum0=`}\fi\hrule \@height #1 \futurelet
\reserved@a\@xhline}
\makeatother

\DeclareRobustCommand{\capitalhyphen}{\raisebox{0.24ex}{\resizebox{0.4em}{\height}{-}}\kern-0.07em}

\DeclareRobustCommand{\LOVELCU}{{LOVE{\capitalhyphen}LCU}\xspace}

\newcommand{\pd}{{\vphantom{\dagger}}}

\begin{document}
    
\title{Block encoding bosons by signal processing}

\author{Christopher F. Kane}
\affiliation{Department of Physics, University of Arizona, Tucson, Arizona 85719, USA}

\author{Siddharth Hariprakash}
\affiliation{Center for Theoretical Physics and Department of Physics, University of California, Berkeley, California 94720, USA}
\affiliation{Physics Division, Lawrence Berkeley National Laboratory, Berkeley, California 94720, USA}
\affiliation{National Energy Research Scientific Computing Center (NERSC), Lawrence Berkeley National Laboratory,
Berkeley, CA 94720, USA,
}

\author{Neel S. Modi}
\affiliation{Center for Theoretical Physics and Department of Physics, University of California, Berkeley, California 94720, USA}
\affiliation{Physics Division, Lawrence Berkeley National Laboratory, Berkeley, California 94720, USA}

\author{Michael Kreshchuk}
\affiliation{Physics Division, Lawrence Berkeley National Laboratory, Berkeley, California 94720, USA}
\affiliation{Phasecraft Inc., Washington, DC 20001, USA}

\author{Christian W Bauer}
\affiliation{Physics Division, Lawrence Berkeley National Laboratory, Berkeley, California 94720, USA} 

\begin{abstract}

Block Encoding (BE) is a crucial subroutine in many modern quantum algorithms, including those with near-optimal scaling for simulating quantum many-body systems, which often rely on Quantum Signal Processing (QSP).
Currently, the primary methods for constructing BEs are the Linear Combination of Unitaries (LCU) and the sparse oracle approach.
In this work, we demonstrate that QSP-based techniques, such as Quantum Singular Value Transformation (QSVT) and Quantum Eigenvalue Transformation for Unitary Matrices (QETU), can themselves be efficiently utilized for BE implementation.
Specifically, we present several examples of using QSVT and QETU algorithms, along with their combinations, to block encode Hamiltonians for lattice bosons, an essential ingredient in simulations of high-energy physics.
We also introduce a straightforward approach to BE based on the exact implementation of Linear Operators Via Exponentiation and LCU (\LOVELCU).
We find that, while using QSVT for BE results in the best asymptotic gate count scaling 
with the number of qubits per site, \LOVELCU outperforms all other methods for operators acting on up to $\lesssim11$ qubits, highlighting the importance of concrete circuit constructions over mere comparisons of asymptotic scalings.
Using \LOVELCU to implement the BE, we simulate the time evolution of single-site and two-site systems in the lattice $\varphi^4$ theory using the Generalized QSP algorithm and compare the gate counts to those required for Trotter simulation.

\end{abstract}

\maketitle

\tableofcontents

%%%%%%%%%%%%%%%%%%%%%%%%%%%%%%%%%%%

\section{Introduction\label{sec:introduction}}

Quantum simulation of quantum many-body physics was the original motivation for the development of quantum computing~\cite{Manin1980,Benioff:1980fes,Feynman:1981tf} and continues to be considered one of its most compelling applications.
For many areas of interest, such as quantum chemistry~\cite{Goings:2022jfz}, nuclear~\cite{Savage:2023qop}, or high-energy physics~\cite{Bauer:2022hpo}, fault-tolerant devices will be necessary to perform calculations of practical value.

Most quantum simulation proposals involve the stages of state preparation and/or time evolution.
Both of these tasks are commonly accomplished with the aid of algorithms based on Product Formulas (PF)~\cite{Lloyd:1996aai,zalka1998simulating}, an approach often referred to as \emph{Trotterization}.
Recent years, however, have witnessed a surge in the development of algorithms based on the so-called \emph{post-Trotter} methods such as Linear Combination of Unitaries (LCU)~\cite{Childs:2012gwh} or Quantum Signal Processing (QSP)~\cite{Low:2016sck,Low:2016znh}.
Although the asymptotic cost of these algorithms can achieve optimal or near-optimal dependence on problem parameters, they often come with large constant prefactors compared to asymptotically sub-optimal methods based on Product Formulas (PF)~\cite{childs2019nearly,Morales:2022ipc,Ikeda:2022tlb,Watson:2024dvw,Watson:2024yqs,Bosse:2024maw}.
Consequently, their usage is sometimes limited to highly specific parameter ranges, such as long evolution times or small errors of the final state~\cite{Hariprakash:2023tla}.

An important ingredient of post-Trotter methods is to encode the Hamiltonian directly into a quantum circuit.
Given that quantum circuits can only implement unitary, not Hermitian, operators, the Hamiltonian is typically encoded as a sub-block of a larger unitary matrix.
This process is called Block Encoding (BE), and is the source for most of the prefactors in the algorithmic cost of post-Trotter algorithms.
The commonly used approaches to constructing block encodings include the LCU algorithm and the sparse oracle approach~\cite{Kreshchuk:2020dla,Kirby:2021ajp,Lin:2022vrd,Anderson:2024kfj,Du:2024ixj}.
Highly optimized constructions of BEs based on both approaches have been developed for various types of physical systems~\cite{Babbush:2017oum,Babbush:2018ywg,Rhodes:2024zbr}.

In this work, we present several novel approaches to constructing BEs for a certain class of Hamiltonians. 
We demonstrate that, while BEs are commonly seen as building blocks for QSP-based algorithms, \emph{QSP-based algorithms themselves can be utilized for the highly efficient construction of BEs}. 
Our approaches apply to Hamiltonians which can be written as 
\begin{align}
\label{eq:GeneralH}
    \hat{H} = \sum_{{\ell}=0}^{N-1} \hat{H}_{\ell}
    \,,
\end{align}
where each $\hat{H}_{\ell}$ acts on at most $n$ qubits.
Many Hamiltonians have this structure, and we shall term them as $n$-\emph{site-local} Hamiltonians.
We will focus on the situation where each term in the Hamiltonian can be diagonalized via a known and efficiently implementable transformation
\begin{align}
\label{eq:diagonal}
    \hat{H}_{\ell}^\pd = U_{\ell}^\dagger \hat{H}_\ell^{(D)} U_\ell^\pd
    \,.
\end{align}

For some of our methods, we will make a further assumption about the structure of the Hamiltonian, which is satisfied for certain bosonic lattice field theories, which have recently garnered attention for their crucial role in high-energy and low-energy nuclear physics~\cite{Jordan:2017lea,Lu:2018pjk,Bauer:2019qxa,Bauer:2023ujy,Chigusa:2022act,Klco:2018kyo,Raychowdhury:2018osk,Raychowdhury:2019iki,Klco:2019evd,Barata:2020jtq,Ciavarella:2021nmj,Bauer:2021gek,Grabowska:2022uos,Farrell:2022wyt,Farrell:2022vyh,Bauer:2022hpo,Watson:2023oov,DAndrea:2023qnr,Farrell:2023fgd,Peng:2023bzl,Hariprakash:2023tla,Nagano:2023uaq,Farrell:2024fit,Rhodes:2024zbr,Ciavarella:2024fzw,Du:2024ixj,Gomes:2024tup}.
Specifically, we consider theories for which the Hamiltonians are formulated in terms of conjugate operators $[\hat \varphi_i, \hat \pi_j] = i \delta_{ij}$ at lattice sites $i, j$.
Using $\nq$ qubits, the operators $\hat \varphi_i$ are digitized by sampling at equally spaced values.
Since the operator $\hat \pi_i$ is conjugate to $\hat \varphi_i$, it can be diagonalized using a Fourier transform $\hat \pi_i = \FT^\dagger \hat \pi_i^{(D)} \FT$.
The Hamiltonian can be written as functions of simple operators $\hat\varphi_i$ and $\hat\pi_i$, and possibly combinations like $\hat\varphi_i-\hat\varphi_j$ needed for gradient terms.
We will use the general notation $\hat \xi_k$ ($k \in \{0,1,2\}$) to denote any of the operators
\begin{align}
    \hat{\xi}_0 \in \{\hat\varphi_i\}\,,\quad 
    \hat{\xi}_1 \in \{\hat\pi^{(D)}_i\}\,, \quad 
    \hat{\xi}_2 \in \{\hat{\varphi}_i-\hat{\varphi}_j\}
    \,,
\end{align}
such that each $\hat{H}_{\ell}^{(D)}$ is then given by a function of one such operator 
\begin{align}
    \hat{H}_{\ell}^{(D)} = f_{k_{\ell}}(\hat \xi_{k_{\ell}})
    \,.
\end{align}
A key property of the operators $\hat\xi_{k_\ell}$ we exploit in this work is that each of them can be written as a linear combination of $\mathcal{O}(\nq)$ many Pauli $Z$ gates.

Based on this set up and these assumptions, we present 3 different methods for the BE of the individual Hamiltonian terms $\hat{H}_{\ell}$, which are then combined via LCU, see~\cref{fig:alg_outline}.
The first method first obtains a block encoding for the operator $\hat \xi_{k_{\ell}}$ and then uses QSVT to construct the BE for $\hat{H}_{\ell}^{(D)}$. 
The second method constructs the unitary operator $\me^{-i \tau \hat{\xi}_{k_{\ell}}}$ (or $e^{-2 i\arccos(\hat{\xi}_k/\alpha)}$) and uses it as a building block to construct the BE for $\hat{H}_{\ell}^{(D)}$ using QETU. 
The final method, which we call Linear Operators Via Exponentiation and Linear Combination of Unitaries (\LOVELCU) constructs the unitary operators $\me^{\pm i \arccos(f_{k_\ell}(\hat{\xi}_{k_{\ell}})/\alpha)}$, from which $\hat{H}_{\ell}^{(D)}$ can easily be constructed by adding the two terms via LCU. Following the construction of the BE for $\hat{H}_{\ell}^{(D)}$ we conjugate it by $U_{\ell}$ resulting in a BE for $\hat{H}_{\ell}$.

In~\cref{sec:review_alg}, we briefly review the notion of BE as well as the GQSP, QSVT, and QETU algorithms.
In~\cref{sec:be}, we first review the general formulation of bosonic field theories considered in this work and then discuss the 3 methods in more detail.
We also demonstrate that, due to the highly symmetric structure of QSVT, QETU, and \LOVELCU circuits, the qubitized walk operator can be constructed without ancillary qubits in a straightforward way.
In~\cref{sec:scalar_field_theory_numerics} we use the proposed methods to construct BEs for several models of interest, construct explicit circuits, and provide metrics such as gate count and the number of ancillary qubits.
For a single site system, we find that using our methods allows to implement time evolution with less gates than PFs for errors as small as $\epsilon \sim 10^{-2}$.
Our conclusions are presented in~\cref{sec:discussion}.

\begin{figure}
    \centering
    \includegraphics[width=\linewidth]{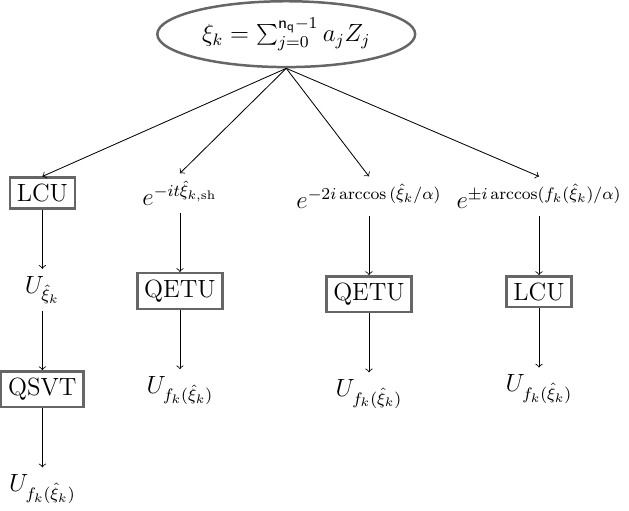}
    \caption{
Various methods for block encoding operators in Hamiltonians of scalar field theories (see~\cref{sec:be}). In the first two approaches, the cost of constructing the building block is polynomial in $\nq$, while in the latter two, the cost is exponential in $\nq$. In the first three approaches, the query depth is determined by the convergence of the polynomial approximation to the desired function. In the last approach, the query depth is constant (two).
    }
    \label{fig:alg_outline}
\end{figure}

\section{Algorithms}
\label{sec:review_alg}

In this section, we provide brief reviews of the simulation algorithms and techniques we consider in this work, along with associated circuit diagrams.
Readers familiar with these concepts are encouraged to skip directly to~\cref{sec:be}.

\subsection{Block encoding and walk operator\label{sec:review_be}}

The key ingredient of most near-optimal simulation algorithms is
the construction of
the so-called \emph{Block Encoding} (BE) of a given operator.
BE typically refers to the embedding of an $n$-qubit operator $A \in \mathcal{C}^{2^n \times 2^n}$, acting on an $n$-qubit Hilbert space $\mathcal{H}_s$, in the (without loss of generality) principle block of a larger $(m+n)$-qubit unitary operator $U_A \in \mathcal{C}^{2^{n+m} \times 2^{n+m}}$, acting on a larger Hilbert space $\mathcal{H}_a \otimes \mathcal{H}_s$, where $\mathcal{H}_a$ is the $m$-qubit ancilla space.
The BE $U_A$ has the following form:
\begin{equation}
\label{eq:bematr}
    U_A = \begin{pmatrix}
        A/\alpha & * \\
        * & *
    \end{pmatrix},
\end{equation}
where each $*$ represents some block to be chosen such that $U_A$ is unitary and the quantity $\alpha$, known as the scale factor, is chosen such that $\lvert\lvert A / \alpha \rvert\rvert_2 \leq 1$ since any sub-block of a unitary matrix must have its singular values upper bounded by 1. In this case, we call $U_A$ the $(\alpha,m)$-block-encoding of $A$.
Schematically, we see that for $\ket{0}_a \equiv \sket{0^{\otimes m}}_a \in \mathcal{H}_a$ and $\ket{\psi}_s \in \mathcal{H}_s$ the action of $U_A$ can be written as
\begin{equation}
    U_A \ket{0}_a\ket{\psi}_s = \frac{1}{\alpha}\begin{pmatrix}
        A \ket{\psi}_s \\
        *
    \end{pmatrix},
\end{equation}
and, thus, we recover the action of the operator $A$ on $\ket{\psi}_s$ only if the ancillary register is $\ket{0}_a$, see~\cref{fig:BE}.\footnote{Note that while this is a common choice, in general one could use a different state to project out the sub-block containing the rescaled matrix of interest~\cite{Low:2016znh}.}
The probability that this occurs, known as the \emph{success probability}, is given by $(\lvert\lvert A\ket{\psi}_s\rvert\rvert / \alpha)^2$.
This can also be written as
\begin{equation}
    \label{eq:beproj}
    A/\alpha = (\sbra{0}_a\otimes \unit_s) U_A (\sket{0}_a\otimes \unit_s)\,,
\end{equation}
which implies that
\begin{equation}
    \label{eq:beact}
    U_A \sket{0}_a\sket{\psi}_s = \sket{0}_a (A/\alpha \sket{\psi}_s) + \sket{\bot}_{as}\,,
\end{equation}
where $\sket{\bot}_{as}$ is a vector perpendicular to $\sket{0}_a$ in the sense that $(\sbra{0}_a\otimes\unit_s)\sket{\bot}_{as}=0$.
For a more formal treatment of block encodings, see Refs.~\cite{Lin:2022vrd,Camps:2022jnx,camps2022fable,Gilyen:2018khw,Hariprakash:2023tla}.

\begin{figure}[h]
\centering
\begin{quantikz}
\lstick{$\sket{0}_a$}& \gate[2]{U_A} \qwbundle{}& \meter{} & \setwiretype{c} \rstick{$\sket{0}_a$}\\
\lstick{$\sket{\psi}_s$} & \qwbundle{} \qw & \qw \rstick{$\dfrac{A\sket{\psi}_s}{\|A\sket{\psi}_s\|}$}
\end{quantikz}
\caption{Circuit for implementing the block encoding $U_A$ of an operator $A$, as defined in~\cref{eq:beproj}. Upon measuring the ancillary register in the state $\sket{0}_a$, the state of the system $\sket{\psi}_s$ is mapped to ${A\sket{\psi}_s}/{\|A\sket{\psi}_s\|}$.}
\label{fig:BE}
\end{figure}
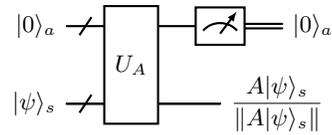

The most common QSP-based approach to simulating dynamics governed by time-independent Hamiltonians~\cite{Low:2016znh} approximates the operator $\me^{-i{H}t}$\footnote{Without loss of generality, we rescale the Hamiltonian and the evolution time as ${H}\to \alpha {H}$ and $t\to t/\alpha$ to ensure that $||{H}||_2 \leq 1$.} using a special form of block of encoding of ${H}$  (also known as the \emph{Szegedy quantum walk operator} or the \emph{iterate}) denoted by $W_H$.
This is known as \emph{qubitization}.\footnote{Confusingly, \emph{qubitization} often refers to the entire process of simulating time evolution based on QSP and repeated calls to $W_H$~\cite{Low:2016znh}. Moreover, sometimes \emph{qubitization} is used as an umbrella term for all quantum simulation algorithms based on the usage of BE. In the present paper, \emph{qubitization} refers solely to the process of constructing $W_H$.}
The operator $W_H$ is designed in such a way that upon its repeated application, one recovers block encodings for \emph{higher powers of} ${H}$ (which does not occur for a general block encoding $U_H$).
This can be achieved if for all Hamiltonian eigenstates $\sket{\lambda}$, the applications of $W_H$ to states of the form $\sket{0}_a\sket{\lambda}_s$ produce a state belonging to the subspace $\operatorname{span}\{\sket{0}_a \sket{\lambda}_s, \sket{\bot^\lambda}_{as}\}$:
\begin{equation}
\begin{alignedat}{9}
\label{eq:Whdef}
&W_H \sket{0}_a\sket{\lambda}_s &&= &&\lambda\, \sket{0}_a\sket{\lambda}_s + \sqrt{1-\lambda^2} &&\sket{\bot^\lambda}_{as} \,,
\\
&W_H \sket{\bot^\lambda}_{as} &&= -&&\sqrt{1-\lambda^2} \sket{0}_a\sket{\lambda}_s + \lambda &&\sket{\bot^\lambda}_{as} \,.
\end{alignedat}    
\end{equation}

The walk operator $W_H$ can be written as
\begin{equation}
    \label{eq:whrsu}
    W_H = \left( R_0 \otimes \unit_s \right) \left(S \otimes \unit_s \right) U_H\,,
\end{equation}
where $R_0 = \bigl(2\ket{0}_a\bra{0}_a - \unit_a\bigr)$ implements reflection with respect to the state $\sket{0}_a$, if there exists an operator $S$ satisfying~\cite{Low:2016znh}
\begin{subequations}
\label{eq:S_cond}
\begin{align}
    \label{eq:S_cond1}
    &\left(\sbra{0}_a \otimes \unit_s \right) \left(S \otimes \unit_s\right) U_H \left(\sket{0}_a \otimes \unit_s \right) = H/\alpha\,, 
    \\
    \label{eq:S_cond2}
    &\left(\sbra{0}_a \otimes \unit_s \right) \left[ \left(S \otimes \unit_s\right) U_H \right]^2 \left(\sket{0}_a \otimes \unit_s \right) = \mathrlap{\unit_s.}
\end{align}
\end{subequations}
Finding an operator $S$ that satisfies the above relations for a given $U_H$ is generally a difficult task.
In situations where $S$ is unknown, one can always construct $W_H$ by adding an additional ancillary qubit and performing two controlled calls to $U_H$~\cite{Low:2016sck}.
While this procedure is fully general, it introduces a large overall coefficient in the cost to construct $W_H$.
Importantly, for all BE methods considered in this work, the operator $S$ was determined and the costly general procedure was avoided.
For a more careful treatment of this subject, including the motivation for and construction of the iterate, we refer the reader to Refs.~\cite{Low:2016znh, Hariprakash:2023tla}.

\subsection{Linear Combination of Unitaries \label{ssec:lcu}}

The Linear Combination of Unitaries (LCU) method enables the construction of BEs of operators by breaking an operator $T$ acting on some system Hilbert space $\mathcal{H}_s$ into a sum of into unitary operators $U_i$, each with known implementation:
\begin{equation}
    T = \sum_{i = 0}^{K-1} c_iU_i\,,
\end{equation}
where the coefficients $c_i$ are chosen to be real and positive.
To implement LCU, one adds an ancillary qubit register with $\lceil \log K \rceil$ qubits and defines a \emph{select} oracle (denoted as $\SELECT$) as follows:
\begin{equation}
    \label{eq:select}
    \SELECT \equiv \sum_{i=0}^{K-1} \ket{i}\bra{i} \otimes U_i\,.  
\end{equation}
$\SELECT$ implements each unitary $U_i$ conditioned on the state $\ket{i}$ encoded as a binary string on the $a$ ancillary qubits. 

One also defines a \emph{prepare} oracle (denoted as $\PREPARE$) as follows:
\begin{equation}
    \label{eq:prepare}
    \PREPARE \ket{0^{\otimes a}} = \frac{1}{\sqrt{c}}\sum_{i=0}^{K-1} \sqrt{c_i}\ket{i},
\end{equation}
where $c \equiv \sum_{i=1}^{K-1}|c_i|$. 
Given the definitions of these oracles, one can verify that the quantum circuit given in \cref{fig:LCU} implements the action of the operator $T$ on some input state $\ket{\psi} \in \mathcal{H}_s$ and thus prepares a BE of $T$. 
This result can be summarized with the following lemma (see for example \cite{Childs:2012gwh},\cite{Lin:2022vrd} for a proof of this result): 
\begin{lemma}
[LCU block encoding]
\label{LCU}
    Let $U_0 \ldots, U_{K-1}$ be $K$ unitary operators acting on a system Hilbert space ${\cal H}_s$ and let $T = \sum_{i=0}^{K-1} c_i U_i$, with $(c_0, \dots, c_{K-1}) \in \mathbb{C}^K$, be an operator also acting on $\mathcal{H}_s$. Then, the unitary given by $\left(\PREPARE^{\dagger} \otimes \unit_s\right)\SELECT\left(\PREPARE \otimes \unit_s\right)$, where $\unit_s \in \mathcal{H}_s$ is the identity operator, is a $(c, \log K)$ BE of the operator $T$, where $c \equiv \sum_{i=0}^{K-1} | c_i |$. The elementary gate complexity of constructing this BE asymptotically scales as $\OO(K \log K)$.
\end{lemma}

\begin{figure*}[t]
\centering
\begin{quantikz}
\lstick[4]{$\sket{0^{\otimes m}}_a$}
& \gate[4]{\PREPARE} & \octrl{1} \gategroup[5,steps=4,style={dashed,rounded corners, inner
sep=6pt}]{$\SELECT$} & \ctrl{1} & \qw \ \cdots \ & \ctrl{1} & \gate[4]{\PREPARE^\dagger} &
\\
& & \octrl{1} & \octrl{1} & \ \cdots \ & \ctrl{1} & &
\\
\vdots \setwiretype{n}& & \ \vdots \  & \vdots & \ \ddots \  \setwiretype{n} & \ \vdots \ & & \ \vdots \
\\
&  & \octrl{1} \wire[u][1]{q} & \octrl{1} \wire[u][1]{q} & \ \cdots \ & \ctrl{1} \wire[u][1]{q} & &
\\
\lstick{$\sket{\psi}_s$} & \qw \qwbundle{} & \gate{U_0} & \gate{U_1} & \ \cdots \ & 
\gate{U_{K-1}} & &
\end{quantikz}
\caption{The Linear Combination of Unitaries (LCU) circuit~\cite{Childs:2012gwh} to implement the BE of an operator ${T = \sum_{i = 0}^{K-1} c_iU_i}$, a linear combination of unitary operators with known implementations. Operators, $\PREPARE$ and $\SELECT$ are defined in \cref{eq:prepare} and \cref{eq:select}, correspondingly.
}
\label{fig:LCU}
\end{figure*}
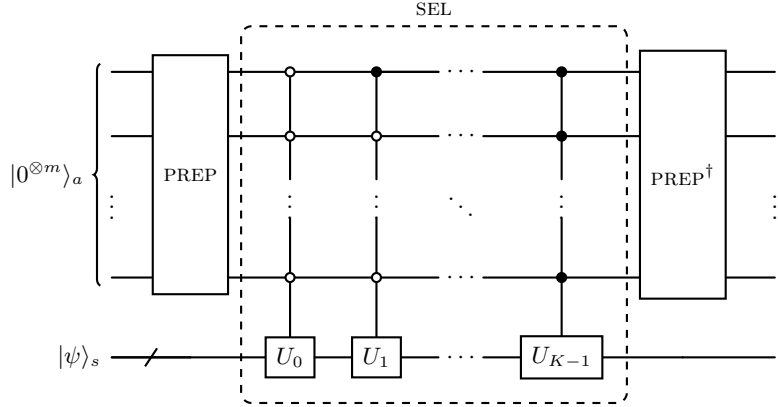

\subsection{Simulating time evolution using Generalized QSP \label{sec:gqsp}}

In this section, we review an improved version of the original QSP method, known as Generalized QSP (GQSP).
Provided access to some unitary $U$, QSP and GQSP allow one to construct circuits implementing the action of polynomials of $U$.
In the context of simulating time evolution, one can choose $U$ to be the Hamiltonian walk operator and use GQSP to implement polynomial approximations of the time evolution operator.

The main result of Ref.~\cite{Motlagh:2023oqc}, relevant to the problem of simulating time evolution, is given by:
\begin{theorem}[Corollary 5 from Ref.~\cite{Motlagh:2023oqc}]
    $\forall P \in \mathbb{C}[x]$ with deg(P) = d, if
    \begin{equation}
        \forall x \in \mathbb{T}, 
        \quad
        |P(x)|^2 \leq 1\,,
    \end{equation}
    where $\mathbb{T} = \{ x \in \mathbb{C}: |x| = 1 \}$ is the unit circle in the complex plane,
    then $\exists \hspace{1mm} \vec{\theta}, \vec{\phi} \in \mathbb{R}^{d+1}$ and $\gamma \in \mathbb{R}$ such that:
    \begin{equation}
        \left(\prod_{j=1}^{d} R(\theta_j, \phi_j, 0) \, C_U \right)\!R(\theta_0,\phi_0,\gamma) = \begin{pmatrix}
            P(U) & * \\
            * & *
        \end{pmatrix},
    \end{equation}
    where $R(\theta,\phi,\gamma)$ represents an SU(2) rotation and $C_U=(\ket{0}\bra{0} \otimes U) + (\ket{1}\bra{1} \otimes \unit)$ is a 0-controlled application of $U$ (the \emph{signal operator}),
\end{theorem}
Furthermore, Ref.~\cite{Motlagh:2023oqc} also provides an efficient classical algorithm to compute the required rotation angles $\vec{\theta}, \vec{\phi}, \gamma$ for a desired polynomial $P(x)$. 

\begin{figure*}[t]
\centering
\begin{quantikz}[row sep={1cm,between origins}, column sep=0.25cm]
 \lstick{$\sket{0}$}& \gate{R(\theta_0,\phi_0,\gamma)} & \octrl{1} & \gate{R(\theta_1,\phi_1,0)} & \octrl{1} & \ \ldots \ & \gate{R(\theta_{d-1},\phi_{d-1},0)} & \octrl{1} & \gate{R(\theta_d,\phi_d,0)} & \\
 \lstick{$\sket{0^{\otimes m}}_a$}& \phantomgate{R(\theta_0,\phi_0,\lambda)} & \gate[2]{W_H} & \phantomgate{R(\theta_1,\phi_1,0)} & \gate[2]{W_H} & \ \ldots \ & \phantomgate{R(\theta_{d-1},\phi_{d-1},0)} & \gate[2]{W_H} & \phantomgate{R(\theta_d,\phi_d,0)} & \\
 \lstick{$\sket{\psi}_s$} & & & & & \ \ldots \ & & & &
\end{quantikz}
\caption{Circuit diagram for simulating time evolution using GQSP, adapted from Ref.~\cite{Motlagh:2023oqc}. Each $R(\theta_0,\phi_0,\gamma)$ is an SU(2) rotation, with the individual angle parameters shown in the diagram chosen to construct a polynomial of $W_H$ in the context of this work.
}
\label{fig:GQSP}
\end{figure*}
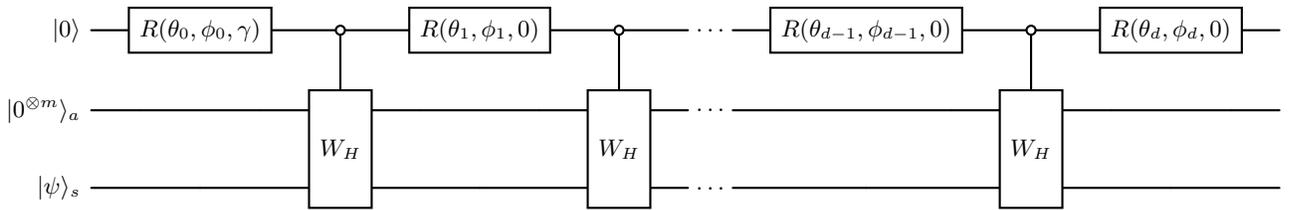

For the purpose of simulating time evolution, we note that the eigenvalues of the operator $W_H^{(\lambda)}$, defined by restricting the action of $W_H$ to the subspace spanned by $\operatorname{span}\{\sket{0}_a \sket{\lambda}_s, \sket{\bot^\lambda}_{as}\}$, are given by $\me^{\pm i\arccos(\lambda)}$. Thus we set the unitary to which controlled calls are made in the GQSP circuit to be $U = W_H$ (see~\cref{fig:GQSP} adapted from Ref.~\cite{Motlagh:2023oqc}) and seek an implementation of the function given by $\me^{- it\cos(x)}$ to achieve the desired polynomial transformation given by $P(W_H^{(\lambda)}) = \me^{-it\lambda}$. This is achieved through the Jacobi-Anger expansion (see for example Refs. \cite{Low:2016znh,Motlagh:2023oqc}), given by:

\begin{equation}
    \me^{-it\cos(x)} = \sum_{k = -\infty}^{\infty}(-i)^{k}J_k(t)\me^{ikx},
\end{equation}
where $J_k$ is the $k^{th}$ Bessel function of the first kind. It has been previously established the coefficients in this Fourier series decay exponentially and in particular, the number of SU(2) rotations required to achieve an $\epsilon$-close approximation to this polynomial transformation is asymptotically bounded by $\mathcal{O}(t + \log(1/\epsilon)/\log\log(1/\epsilon))$ (see Refs. \cite{Motlagh:2023oqc,Low:2016znh}).
We conclude by noting that, while using QSP for Hamiltonian simulation achieves the same asymptotically optimal scaling, using GQSP requires half as many controlled calls to $W_H$~\cite{Motlagh:2023oqc, Berry:2024ghc}.

\subsection{QSVT \label{ssec:qsvt}} 
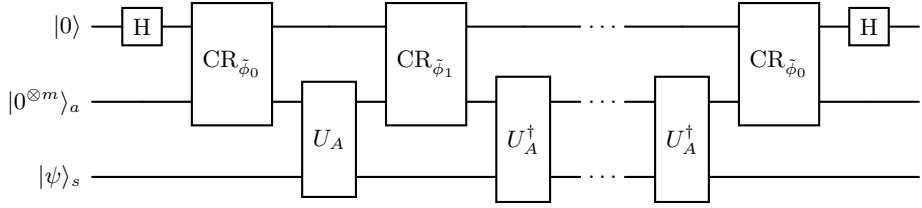
\begin{figure*}[t]
    \centering
    \begin{quantikz}[row sep={1cm,between origins}, column sep=0.4cm]
        \lstick{$\sket{0}$} & \gate{\Hgate} & \gate[2]{{\rm CR}_{\tilde{\phi}_0}} & & \gate[2]{{\rm CR}_{\tilde{\phi}_1}} & & \ \ldots\ & & \gate[2]{{\rm CR}_{\tilde{\phi}_0}} & \gate{\Hgate} &
        \\
        \lstick{$\sket{0^{\otimes m}}_a$} & & & \gate[2]{U_A} & & \gate[2]{U_A^\dagger} & \ \ldots\ & \gate[2]{U_A^\dagger} & & &
        \\
        \lstick{$\sket{\psi}_s$} & & & & & & \ \ldots\ & & & &
    \end{quantikz}
    \caption{Circuit diagram for QSVT, reproduced from Ref.~\cite{Lin:2022vrd}, using symmetric phases $(\tilde{\phi}_0, \tilde{\phi_1}, \dots, \tilde{\phi_1}, \tilde{\phi_0})$. The Hadamard gates serve to isolate the real part of the matrix function being implemented. The circuit for $\text{CR}_{\phi}$ is shown in~\cref{fig:crphi_circ}.}
    \label{fig:qsvt_circ}
\end{figure*}

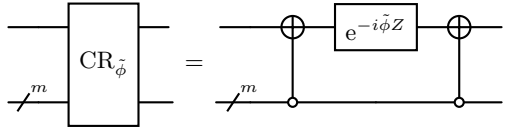
\begin{figure*}[t]
    \centering
    \begin{quantikz}[row sep={1cm,between origins}, column sep=0.4cm]
        && \gate[2]{{\rm CR}_{\tilde{\phi}}} &
        \midstick[2,brackets=none]{=}
    &
    &  \targ{} & \gate{\me^{-i \tilde{\phi} Z}} & \targ{} &
        \\
        & \qwbundle{m}
        & & & \qwbundle{m}
        & \octrl{-1} & & \octrl{-1} &
    \end{quantikz}
    \caption{Diagram for the $\text{CR}_{\phi}$ circuit appearing in the QSVT circuit in~\cref{fig:qsvt_circ}.}
    \label{fig:crphi_circ}
\end{figure*}

While QSP and GQSP are powerful algorithmic tools, in certain cases, the required polynomial transformations can be achieved using the simpler Quantum Singular Value Transformation (QSVT)~\cite{Gilyen:2018khw}.
The QSVT method is very similar to the QSP and GQSP methods in that it allows one to implement general matrix functions of some Hermitian operator $A$ by performing repeated calls to its BE $U_A$.
The major difference of QSVT has to do with the classes of matrix functions that can be implemented.
In particular, QSVT implements real functions with definite parity, whereas QSP implements complex functions with the real and imaginary parts having opposite parity.
Only requiring the real part of the matrix polynomial leads to several advantages.
First, the fundamental building block is now the BE $U_A$, rather than a controlled call to the walk operator $W_A$, as in QSP or GQSP.
Second, the conditions on implementable real polynomial are relaxed relative to the complex polynomials implemented using QSP.
Third, the phases that parameterize the circuit can be solved using a straightforward optimization procedure for large polynomials using e.g. QSPPACK~\cite{qsppack}.
The QSVT theorem for Hermitian matrices is as follows \cite{Lin:2022vrd}:

\begin{theorem}[QSVT for Hermitian matrices with real polynomials]
\label{thm:qsvt}
    Let $A \in \mathbb{C}^{2^n \times 2^n}$ be a normalized Hermitian operator with a $(1,m)$ block encoding $U_A$. Given a degree $d$ polynomial $P_{\mathrm{Re}}(x) \in \mathbb{R}[x]$ such that
    \begin{enumerate}
        \item  $P_{\mathrm{Re}}$ has parity $d \pmod{2}$,
        \item  $|P_{\mathrm{Re}}(x)| \leq 1 \hspace{2mm} \forall x \in [-1,1]$,
    \end{enumerate}
    then there exist symmetric phase factors $(\phi_0,\dots,\phi_{d}) \in \mathbb{R}^{d+1}$ such that the circuit presented in \cref{fig:qsvt_circ} implements a $(1,m+1)$ block-encoding of $P_{\mathrm{Re}}\left(A\right)$. 
\end{theorem}

One important property of QSVT circuits is that their controlled versions can be implemented with a small overhead~\cite{Gilyen:2018khw}.
For even polynomials, only the single-qubit $R_z$ gates in the $\text{CR}_{\tilde{\phi}_j}$ circuits (see Fig.~\ref{fig:crphi_circ}) acting on the signal qubit need to be controlled.
For odd polynomials, in addition to controlling the single-qubit $R_z$ gates, one must also control a single call to the BE $U_A$.
This property can be used to combine different QSVT circuits using LCU without introducing a large overall constant prefactor, which is important when using LCU to combine BEs of different local terms in the Hamiltonian constructed using QSVT, as explained in~\cref{ssec:qsvtbe}.

\subsection{QETU~\label{ssec:qetu}}

\begin{figure*}[t]
    \centering
    \begin{quantikz}[row sep={1cm,between origins}, column sep=0.25cm]
        \lstick{$\sket{0}$} & \qw & \gate{\me^{i \tilde{\varphi}_{0} X}} & \ctrl{1} & \gate{\me^{i \tilde{\varphi}_{1} X}} & \ctrl{1} & \qw & \ \ldots\ & \qw &  \ctrl{1} & \gate{\me^{i \tilde{\varphi}_{1} X}} & \ctrl{1} & \gate{\me^{i \tilde{\varphi}_{0} X}} &\qw &
        \\
        \lstick{$\sket{\psi}_s$} & \qw & \qw & \gate{U} & \qw & \gate{U^\dagger} & \qw & \ \ldots\ & \qw & \gate{U} & \qw & \gate{U^\dagger} & \qw & \qw & 
    \end{quantikz}
    \caption{Circuit diagram for QETU. The top qubit is the control qubit and the bottom register is the state that the matrix function is applied to.
    Here $U = \me^{-i \tau A}$ is the time evolution circuit.
    If the control qubit is measured to in the zero state, one prepares the normalized quantum state $F(\cos(\tau A/2))\sket{\psi}/||F(\cos(\tau A/2)) \sket{\psi}||$ for some even polynomial $F(x)$.
    For symmetric phase factors $(\tilde{\varphi}_0, \tilde{\varphi}_1, \dots, \tilde{\varphi}_1, \tilde{\varphi}_0) \in \mathds{R}^{d+1}$, the function $F(x)$ is a real even polynomial of degree $d$. 
    The probability of measuring the control qubit in the zero state is $p = ||F(\cos(\tau A/2)) \sket{\psi}||^2$.}
    \label{fig:qetu_circ}
\end{figure*}
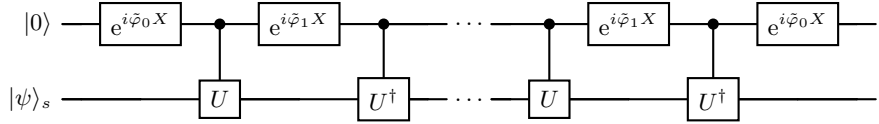

QETU is similar to the QSVT method described in~\cref{ssec:qsvt} in that it allows one to implement a broad class of matrix functions of some Hermitian operator $A$ by performing repeated calls to some fundamental building block. 
The major difference of QETU is that, instead of performing repeated calls to a BE $U_A$ as with QSVT, the building block for QETU is the time-evolution operator $\me^{-i \tau A}$ for some value of~$\tau$.
In a broad sense, QETU replaces the BE $U_A$ in the QSVT circuit in \cref{fig:qsvt_circ} with the controlled time-evolution operator $\me^{-i \tau A}$. 
The QETU circuit for implementing even functions is shown in \cref{fig:qetu_circ}. 
Taken from Ref.~\cite{Dong:2022mmq}, the QETU theorem is:

\begin{theorem}[QETU]\label{thm:qetu}
Let $U = \me^{-i \tau A}$ be a unitary operator, where $A$ is an $n$-qubit Hermitian operator and $\tau$ is real parameter. 
For any real even polynomial $F(x)$ of degree $d$ satisfying $|F(x)| \leq 1, \forall x \in [-1,1]$, one can find a sequence of symmetric phase factors $(\varphi_0, \dots, \varphi_d) \in \mathbb{R}^{d+1}$ such that the circuit in~\cref{fig:qetu_circ} denoted by $\mathcal{U}$ satisfies $(\sbra{0} \otimes \unit) \mathcal{U} (\sket{0} \otimes \unit) = F(\cos(\tau A / 2))$. 
\end{theorem}

In practice, the QETU circuit is used for approximately implementing $\fex(A)$ by realizing a transformation $F(\cos (\tau A/2))$, where $F(x)$ is a polynomial approximation to $\Fex(x)\equiv \fex(\frac{2}{\tau}\arccos(x))$.

Similar to QSVT, the phases $\{\tilde{\varphi}_j \}$ appearing in the circuit in~\cref{fig:qetu_circ} can be calculated using QSPPACK~\cite{qsppack}.
As described in Appendix B of Ref.~\cite{Dong:2022mmq}, the phases in the QETU circuit $\{\tilde{\varphi}_j \}$ are related to the phases calculated using QSPPACK $\{\tilde{\phi}_j\}$ by $\tilde{\varphi}_j = \tilde{\phi}_j + (2-\delta_{j\, 0}) \pi/4$.

Completely analogous to the case with QSVT circuits, controlled calls to QETU circuits can also be implemented for a small additional overhead.
This property can be used to combine BEs of different local terms in the Hamiltonian that were constructed using QETU; the procedure for constructing BEs using QETU is given in~\cref{ssec:qetube}.

\section{Block encodings for a bosonic lattice field theory\label{sec:be}}

In this section we describe how to construct block encodings for certain types of bosonic lattice field theories, and how to use those to simulate time evolution using GQSP.
Our choice of GQSP is guided by the fact that, while this algorithm has the same asymptotic complexity as QSP or QSVT, but it comes with a smaller constant prefactor~\cite{Motlagh:2023oqc, Berry:2024ghc}.

The general strategy for using GQSP to simulate time evolution of the Hamiltonian in Eq.~\eqref{eq:GeneralH} is achieved with the following steps:
\begin{enumerate}
    \item Construct the BEs $U_{f_k(\hat{\xi}_k)}$ for the individual terms of the form $f_k(\hat{\xi}_{k})$.\footnote{Since the functional form of $f_{k_{\ell}}$ is the same $\forall \ell \in \{0,1,\dots,N-1\}$ and fixed $k \in \{0,1,2\}$, in what follows we will drop the $\ell$ subscript on $f_{k}$ when the statement is true / applies $\forall \ell$.}
    \item Construct $U_{f_1(\hat\pi_i)} = \FT_i^\dagger U_{f_1\big(\hat\pi^{(D)}_i\big)} \FT_i$ using the Quantum Fourier transform circuit.
    \item Use LCU to combine these site-local BEs to obtain the BE for the full Hamiltonian $U_H$.
    \item Construct the qubitized walk operator $W_H$.
    \item Simulate time evolution using the GQSP circuit in~\cref{fig:GQSP}.
\end{enumerate}

After introducing the Hamiltonian of the scalar bosonic lattice field theory and discuss the digitization of the local bosonic Hilbert space, we present three different techniques for preparing BEs for a broad class of operators appearing in site-local Hamiltonians encountered in various formulations of lattice field theories.
In all three techniques the BE of $f_k(\hat \xi_k)$ will start with an encoding of a basic building block, which is then combined to form $f_k(\hat \xi_k)$. 

The first techniques uses QSVT, and starts from the simplest building block possible, namely the operator $\hat \xi_k$. Since this operator is not unitary, it needs to be block-encoded. 
As we will see, this encoding can be done with resources that are linear in $\nq$, but comes with a large prefactor.
Constructing the final function $f_k(\hat \xi_k)$ adds a constant multiplicative factor on the resource requirement.

The second method uses QETU, and directly uses the unitary operator $\me^{-i \tau \hat \xi_k}$ as the basic building block. The building block of QETU can still be implemented with resources that are linear in $\nq$, but removes the large pre-factors arising in the previous method. 
The downside is that constructing  $f_k(\hat \xi_k)$ introduces an exponential dependence on $\nq$. As we will see, despite the exponential dependence on $\nq$, the absence of the large pre-factor make this method more efficient for $\nq \lesssim 6$.

The final method takes a different approach, and chooses a more complicated building block, which requires resources with exponential scaling in $\nq$ for its construction. 
This building block is chosen such to make the construction of $f_k(\hat \xi_k)$ as easy as possible, and is in fact trivial. This gives better scaling than the second method, and in fact outpeforms all methods $\nq \lesssim 11$. 
We call this final method Linear Operators Via Exponentiation and Linear Combination of Unitaries (\LOVELCU).

The comparison of different BEs is summarized in~\cref{table:local_BEs}.

\begin{table*}
\centering
\begin{tabular}{llccc}
% \toprule
Algorithm & Gate Complexity & Ancillary qubits & Indefinite parity & Multivariable functions \\
\hline
Standard LCU & $\OO(\nq^d)$ & $\OO(d \log \nq)$ & Yes & Yes \\
QSVT & $\OO(d\cdot\nq\log\nq)$ & $\OO(\log \nq)$ & No & No \\
QETU with $\me^{i \tau \hat\xi}$ & $\OO(2^{\nq}\cdot\nq)$ & 1 & No & No \\
QETU with $\me^{i 2 \arccos \hat\xi}$ & $\OO(d\cdot 2^{\nq})$ & 1 & No & No \\
\LOVELCU  & $\OO(2^{\nq})$ & 1 & Yes & Yes \\
% \bottomrule
\end{tabular}
\caption{Properties of BE methods considered in this work. Gate count and ancillary qubit count are for constructing a BE of a degree $d$ function $f(\hat{\xi})$ of an $n$-qubit diagonal Hermitian operator $\hat{\xi}$ that is a sum of $\OO(n)$ single $\Zgate$-gates (for operators considered in this work, $\hat{\xi} \in \{\hat\varphi_i, \hat\pi^{(D)}, \hat\varphi_i-\hat\varphi_j \}$). The QSVT- and QETU-based methods require $f$ to have definite parity and be a function of a single variable.
Standard LCU and \LOVELCU can be applied to functions with indefinite parity.
}
\label{table:local_BEs}
\end{table*}

\subsection{General Form of Hamiltonians of Interest \label{ssec:general_H}}

A common approach to formulating quantum field theories in a way suitable for solving those on quantum computers amounts to 1) discretizing the real space so that the fields were defined on a finite subset $\mathcal{V}$ of sites of some lattice, e.g., $\mathcal{V}\subset a\mathbb{Z}^d\subset\mathbb{R}^d$, subject to an appropriate boundary condition; 2) digitizing quantum degrees of freedom residing on lattice sites, i.e., truncating the infinite-dimensional local Hilbert space to a finite-dimensional one.
While alternative approaches are also subject to active investigation~\cite{Brower:1997ha,Kreshchuk:2020dla,Buser:2020cvn}, the aforementioned paradigm has two major advantages: the locality of the discretized Hamiltonian (for local quantum field theories) and linear growth of the Hamiltonian norm with problem size.\footnote{Here we do not touch upon the issue of renormalization which may affect the Hamiltonian norm via adjustment of coupling constants.}
Both of these are of great utility from the quantum computing perspective.
While fermionic lattice systems have been subject to active investigation in the literature (for recent studies, see Ref.~\cite{Setia:2018qmu,Farrell:2024fit} and references therein), in this work we focus on lattice bosons, which are an essential ingredient in nuclear~\cite{Savage:2023qop} and high-energy particle physics~\cite{Bauer:2022hpo}.

Many physically relevant Hamiltonians, including scalar field theories~\cite{Jordan:2011ci, Klco:2018zqz}, dual basis formulations of U(1)~\cite{Bender:2020ztu, Bauer:2021gek} and SU(2)~\cite{DAndrea:2023qnr} gauge theories, are (at least partially) formulated in terms of conjugate operators defined on each lattice site.
For this work we use the notation of a scalar field theory where we denote by $\hat{\varphi}_i$ and $\hat{\pi}_i$ the ``field'' and ``momentum'' operators at site $i$, respectively, satisfying the canonical commutation relations $[\hat{\varphi}_i, \hat{\pi}_j] = i \delta_{ij}$.
The general idea is to exploit the fact that, for a particular highly efficient digitization scheme of such operators (described in detail below), the local operators $\hat{\varphi}_i$ and $\hat{\pi}_i$ take on simple forms and can therefore be used as low cost building blocks in the QSVT or QETU circuits. These can then be used to prepare low cost BEs of the individual terms appearing in the Hamiltonian; and the BE for the full Hamiltonian can be prepared by using LCU to combine the BEs of the site-local terms.

To be concrete, we consider Hamiltonians of the form
\begin{equation}
\label{eq:general_H}
    \hat{H} = \sum_{i=1}^{\NLat} \left( f_0\left(\hat{\varphi}_i\right) + f_1\left(\hat{\pi}_i\right) \right) + \sum_{\langle ij \rangle} f_2\left(\hat{\varphi}_i - \hat{\varphi}_j \right) ,
\end{equation}
where $\NLat$ is the number of lattice sites and $\sum_{\langle i j \rangle}$ is a sum over all adjacent sites $i$ and $j$.
A discussion of applying our methods to more general forms of the Hamiltonian is given in~\cref{sec:discussion}.

To \emph{digitize} the theory, the local bosonic Hilbert space at each lattice site $j$ is represented using $\nq$ qubits.
Following Refs.~\cite{Jordan:2011ci, Klco:2018zqz,Bauer:2021gup}, we work in the eigenbasis of the digitized field operators $\hat{\varphi}_j$, i.e., we identify the computational basis states $\sket{k}_j$ at site $j$ with the eigenstates $\sket{\varphi_j^{(k)}}$.
The eigenvalues of the field operator $\hat{\varphi}_j$ are chosen as
\begin{equation}
\begin{gathered}
\label{eq:phi_op}
\hat{\varphi}_j \sket{\varphi_j^{(k)}} = (-\varphi_\text{max} + k \delta \varphi) \sket{\varphi_j^{(k)}}\,, \\ k=0, \dots, 2^{\nq}-1\,,    
\end{gathered}
\end{equation}
where $\delta \varphi = 2\varphi_{\mathrm{max}}/(2^{\nq}-1)$.
The free parameter $\varphi_\text{max}$ can be chosen to minimize the digitization errors and is theory dependent. Generally, by choosing $\varphi_\text{max}$ to scale exponentially with $\nq$~\cite{Jordan:2011ci, Klco:2018zqz, Bauer:2021gup, Bauer:2021gek, Kane:2022ejm, DAndrea:2023qnr}, the digitization errors decrease exponentially with $\nq$.
This fact implies that the number of qubits per site $\nq$ can generally be kept small, which is important to keep in mind for the discussion moving forward.

Exploiting the fact that $\hat{\pi}_j$ is conjugate to $\hat{\varphi}_j$, the momentum operators are digitized as
\begin{equation}
\begin{gathered}
    \hat{\pi}^{(D)}_j \sket{\pi_j^{(k)}} = (-\pi_\text{max} + k \delta \pi) \sket{\pi_j^{(k)}}\,,
    \\
    k=0, \dots, 2^{\nq}-1\,,
\end{gathered}
\end{equation}
where
\begin{equation}
    \hat{\pi}^{(D)}_j = \FT_j \hat{\pi}^{\vphantom{\dagger}}_j \FT_j^\dagger
\end{equation}
is the conjugate momentum operator in the momentum basis, $\FT_j$ denotes the usual discrete Fourier transform (FT) at site $j$, $\pi_\text{max} = \pi/\delta \varphi$ and $\delta \pi = 2\pi_\text{max}/(2^{\nq}-1)$.
The Hamiltonian in this digitizaiton is then written through functions of the diagonal operators $\hat{\varphi}_j$ and $\hat{\pi}_j^{(D)}$ as 
\begin{equation}
\label{eq:digitized_H}
\begin{alignedat}{9}
    \hat{H} = \sum_{i=1}^{\NLat} \Bigl( f_0\left(\hat{\varphi}_i\right) &+ \FT^\dagger_i f_1(\hat{\pi}^{(D)}_i) \FT_i \Bigr) 
    \\
    &+ \sum_{\langle ij \rangle} f_2\left(\hat{\varphi}_i - \hat{\varphi}_j \right).
\end{alignedat}
\end{equation}

\subsection{Implementation of field operators and their functions}

The operators $\hat{\varphi}_j$ and $\hat{\pi}_j^{(D)}$ are $2^{\nq}$-dimensional diagonal operators with evenly spaced entries on the diagonal and hence can be expressed as a linear combination of $\nq$ single-qubit Pauli-Z gates~\cite{Welch_2014, Klco:2018zqz}.
For the field operator $\hat \varphi$ for example (dropping the site indices for brevity), one can write
\begin{equation}
\label{eq:pauli_decomp_phi}
    \hat{\varphi} = -\frac{\varphi_{\rm max}}{2^{\nq}-1} \sum_{m=0}^{\nq-1} 2^{\nq-1-m} Z_m\,,
\end{equation}
where $Z_m$ is a Pauli-Z gate acting on the $m$\textsuperscript{th} qubit of the local Hilbert space considered.
Aside from the overall prefactor, the operator $\hat{\pi}^{(D)}$ has an identical gate representation.
Given that the operators $\hat{\xi}_k$ are diagonal, one can readily determine complicated functions $f_k(\hat{\xi}_k)$ using classical resources by simply applying the function $f_k$ to the diagonal elements of $\hat{\xi}_k$. 

\subsection{LCU to Block Encode Individual Terms}
In this section, we review the cost of using standard LCU techniques to BE the individual $f_k(\hat \xi_k)$ terms.
For a more detailed analysis, see Ref.~\cite{Hariprakash:2023tla} and references therein.

% Adding $M$ terms using LCU requires $\lceil \log_2 M \rceil$ ancillary qubits.
% Using general state preparation techniques, the $\PREPARE$ oracle is implemented using $\OO(M)$ gates.
% The $\SELECT$ oracle can be implemented using $\OO(M \log M)$ gates.
% Taken together, LCU can add $M$ terms using $\OO(\log M)$ ancillary qubits and $\OO(M \log M)$ gates.
% The scale factor when using LCU is the sum of the absolute value of the 

Consider first using LCU to BE $f_0(\hat \varphi)$ and $f_1(\hat\pi^{(D)})$
From \cref{eq:pauli_decomp_phi}, we see that a general degree $d$ function of $\hat \varphi$ (or $\hat\pi^{(D)}$) is a sum of $\OO(\nq^d)$ Pauli strings.
It is important to note that the number of Pauli strings is upper bounded by $2^{\nq}$.
This, combined with \cref{LCU}, implies that the asymptotic gate cost to BE $f_0(\hat \varphi)$ and $f_1(\hat\pi^{(D)})$ is either $\OO(\nq^d \log \nq)$ or $\OO(\nq 2^{\nq})$, depending on the degree $d$ and value of $\nq$ used.

Turning to the final term, for a degree $d$ function, $f_2(\hat \varphi_i - \hat \varphi_j)$ is a sum of $\OO(\nq^d)$ terms.
In this case, the number of Pauli strings is upper bounded by $2^{2\nq}$. 
Using similar arguments, the asymptotic gate cost to BE $f_2(\hat \varphi_i-\hat\varphi_j)$ is either $\OO(\nq^d \log \nq)$ or $\OO(\nq 2^{2\nq})$.

Lastly, we discuss the scale factor when using LCU.
The scale factor is the sum of the absolute values of the coefficients of the unitary terms being added together.
Given a particular $f_k(\hat \varphi_k)$, the scale factor can be readily determined using using \cref{eq:pauli_decomp_phi}.
The scale factors for operators considered in Sec.~\ref{sec:scalar_field_theory_numerics} are derived in Appendix~\ref{app:lcu_scale_factor}.

\subsection{QSVT to Block Encode Individual Terms \label{ssec:qsvtbe}}

In the context of quantum simulation, the QSVT algorithm is typically utilized as a way of constructing polynomial functions of Hamiltonian operator~\cite{Gilyen:2018khw,Lin:2022vrd}.
In this section we use it to construct the functions $f_k(\hat{\xi}_k)$ and demonstrate that QSVT can also be used as a highly-efficient tool for constructing block encodings.
The high level flow of this technique is depicted in Fig.~\ref{fig:alg_outline}.
% \MK{OK.}
For this discussion we assume the functions $f_k$ to have even parity.

The building block in the QSVT circuit is the BE of the simple operator $\hat{\xi}_k$. 
From Eq.~\eqref{eq:pauli_decomp_phi} we see that each $\hat{\xi}_k$ is written as a sum of $\OO(\nq)$ Pauli-Z gates, each of which are unitary. 
The BE $U_{\hat{\xi}_k}$ can therefore be constructed with LCU using $\OO(\nq \log \nq)$ gates and $\OO(\log \nq)$ ancillary qubits.

If $f_k$ is a polynomial function, QSVT can prepare the BE $U_{f_k}$ using a constant number of calls to $U_{\hat{\xi}_k}$. 
For non-polynomial functions, one can express the function $f_k(\hat{\xi}_k)$ as a (potentially infinite) Chebyshev series of the operator $\hat{\xi}_k$, and then use QSVT to construct the BEs $U_{f_k}$. 

To determine the Chebyshev expansion, we write
\begin{equation}
    \sum_{j=0} c_{2j}  T_{2j}\!\left(\frac{\hat{\xi}_k}{\alpha}\right) = \frac{1}{\beta} f_k(\hat{\xi}_k)\,,
\end{equation}
where $\alpha$ is the scale factor used in the BE of $\hat{\xi}_k$ and $\beta$ normalizes the function $f_k(\hat{\xi}_k)$ which can be obtained classically.  
The upper limit of the sum has been left blank to stress the series can be a sum of a finite or infinite number of terms.
The coefficients $c_{2j}$ can be then be found using the orthogonality relations of Chebyshev polynomials.
The gate complexity depends on the convergence of the Chebyshev polynomial approximation, and for a degree-$d$ polynomial the gate complexity is therefore given by using $\OO(d \, \nq \log \nq)$.
The scale factor of the BE is $\beta = \|f_k(\hat{\xi}_k)\|$.

We conclude by pointing out that using $\hat{\varphi}_i$ and $\hat{\pi}_i$ as building blocks to construct a BE of the Hamiltonian is not a unique choice.
One could equivalently use, e.g., $\hat{\varphi}_i^2$ and $\hat{\pi}_i^2$, and modify the Chebyshev polynomial used accordingly.
However, because the number of terms, as well as the length of the largest Pauli string, in $\hat{\varphi}_i^n$ (and $\hat{\pi}_i^n$) scales as $\OO(\nq^n)$, using the simpler building blocks $\hat{\varphi}_i$ and $\hat{\pi}_i$ will result in better asymptotic scaling with $\nq$.
We leave investigations of alternative constructions for future work.

\subsection{QETU to Block Encode Individual Terms \label{ssec:qetube}}

While QETU was initially proposed to eliminate the need for block encoding in early fault-tolerant quantum algorithms~\cite{Dong:2022mmq}, Ref.~\cite{Kane:2023jdo} has investigated its utility for preparing functions of Hermitian operators, using Gaussian states as an example.
% Below, we show that QETU can be utilized for constructing block encodings for Hamiltonians of the form~\eqref{eq:general_H}.
Below, we show that QETU can be utilized for constructing block encodings for Hamiltonians of the form~\eqref{eq:general_H}.

Preparing BEs using QETU is conceptually similar to using QSVT. 
Following the procedure outlined in Fig.~\ref{fig:alg_outline},
the building block used is $\me^{-i \tau \hat{\xi}_k}$, which is the simplest unitary operator containing information about $\hat \xi_k$. 
By performing repeated calls to $\me^{-i \tau \hat{\xi}_k}$, QETU can prepare a BE of $f_k(\hat{\xi}_k)$.

The remainder of this section is split into two parts.
We first describe the technical details of using QETU to prepare general matrix functions.
From there, we discuss the computational cost of preparing BEs of $f_k(\hat{\xi}_k)$.
% \MK{OK.}

% {\color{magenta} Preparing BEs using QETU is conceptually similar to using QSVT, and is outlined in Fig.~\ref{fig:alg_outline}.
% The building block in this case is now $\me^{-i \tau \hat{\xi}_k}$, which is the simplest unitary operator containing information about $\hat \xi_k$. 
% By performing repeated calls to $\me^{-i \tau \hat{\xi}_k}$, QETU can prepare a BE of $f_k(\hat{\xi}_k)$. 
% However, due to technical details of the QETU circuit, doing so introduces an exponential dependence on $\nq$ in the cost. 

% Because $\me^{-i \tau \hat{\xi}_k}$ can be implemented using $\nq$ $R_z$ gates, 
% At a high level, QETU takes a unitary matrix $\me^{-i \hat{A}}$ as an input and produces a BE of $f(\hat{A})$. 
% The main idea is to exploit the simplicity of the operator $\me^{-i \tau $

% motivation is the observation that the simple operator $\me^{-i t \hat{\xi}}$ can be used }

\subsubsection{Technical details}
We begin by discussing QETU for a general function $f(\hat{A})$ of a Hermitian matrix $\hat{A}$ (not necessarily diagonal) with known eigenvalues $\hat{A} \sket{a_j} = a_j \sket{a_j}$, and will specialize to $f_k(\hat{\xi}_k)$ later.
We will call the minimum and maximum eigenvalues $a_{\rm min}$ and $a_{\rm max}$, respectively.
QETU allows to start from an implementation of the operator $\me^{-i \tau \hat A}$ and implement a BE of the function $F(\cos(\tau \hat A / 2))$. 
Since one needs to shift the spectrum of the operator %$\hat A$ to lie in the range $[0, \pi]$,
$\tau \hat{A}$ to lie in the range $[0, 2\pi]$,
we therefore define a shifted operator
\begin{equation}
\label{eq:shifted_A_qetu}
    \hat{A}_\text{sh} = c_1 \hat{A} + c_2\unit\,,
\end{equation}
where
\begin{equation}
    c_1 = \frac{\pi}{a_\text{max} - a_\text{min}}\,, \quad c_2 = - c_1 a_\text{min}\,,
\end{equation}
and require $\tau \leq 2$.
In order to implement the desired function $f_k(\hat A) / \beta$ (where $\beta$ is required to normalize the function to allow a BE), one needs to choose
\begin{equation}
\label{eq:Fexdef}
    \Fex(x) = \frac{1}{\beta} \fex\!\left(\frac{\frac{2}{\tau}\arccos(x) - c_2}{c_1}\right)
    .
\end{equation}

Note that the direct implementation of QETU applies only in cases when $\Fex(x)$ has definite parity.
Since the function in~\cref{eq:Fexdef} does not generally have definite parity, one would need LCU to combine even and odd functions together to obtain the general function $F(\hat A)$. 
However, as shown in Ref.~\cite{Kane:2023jdo}, the need for LCU to add the even and odd pieces can, in certain cases, be avoided altogether by choosing by choosing $\tau$ in such a way as to make $\Fex(x)$ have definite parity.
To see this, note that if one chooses $\tau = \pi/c_2$, the function $\Fex(x)$ becomes
\begin{equation}
\label{eq:Fx_qetu_arcsin}
\begin{split}
    \Fex(x)\Big|_{\tau = \pi/c_2} &= \frac{1}{\beta}\fex\!\left(-\frac{2}{\pi} \frac{c_2}{c_1} \arcsin(x) \right)
    \\
    &= \frac{1}{\beta}\fex\!\left(\frac{2 a_{\rm min}}{\pi}\arcsin(x) \right)
\end{split}
\end{equation}
where we have used $\arcsin(x) = \pi/2 - \arccos(x)$.
Because $\arcsin(x)$ has odd parity, the function $\Fex(x)$ has the same parity as the function $\fex(x)$.
This means that LCU is not required if the original function $f_k(\hat A)$ one aims to implement has definite parity. 
Note if $\tau > 2$, this technique cannot be applied and one must construct the even and odd pieces separately.
Importantly, all operators $\hat{\xi}_k$ considered in this work are digitized such that their minimum and maximum eigenvalues are equal, which implies $\tau = 2$.

Next, we consider the scale factor $\beta$, which is given by $\beta = \max_{x\in [-1,1]} |F(x)|$. 
Using the fact that $\arcsin(x) \in [-\frac{\pi}{2}, \frac{\pi}{2}]$, we can rewrite the scale factor as 
\begin{equation}
\label{eq:qetu_scale_factor}
    \beta = \max_{y \in [-a_{\rm min}, a_{\rm min}]} |f(y)|  \,.  
\end{equation}
We will use this to determine the scale factor for the BEs of interest in the next section.

The final observation we make regarding $\Fex(x)$ is that, for general functions $\fex(x)$, its first derivative diverges at $x = \pm 1$ due to the presence of the $\arcsin(x)$ (or $\arccos(x)$) term.
Chebyshev approximations of functions with $m$ continuous derivatives on $x \in [-1,1]$ are known to converge polynomially to the true function with the typical rate $(1/n_\text{Ch})^m$, where $n_\text{Ch}$ is the number of Chebyshev polynomials used in the approximation~\cite{trefethen_apprx_theory}.
This is to be contrasted with the exponential convergence of Chebyshev series to infinitely differentiable functions.
If one needed to reproduce $\Fex(x)$ for all $x \in [-1,1]$, this poor convergence could lead to a large gate cost.

As was shown in Ref.~\cite{Kane:2023jdo} this issue can be overcome by exploiting the fact that the spectrum of the operator $\hat{A}_\text{sh}$ is known exactly 
and therefore one only has to reproduce $\Fex(x)$ at those finitely many points. 
If $\hat{A}_\text{sh}$ is represented using $\nq$ qubits, there are $\OO(2^{\nq})$ such values, and one can therefore reproduce the function at those points exactly by using $\OO(2^{\nq})$ Chebyshev polynomials.
In this way, one avoids the poor polynomial convergence of the Chebyshev expansion, at the cost of introducing an exponential scaling with $\nq$.
This exponential scaling, however, does not pose a problem as long we work with small values of $\nq$ and as previously discussed, this is precisely the case in lattice field theories owing to the fact that the digitization errors generally decrease exponentially with $\nq$.

\subsubsection{Cost to prepare block-encodings of bosonic operators}

We now describe the cost of using QETU to prepare BEs of $f_k(\hat{\xi}_k)$, where, again, each $f_k$ is assumed to have even parity.
After constructing the shifted operator $\hat{\xi}_{k, {\rm sh}}$ according to Eq.~\eqref{eq:shifted_A_qetu}, the building block in the QETU circuit is a simple operator $e^{-i \tau \hat{\xi}_{k,{\rm sh}}}$.
Because each $\hat{\xi}_{k, {\rm sh}}$ is still a diagonal operator with evenly spaced eigenvalues (or a Kronecker sum of two such operators in the case of $\hat{\xi}_2 \in \{\hat\varphi_i-\hat\varphi_j\}$), they can be written as a sum of $\OO(\nq)$ single-qubit Pauli-Z gates, and $e^{-i \tau \hat{\xi}_{k,{\rm sh}}}$ can be implemented using only $\OO(\nq)$ single-qubit $R_z$ gates.
A single controlled call to $e^{-i \tau \hat{\xi}_{k,{\rm sh}}}$ therefore requires $\OO(\nq)$ of $R_z$ and CNOT gates.
As discussed, QETU can prepare the BE $U_{f_k}$ exactly using a polynomial of degree equal to the number of times we sample the function $f_k$.

The operators $\hat{\varphi}_i$ and $\hat{\pi}_i^{(D)}$ are sampled $2^{\nq}$ times, implying one must use a degree-$\OO(2^{\nq})$ polynomial to reproduce $\Fex_0(\hat{\varphi}_i) \sim \fex_0(\arcsin(\hat{\varphi}_i))$ and $\Fex_1(\hat{\pi}^{(D)}_i) \sim \fex_1(\arcsin(\hat{\pi}^{(D)}_i))$ exactly.
Note that one can reduce the degree of the polynomial by a factor of two by exploiting the fact that the operators are sampled symmetrically about zero and that $f_0$ and $f_1$ are even functions.
Therefore, the asymptotic gate cost of preparing the BEs $U_{f_0}$ and $U_{f_1}$ using QETU is $\OO(\nq 2^{\nq})$, independent of the degree $d$ of the polynomial of $f_0$ and $f_1$.
To determine the scale factor, note that the minimum eigenvalue of $\hat{\varphi}$ ($\hat{\pi}$) is $-\varphi_{\rm max}$ ($-\pi_{\rm max}$).
Using Eq.~\eqref{eq:qetu_scale_factor}, the scale factor for $U_{f_0}$ ($U_{f_1}$) is therefore $\max_{x \in [-\varphi_{\rm max}, \varphi_{\rm max}]} |f_0(x)|$ ($\max_{x \in [-\pi_{\rm max}, \pi_{\rm max}]} |f_1(x)|$.
For many systems of physical interest, including scalar field theories (considered in Sec.~\ref{sec:scalar_field_theory_numerics}), the maximum value of $f_0, f_1$ occur at the largest argument they are evaluated.
In this situation, the scale factors for $U_{f_0}$ and $U_{f_1}$ are given by the smallest possible values, namely $\|f_0(\hat{\varphi})\|$ and $\|f_1(\hat{\varphi})\|$, respectively.

The final operator we need to block-encode is $\fex_2(\hat{\varphi}_i - \hat{\varphi}_j)$. 
Because $(\hat{\varphi}_j-\hat{\varphi}_i)$ is a $2^{2\nq}$-dimensional matrix, one would naively expect that $\OO(2^{2 \nq})$ Chebyshev polynomials are needed to exactly reproduce $\Fex_2(\hat{\varphi}_j-\hat{\varphi}_i)$ at all the $2^{2\nq}$ values it is sampled.
However, due to the symmetric digitization of $\hat{\varphi}$, the matrix $(\hat{\varphi}_j-\hat{\varphi}_i)$ only has $\OO(2^{\nq})$ unique values, and can be implemented exactly using $\OO(2^{\nq})$ Chebyshev polynomials. 
The asymptotic gate complexity of preparing the BE $U_{f_2}$ using QETU is therefore $\OO(\nq 2^{\nq})$, independent of the degree $d$ of the polynomial of $f_2$.
Because the minimum eigenvalue of $(\hat{\varphi}_j-\hat{\varphi}_i)$ is $-2\varphi_{\rm max}$, the scale factor is $\max_{x \in [-2\varphi_{\rm max}, 2\varphi_{\rm max}]}$.
As before, for many physical systems, the scale factor is given by the minimum value of $\|f_2(\hat \varphi_i-\hat \varphi_j)\|$.

Additional gate cost savings can be achieved if one uses the control-free version of QETU~\cite{Dong:2022mmq}.
In practice, doing so reduces the number of $R_z$ gates in the controlled call to $e^{-i \tau \hat{\xi}_{k, {\rm sh}}}$ by a factor of 2.
A detailed procedure for implementing the control-free version of QETU for diagonal operators can be found in Ref.~\cite{Kane:2023jdo}.

\subsection{\LOVELCU\label{ssec:love}}

In the previous section, we showed that, by employing $\me^{i \tau \hat{\xi}_k}$ as a building block, QETU can be used to construct BEs of $f_k(\hat{\xi}_k)$ with an asymptotic gate cost $\OO(\nq 2^{\nq})$.
In this section, we study how the cost changes if different building blocks are utilized.
First, we will argue that, in general, it is not possible to eliminate the exponential scaling with $\nq$.
With this fundamental limitation in mind, we show that using more complicated building blocks can reduce the asymptotic cost to prepare a BE.
Specifically, the cost of preparing $f(\hat{\xi}_k)$ is reduced from $\OO(\nq 2^{\nq})$ to $\OO(d2^{\nq})$, where $d$ is the degree of $f$.
Finally, extending this logic further, we develop a conceptually simple framework that does not require QETU, and can prepare BEs of $f(\hat{\xi}_k)$ using $\OO(2^{\nq})$ gates, independent of the degree $d$ of $f$.
% \MK{Made minor adjustments, plz check.}

% \begin{itemize}
%     \item In this section, we ask if different building blocks can improv the scaling
%     \item Start by generalizing building block
%     \item Exponential sclaing persists if $g$ is finite polynomial
%     \item Accept exponential, consider $g$ non-analytic.
% \end{itemize}

% In the previous section, we showed that, using $\me^{i \tau \hat{\xi}_k}$ as a building block, QETU can be used to construct BEs of $f_k(\hat{\xi}_k)$ with an asymptotic gate cost $\OO(\nq 2^{\nq})$.
% The exponential scaling in $\nq$ originated from having to work with functions containing $\arccos(x)$ (or $\arcsin(x)$) terms, requiring high-order Chebyshev polynomials.

\subsubsection{Main idea}

To begin the discussion, recall that the exponential scaling in $\nq$ arises from the need to work with functions whose derivatives diverge, such as $\arccos(x)$ or $\arcsin(x)$.
% As a result, any building block involving these functions will exhibit the same exponential scaling.
This result generalizes, and any building block that requires working with these functions will exhibit the same exponential scaling.
% \CFK{MK: please see new sentence. First sentence was not clear to me if we meant including arccos in building block, or if building block resulted in needing $F(x)\sim\arccos(x)$.}

To understand what functions fall into this class, consider the general building block $\me^{-i \tau g(\hat{\xi}_k)}$. 
Using this building block, QETU returns $F(\cos(\frac{\tau g(\hat{\xi}_k)}{2}))$, which implies that the function we need to implement using Chebyshev polynomials is $F(x) \sim f(g^{-1}(\frac{2}{\tau} \arccos(x)))$.
Any function $g(x)$ that results in $f(g^{-1}(\frac{2}{\tau} \arccos(x)))$ having divergent derivatives for $x \in [-1, 1]$ will consequently require $\OO(2^{\nq})$-degree Chebyshev polynomials to produce $F(x)$ exactly.

Consider first choosing $g(x) = x^m$ for $m = \{1, 2, \dots\}$.
As was the case with $m=1$, the function $F(x)$ diverges as $x = \pm 1$ for, and therefore requires $\OO(2^{\nq})$ Chebyshev polynomials to reproduce $F(x)$ exactly. 
Because each call to $\me^{-i \tau g(\hat{\xi}_k)}$ requires $\OO(\nq^m)$ gates, the cost of preparing a BE of $f(\hat{\xi}_k)$ using $\me^{-i \tau g(\hat{\xi}_k)}$ as a building block is $\OO(\nq^m 2^{n_q})$.

Consider now using non-polynomial $g(x)$.
Further suppose that some $g(x)$ exists (which we discuss next) that results in $F(x) \sim f(g^{-1}(\frac{2}{\tau}\arccos(x)))$ being infinitely differentiable. 
While this choice may reduce the number of Chebyshev polynomials required, the cost of implementing $\me^{-i \tau g(\hat{\xi}_k)}$ for non-polynomial $g(x)$ is generally $\OO(2^{\nq})$.

Taken together, these arguments imply that the cost of preparing BEs of $f_k(\hat{\xi}_k)$ using QETU generally scales exponentially with $\nq$.
With this in mind, the rest of this section focuses on finding a $g(x)$ that reduces the cost from $\OO(\nq 2^{\nq})$ to $\OO(2^{\nq})$.
%We now show that choosing $g(x) \sim \arccos(x)$ 
% \MK{Made minor adjustments throughout the paragraphs, plz check.}

\subsubsection{Technical details}

Our objective is to identify $g(x)$ such that $F(x) \sim f(g^{-1}(\frac{2}{\tau}\arccos(x)))$ is infinitely differentiable.
The form of $F(x)$ suggests that $g(x) \sim \arccos(x)$ is a good starting point.
% \MK{Minor change, check.}
The easiest way to accomplish this is to use the building block $\me^{-i 2 \arccos(\hat{\xi}_k/\alpha)}$, where $\alpha = \|\hat{\xi}_k\|$ to ensure the $\arccos$ is well defined.
Setting $\tau=1$ and replacing $\hat{A}$ in~\cref{thm:qetu} with $2 \arccos(\hat{\xi}_k/\alpha)$, we observe that the QETU circuit now returns $F(\hat{\xi}_k/\alpha)$, which is only a function of $\hat{\xi}_k$.
The procedure for determining the Chebyshev expansion of $f_k(\hat{\xi}_k)$ is now identical to that explained in~\cref{ssec:qsvt}; 
degree-$d$ polynomial functions of $\hat{\xi}_k$ can be implemented using $\OO(d)$ controlled calls to $\me^{-i 2\arccos(\hat{\xi}_k/\alpha)}$.

Implementing $\me^{-i 2\arccos(\hat{\xi}_k/\alpha)}$, however, is more complicated than $\me^{-i \tau \hat{\xi}_k}$.
While the operator $\arccos(\hat{\xi}_k/\alpha)$ is still diagonal, and can easily be determined classically, it is an infinite degree Chebyshev polynomial in its argument $\hat{\xi}_k/\alpha$.
This implies that its Pauli decomposition consists of $\OO(2^{\nq})$ identity and $R_z$ gates~\cite{Welch_2014} and a controlled implementation of the operator $\me^{-i 2\arccos(\hat{\xi}_k/\alpha)}$ requires $\OO(2^{\nq})$ CNOT and $R_z$ gates. 
On the other hand, constructing degree-$d$ polynomial functions $f_k(\hat \xi_k)$ out of this building block only requires an additional multiplicative factor of $d$, giving an overall asymptotic scaling of $\OO(d 2^{\nq})$. 

We thus see that choosing more complicated building blocks, such as $\me^{-i 2\arccos(\hat{\xi}_k/\alpha)}$ which requires exponential (in $\nq$) resources to implement, can potentially make constructing the functions $f_k(\hat \xi_k)$ simpler. One such method that we propose in this section, which we call Linear Operators Via Exponentiation and Linear Combination of Unitaries (\LOVELCU), uses $\me^{\pm i \arccos (f_k(\hat{\xi}_k)/\beta)}$
as a building block.
Since 
\begin{align}
    \frac{f_k(\hat{\xi}_k)}{\beta} = \frac{\me^{i \arccos \left(\frac{f_k(\hat{\xi}_k)}{\beta}\right)} + \me^{-i \arccos \left(\frac{f_k(\hat{\xi}_k)}{\beta}\right)}}{2}
    \,,
\end{align}
the construction of $f_k(\hat{\xi}_k)$ only needs LCU to combine these two terms.
This can be implemented using a simple LCU circuit, shown in~\cref{fig:lcu_arccos}, with $\PREPARE = \Hgate \otimes \unit_s$ and
\begin{equation}
\label{eq:sel_love_lcu}
\begin{alignedat}{9}
    \SELECT = \sket{0}\sbra{0}\otimes\, &\me^{i \arccos \left(\frac{f_k(\hat{\xi}_k)}{\beta}\right)} 
    \\
    &+ \sket{1}\sbra{1}\otimes \me^{-i \arccos \left(\frac{f_k(\hat{\xi}_k)}{\beta}\right)}\,.
\end{alignedat}
\end{equation}
The scale factor $\beta$ can be chosen as the smallest possible value $\beta = \|f_k(\hat{\xi}_k)\|$.
    
The construction of the building block is very similar to before. The operators $\me^{\pm i \arccos \left(f_k(\hat{\xi}_k)/\beta\right)}$ are diagonal, which can be determined in a straightforward way using classical resources.
It can then be decomposed into Pauli strings containing only the identity or $\Zgate$-gates; in general, the decomposition will contain $\OO(2^{\nq})$ Pauli strings.
The circuits for the diagonal unitary matrices $\me^{\pm i \arccos(\hat{A})}$ can then be implemented using at most\footnote{If the operator $\hat{A}$ has definite parity, for example $\hat{A} = \hat\varphi^n$ for some polynomial $n$, then only half of the Pauli strings will have non-zero coefficients.
This can be understood by noting that $\arccos(x) = \pi/2 - \arcsin(x)$, which implies that $\arccos(x)$ has the same parity of the argument up to an overall constant shift.} $2^{\nq}-1$ $R_z$ gates and $2^{\nq} - 2$ CNOT gates~\cite{Welch_2014},
and the entire \LOVELCU circuit requires $\OO(2^{\nq})$ gates and a single ancillary qubit.
This gate count can be further reduced by approximating the circuits for $\me^{i \arccos \hat A}$ in a systematic way by dropping rotation gates with small angles~\cite{Welch_2014, Li:2024lrl}.
\LOVELCU can therefore prepare BEs of $f_k(\hat\xi_k)$ with resources that are independent of the degree $d$ of the function $f$.

Another advantage of \LOVELCU is that it is not restricted to BEs of single variables, but can also be applied to functions of multiple variables, \emph{i.e.} the gate cost of using \LOVELCU to construct a BE of the operator $f(\hat\varphi_1-\hat\varphi_2)$ is the same as a general function $f(\hat\varphi_1, \hat\varphi_2)$.
This implies that, instead of preparing BEs of $f_0(\hat{\varphi}_i)$, $f_0(\hat{\varphi}_j)$ and $f_2(\hat{\varphi}_i-\hat{\varphi}_j)$ separately and then adding them using another layer of LCU, one can directly use \LOVELCU to prepare a BE of the sum $f_0(\hat{\varphi}_i) + f_0(\hat{\varphi}_j) + f_2(\hat{\varphi}_i-\hat{\varphi}_j)$.

The cost of \LOVELCU can be further reduced by exploiting the structured form of the $\SELECT$ oracle.
In particular, $\SELECT$ oracles of the form in Eq.~\eqref{eq:sel_love_lcu} can be prepared with the techniques used in the control-free implementation of QETU.
Rather than performing two controlled calls to $\me^{-i \arccos(f(\hat{\xi}_k/\alpha))}$, one can construct $\SELECT$ for the cost of a single non-controlled circuit for $\me^{-i \arccos(f(\hat{\xi}_k/\alpha))}$ and an additional $2 \nq$ CNOT gates; the procedure for this control-free implementation can be found in Appendix~B of Ref.~\cite{Kane:2023jdo}.

While here we focused on using \LOVELCU for block encoding diagonal operators, a similar construction can be applied to general Hermitian operators. 
The procedure is similar to the diagonal operator case; for an $\nq$-qubit operator $\hat A$, one would evaluate $\me^{-i \arccos(\hat A)}$ classically, and then exactly decompose the unitary operator into $\OO(4^{\nq})$ gates~\cite{Shende:2006onn}.

\begin{figure}[t]
    \centering
    \begin{quantikz}[row sep={1cm,between origins}, column sep=0.25cm]
        \lstick{$\sket{0}$} & \qw & \gate{\Hgate} & \ctrl{1} & & \octrl{1} & \gate{\Hgate} &
        \\
        \lstick{$\sket{\psi}$} & \qw & \qw & \gate{\me^{i \arccos \hat{A}}} & \qw & \gate{\me^{-i \arccos \hat{A}}} & \qw &
    \end{quantikz}
    \caption{LCU circuit for block encoding of operator $\hat{A}$.}
    \label{fig:lcu_arccos}
\end{figure}
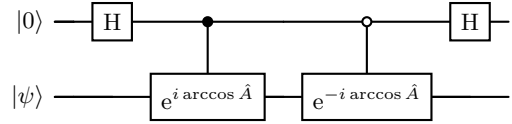

\subsection{Constructing full Hamiltonian Block Encoding and Walk Operator \label{ssec:UH_WH_general}}

The previous sections discussed methods for constructing BEs of local terms in the Hamiltonian in Eq.~\eqref{eq:general_H}.
In this section, we discuss how to construct a BE of the full Hamiltonian $U_{\hat{H}}$.
Furthermore, we show that there exists a simple operator $S$ satisfying the conditions in Eqs.~\eqref{eq:S_cond1} and~\eqref{eq:S_cond2}, which can be used to construct the walk operator $W_{\hat{H}}$ while avoiding the large overall prefactor when using the general qubitization procedure in Ref~\cite{Low:2016znh}.
Here we only provide a high-level procedure; the detailed proofs are given in~\cref{app:WH}.

For this discussion, we write the Hamiltonian in the general form $\hat{H} = \sum_{\ell} \hat{H}_{\ell}$, where $\hat{H}_{\ell}$ is a local term in the full Hamiltonian $\hat{H}$ that it is diagonalized by an efficiently implementable unitary (see \cref{eq:diagonal}). Each $\hat{H}^{(D)}_{\ell}$ is then given by
\begin{align}
    \hat{H}^{(D)}_{\ell} = f_{k_{\ell}}(\hat{\xi}_{k_{\ell}}) 
\end{align}
The BE of the full Hamiltonian $U_{\hat{H}}$ is constructed using a standard LCU procedure to add BEs of local terms $U_{\hat{H}_{\ell}}$.
We denote by $\chi_{\ell}$ the number of gates required to construct a given $U_{\hat{H}_{\ell}}$, and define $\chi_{\rm site} = \max_{\ell} \chi_{\ell}$.
Using LCU to add each of the $\OO(\NLat)$ local BEs $U_{\hat{H}_{\ell}}$ requires $\OO(\log \NLat)$ ancillary qubits, the register of which we denote by $a_\Lambda$.
Since the local terms $f_{k_{\ell}}(\hat{\xi}_{k_{\ell}})$ for a given $k_{\ell}$ all have identical coefficients, the $\PREPARE$ oracle will have a relatively simple circuit implementation and will be sub-leading in the asymptotic gate cost.
The $\SELECT$ oracle is given by $\SELECT = \sum_{\ell} \sket{\ell}\sbra{\ell} \otimes U_{\hat{H}_{\ell}}$, where each $U_{\hat{H}_{\ell}}$ is controlled on $\OO(\log \NLat)$ qubits, each requiring $\OO(\chi_{\rm site} \log \NLat)$ gates.
The gate cost of the full $\SELECT$ oracle, and therefore the gate cost of constructing $U_{\hat{H}}$, is $\OO(\chi_{\rm site} \NLat \log \NLat)$. 

We now describe how to construct the walk operator $W_{\hat{H}}$ in a way that avoids the expensive general qubitization procedure in Ref.~\cite{Low:2016znh}.
The first step is to determine operators $S$ for individual BEs $U_{\hat{H}_{\ell}}$ constructed using the methods in this work.
Using this information, we then identify an $S$ that can be used to construct $W_{\hat{H}}$ from the BE $U_{\hat{H}}$.

As shown in detail in Appendix~\ref{app:WH}, for BEs $U_{\hat{H}_{\ell}}$ constructed using QSVT or QETU for even polynomials, $W_{\hat{H}_{\ell}}$ can be constructed by choosing $S$ to be a single $\Zgate$-gate acting on the signal and control qubit, respectively; this is due to the highly symmetric structure of QSVT and QETU circuits.
Turning to \LOVELCU, choosing $S$ to be a single $\Zgate$-gate acting on the ancillary qubit in the \LOVELCU circuit also satisfies the relations in Eq.~\eqref{eq:S_cond} necessary to construct $W_{\hat{H}_j}$.
To see this, first observe that $\Zgate \sket{0} = \sket{0}$, which implies  that~\cref{eq:S_cond1} is automatically satisfied.
If we denote by $U_A$ the circuit in~\cref{fig:lcu_arccos}, setting $U = \me^{i \arccos \hat{A}}$, we see
\begin{equation}
\begin{split}
    (\Zgate &\otimes \unit_s ) U_A 
    \\
    &= \left(\Hgate \Xgate \otimes \unit_s \right)\left( \sket{0}\sbra{0}\otimes U + \sket{1}\sbra{1}\otimes U^\dagger \right) \left(\Hgate \otimes \unit_s \right)
    \\
    &= \left(\Hgate \otimes \unit_s \right)\left( \sket{1}\sbra{1}\otimes U + \sket{0}\sbra{0}\otimes U^\dagger \right) \left(X \Hgate \otimes \unit_s \right)
    \\
    &= U_A^\dagger (Z \otimes \unit_s ),
\end{split}
\end{equation}
where we have used the relations $\Zgate \Hgate = \Hgate \Xgate$, $\Xgate \sket{0} \sbra{0} = \sket{1}\sbra{1}\Xgate$ and $\Xgate \sket{1} \sbra{1} = \sket{0}\sbra{0}\Xgate$, which satisfies the relation in~\cref{eq:S_cond2}.

The BE $U_{\hat{H}}$ for the full Hamiltonian is constructed using LCU to combine local BEs $U_{\hat{H}_{\ell}}$, which, for BE methods used in this work, share a common $S$ operator.
By using a specific qubit organization of the individual $U_{\hat{H}_{\ell}}$'s, these properties imply that $W_{\hat{H}}$ for the full Hamiltonian can be constructed from $U_{\hat{H}}$ using the same common operator $S$, namely a single $\Zgate$-gate.
In particular, each $U_{\hat{H}_{\ell}}$ shares the same ancillary qubit register, denoted by $a$; if each $U_{\hat{H}_j}$ uses $n_{\rm anc}^{(j)}$ ancillary qubits, then the register $a$ contains $n_{a} = \max_j(n_{\rm anc}^{(j)})$ qubits.
Additionally, we make the choice that, if $n_{\rm anc}^{(j)} < n_{a}$ for some $j$, the ancillary qubits used to construct $U_{\hat{H}_{\ell}}$ are $a_0, a_1, \dots, a_{n_{\rm anc}^{(j)}-1}$.
Under these assumptions, it is shown in Appendix~\ref{app:WH} that $W_{\hat{H}}$ can be constructed by choosing $S$ to be a single $\Zgate$, specifically $S = \unit_{a_\Lambda} \otimes \left(\Zgate \otimes \unit \right)_a \otimes \unit_s$.

We conclude by stressing that this simple construction of $W_{\hat{H}}$ is applicable to situations where one uses different methods to prepare different local BEs $U_{\hat{H}_{\ell}}$.
This implies that one is free to choose the BE method that results in the lowest gate cost for each individual term $U_{\hat{H}_{\ell}}$.

\section{Numerical demonstration \label{sec:scalar_field_theory_numerics}}

In~\cref{sec:be} we described three main methods to block encode each local Hamiltonian $\hat H_{\ell}$, and provided the dependence of the required resources on the number of qubits $\nq$ and degree of the polynomial of $f_{k}$.  
We did mention at the beginning of that section that some methods come with large pre-factors, which are not captured by the asymptotics provided.
In this section we provide a more detailed numerical analysis of the resources required, including all required prefactors. 
Such an analysis requires making a detailed choice of the Hamiltonian, and we apply our techniques to the example of a scalar field theory Hamiltonian.
Next, we compare the explicit gate cost of preparing BEs of local terms in the Hamiltonian using the methods described in~\cref{sec:be}.
We then use the BE with lowest resource requirements and obtain explicit gate counts for using GQSP to time evolve a both a single and two-site system, and compare the cost to time evolution using 2\textsuperscript{nd} and 4\textsuperscript{th} order product formulas.

We consider a hypercubic lattice $\Lambda$ of spatial dimension $\dims$ and lattice spacing $a$. The Hamiltonians we consider can be written (using the notation from \cref{ssec:general_H}) as follows:
\begin{subequations}
\label{eq:hamiltonian}
\begin{alignat}{9}
    &\hat{H} = \hat{H}_\varphi + \hat{H}_\pi\,,
    \\
    \label{eq:Hphib}
    &\hat{H}_\varphi = a^\dims \left[ \sum_{i=1}^{\NLat} V\left(\hat\varphi_i\right) 
    +
    \sum_{\langle i j \rangle} \dfrac{(\hat\varphi_j - \hat\varphi_i)^2}{2a^2}\right],
    \\
    &\hat{H}_\pi = a^\dims \sum_{i=1}^{\NLat} \dfrac{1}{2}\hat\pi^2_i = a^\dims \sum_{i=1}^{\NLat} \dfrac{1}{2}\FT_i^\dagger \bigl(\hat{\pi}_i^{(D)}\bigr)^2 \FT_i\,,
\end{alignat}
\end{subequations}
where $V\left(\hat\varphi_i\right)$ denotes the potential function. In this work, we consider the following potentials \begin{align}
    V_1\left(\hat\varphi_i\right) &= \dfrac{m^2}{2} \hat\varphi_i^2 + \dfrac{\lambda}{4!} \hat\varphi_i^4\,, \\
    V_2\left(\hat\varphi_i\right) &= g \cos(\hat\varphi_i)\,,
\end{align}
where $V_1(\hat\varphi_i)$ is the standard $\varphi^4$ scalar field theory potential, and the periodic potential $V_2(\hat\varphi_i)$ is relevant for dual basis formulations of compact U(1) gauge theories~\cite{Haase:2020kaj,Bauer:2021gek}, and a mixed-basis formulation of SU(2) gauge theory~\cite{DAndrea:2023qnr}.
Setting the lattice spacing $a=1$, the functions $f_k(\hat\xi_k)$ for this Hamiltonian are
\begin{subequations}
\label{eq:local_f_scalar_ft}
\begin{alignat}{9}
    f^{(1)}_0(\hat\varphi_i) &= V_1(\hat\varphi_i) = \dfrac{m^2}{2} \hat\varphi_i^2 + \dfrac{\lambda}{4!} \hat\varphi_i^4\,,
    \\
    f^{(2)}_0(\hat\varphi_i) &= V_2(\hat\varphi_i) = g \cos(\hat\varphi_i)\,,
    \\
    f_1(\hat\pi_i) &= (\hat{\pi}_i^{(D)})^2\,,
    \\
    f_2(\hat\varphi_i-\hat\varphi_j) &= \frac{1}{2}(\hat\varphi_j - \hat\varphi_i)^2\,.
\end{alignat}
\end{subequations}

For all numerical results shown in this section, we set $m=1, \lambda = 32$ and $g=1$.
To obtain the gate counts we use the basis gate-set containing CNOT, $R_x$, and $R_z$ gates.
There are several options for compiling multicontrolled gates.
One choice is to use ancillary qubits as in Ref.~\cite{Nielsen:2012yss} and then compile the resulting circuit down into CNOT, $R_x$, and $R_z$ gates using, e.g., QISKIT transpilers.
Another choice is to exploit the fact that all multicontrolled gates required for this work are diagonal matrices, which implies they can be compiled down into $R_z$ and CNOT gates using the Walsh function formalism~\cite{Welch_2014}.
For the circuits studied in this work, we found that the latter choice resulted in smaller total gate counts, with significant reductions in some cases.
All $\PREPARE$ oracles are implemented using QISKIT's exact state preparation algorithm.
After compiling multicontrolled gates in this way, the final gate count in terms of CNOT, $R_x$ and $R_z$ gates is obtained using QISKIT transpiler with optimization level 1.

Following the discussion in Sec.~\ref{sec:be}, the scale factor for $f_0^{(1)}(\hat \varphi_i)$, $f_1(\hat \pi_i)$, and $f_2(\hat\varphi_i-\hat\varphi_j)$ is the same for all BE methods considered (see Appendix~\ref{app:lcu_scale_factor} for a derivation of the scale factors when using vanilla LCU).
Furthermore, the associated scale factors are given by their minimum values.
As explained in more detail in Appendix~\ref{app:lcu_scale_factor}, the same is not true for $f_0^{(2)} = g \cos(\hat \varphi)$; using standard LCU techniques results in a value of the scale factor $\sim 30\%$ larger than optimal.
However, as we will show, \LOVELCU has exponentially better complexity in $\nq$ compared to standard LCU, and this slight difference in scale factor does not change our qualitative findings.

\Cref{fig:be_cost_indiv_terms} shows the number of rotation gates and ancillary qubits required to construct BEs of the functions in~\cref{eq:local_f_scalar_ft}.
Plots of CNOT gate counts are shown in~\cref{app:cnot_count}.

Considering first constructing BEs of $f_0^{(1)}(\hat{\varphi}_i)$ and $f_1(\hat\pi_i^{(D)})$, we find that, despite the asymptotically inefficient scaling with $\nq$, using \LOVELCU requires the fewest rotation gates for $\nq \lesssim 11$; for $\nq=4$ qubits, \LOVELCU requires only 13 rotation gates.
For $\nq \gtrsim 11$, due to the linear scaling with $\nq$, using QSVT requires the fewest rotation gates. 
However, because digitization errors typically decrease exponentially with $\nq$~\cite{Klco:2018zqz, Bauer:2021gek}, realistic quantum simulations will not require so large a value of $\nq$, implying that in some cases of practical interest \LOVELCU will require the fewest rotation gates.
Furthermore, QSVT requires $\OO(\log \nq)$ ancillary qubits, while \LOVELCU requires only a single ancillary qubit.

\begin{figure*}[htp]
    \centering
    \includegraphics[width=0.48\textwidth]{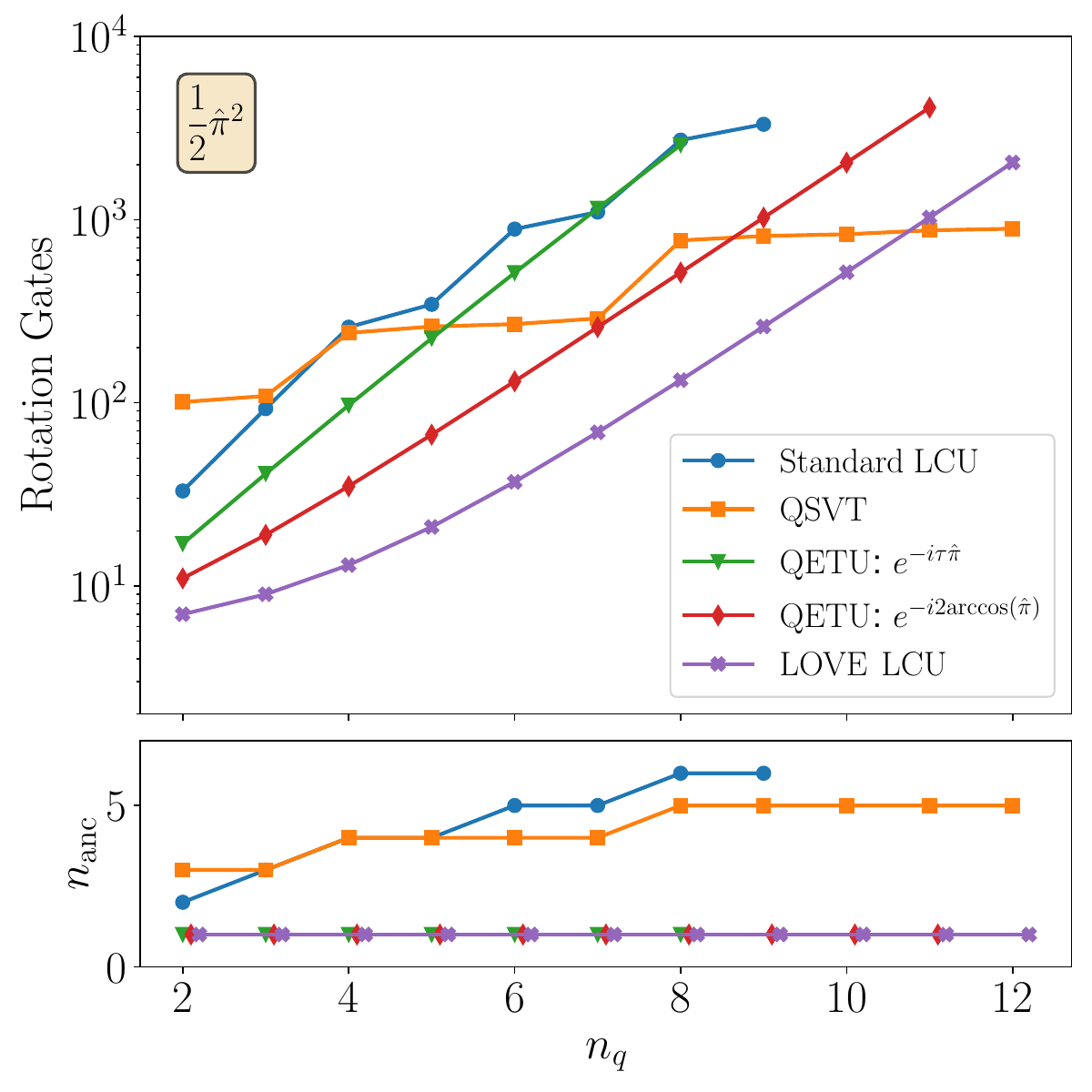}
    \includegraphics[width=0.48\textwidth]{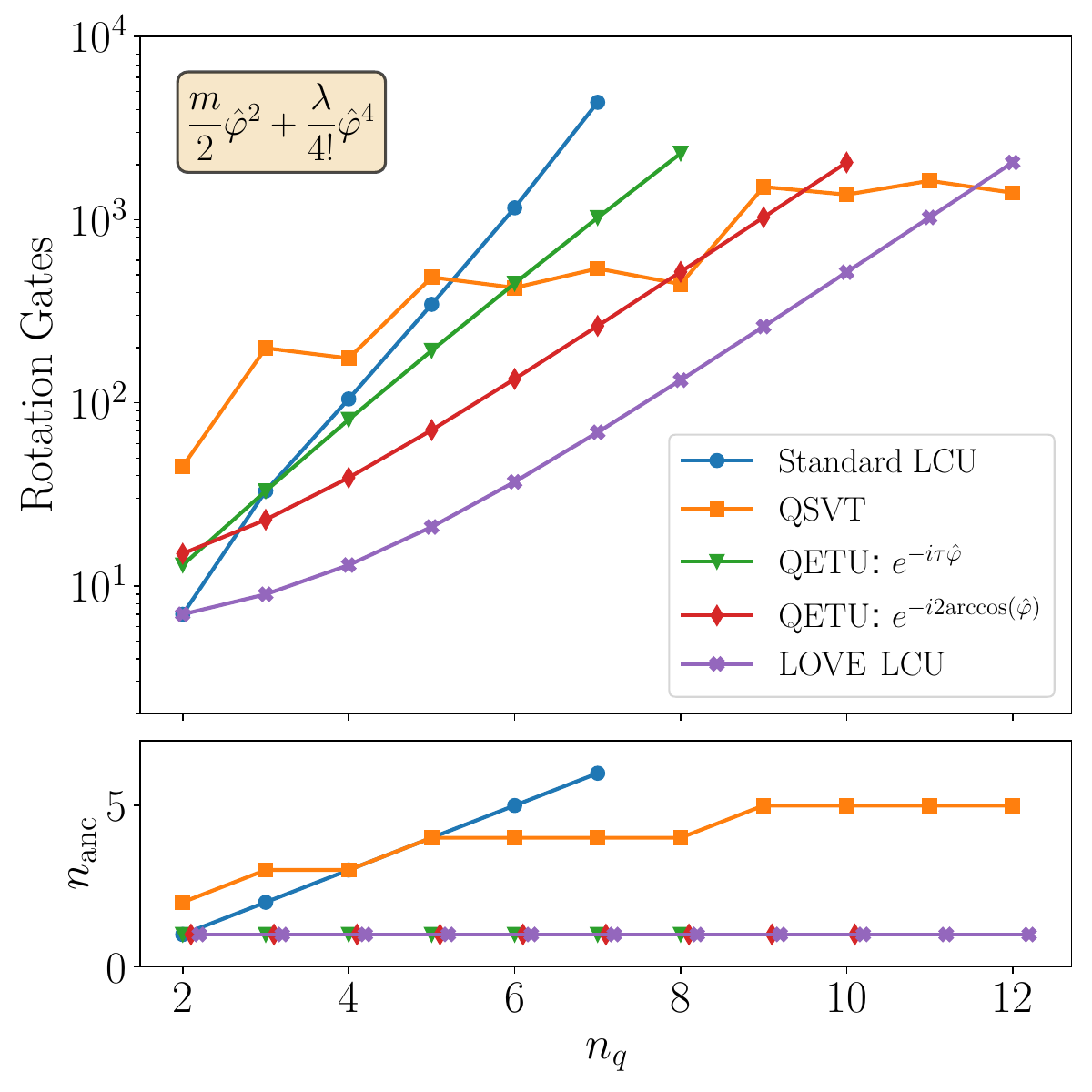}
    \includegraphics[width=0.48\textwidth]{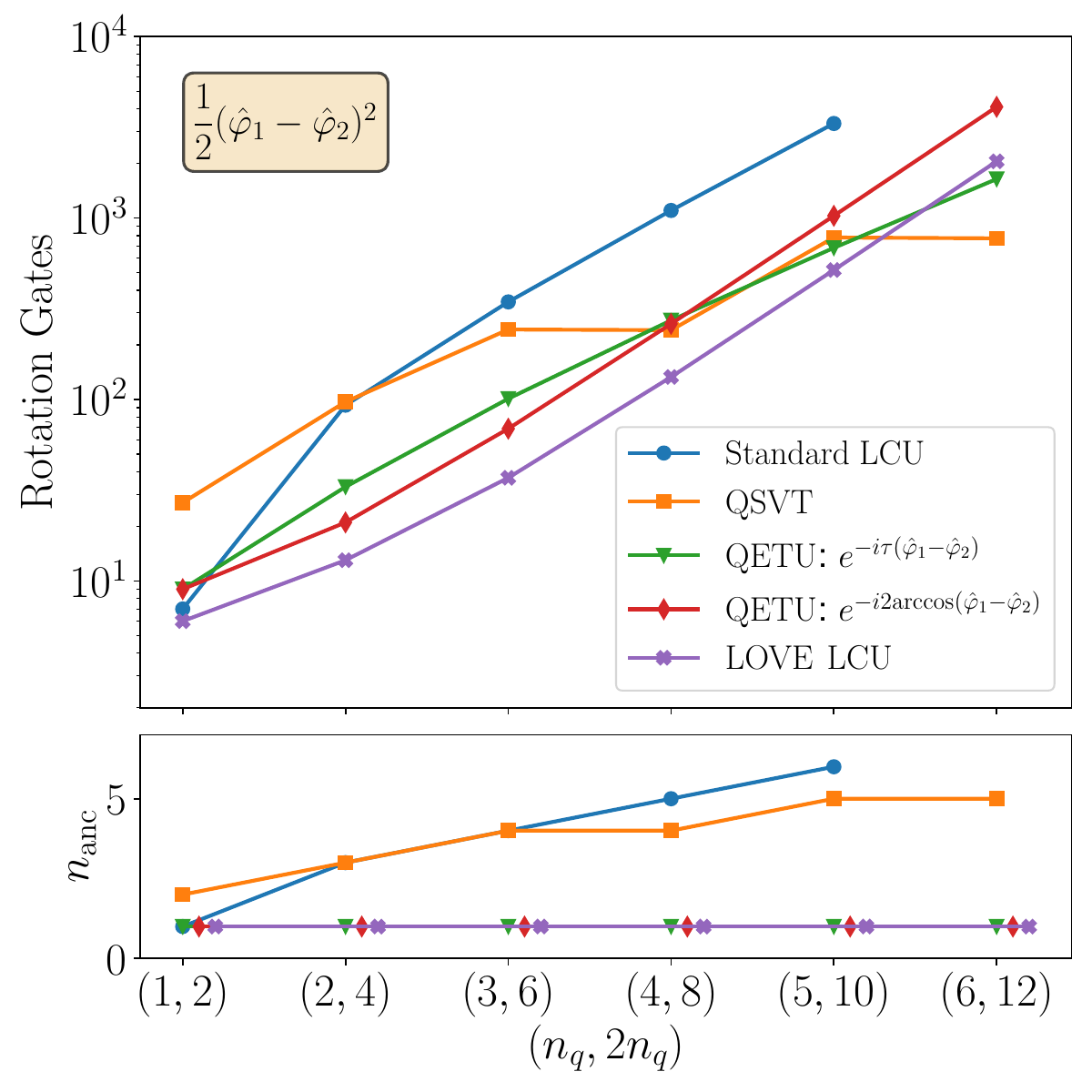}
    \includegraphics[width=0.48\textwidth]{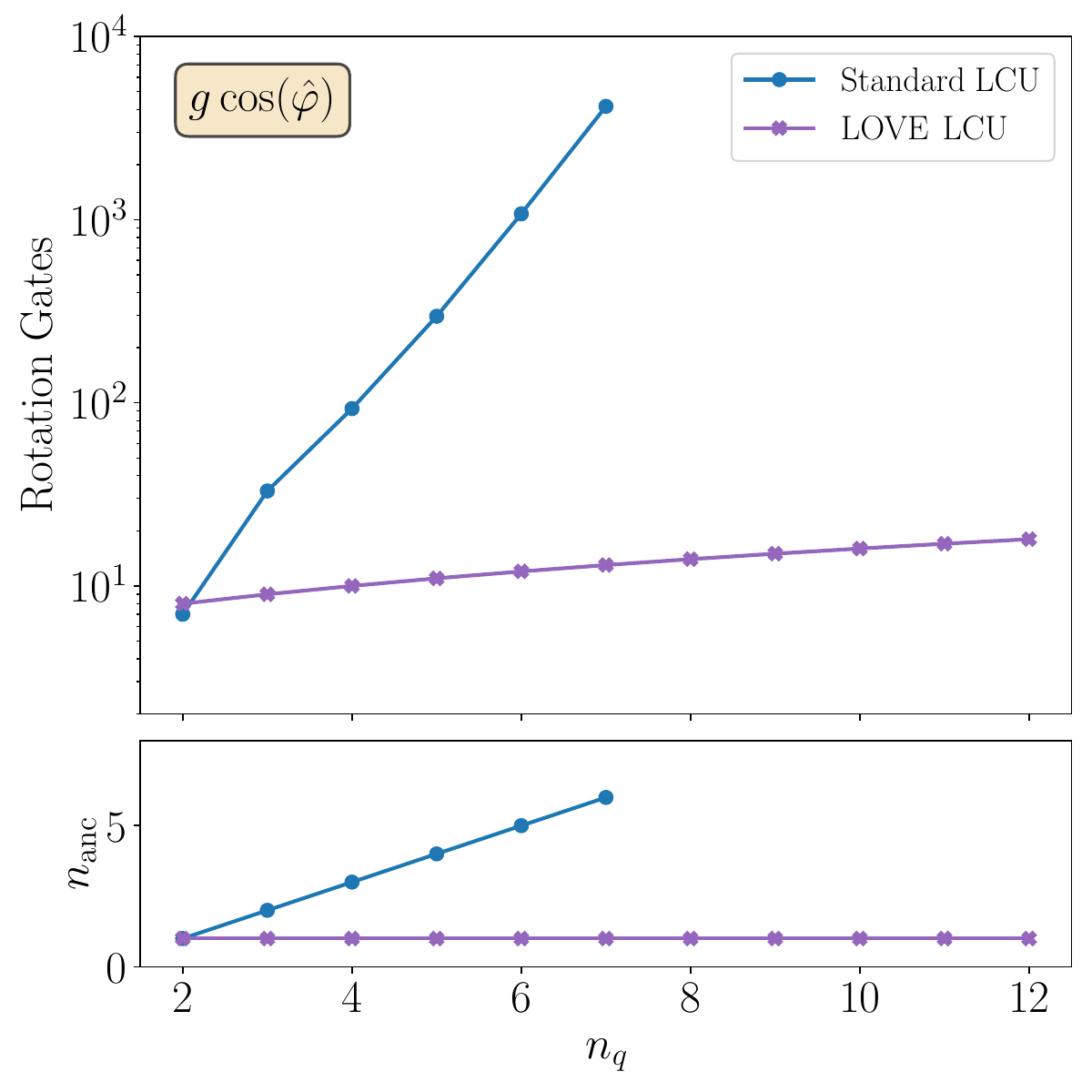}
    \caption{Rotation gate count and number of ancillary qubits required to block encode local bosonic operators. 
    The top left, top right, bottom left, and bottom right plots show resource requirements to block encode $\frac{1}{2}\hat\pi^2$, $\frac{m}{2}\hat\varphi^2+\frac{\lambda}{4!} \hat\varphi^4$, $\frac{1}{2}(\hat\varphi_1-\hat\varphi_2)^2, g\cos(\hat\varphi)$, respectively, for $m=1, \lambda=32, g=1$.
    Different colored and shaped data points correspond to different methods to prepare the BE.
    Points with the same number of ancillary qubits have been shifted slightly for clarity.
    For the single-site operators $\hat\pi^2$ and $\frac{1}{m}\hat\varphi^2+\frac{\lambda}{4!} \hat\varphi^4$, \LOVELCU requires the fewest rotation gates for $\nq \leq 11$.
    For the two-site operator $\frac{1}{2}(\hat\varphi_1-\hat\varphi_2)^2$, \LOVELCU requires the fewest rotation gates for $\nq < 6$.
    \LOVELCU constructs a BE of $\cos(\hat\varphi)$ using the fewest gates for all $\nq$.}
    \label{fig:be_cost_indiv_terms}
\end{figure*}

\begin{figure}[htp]
    \centering
    \includegraphics[width=\linewidth]{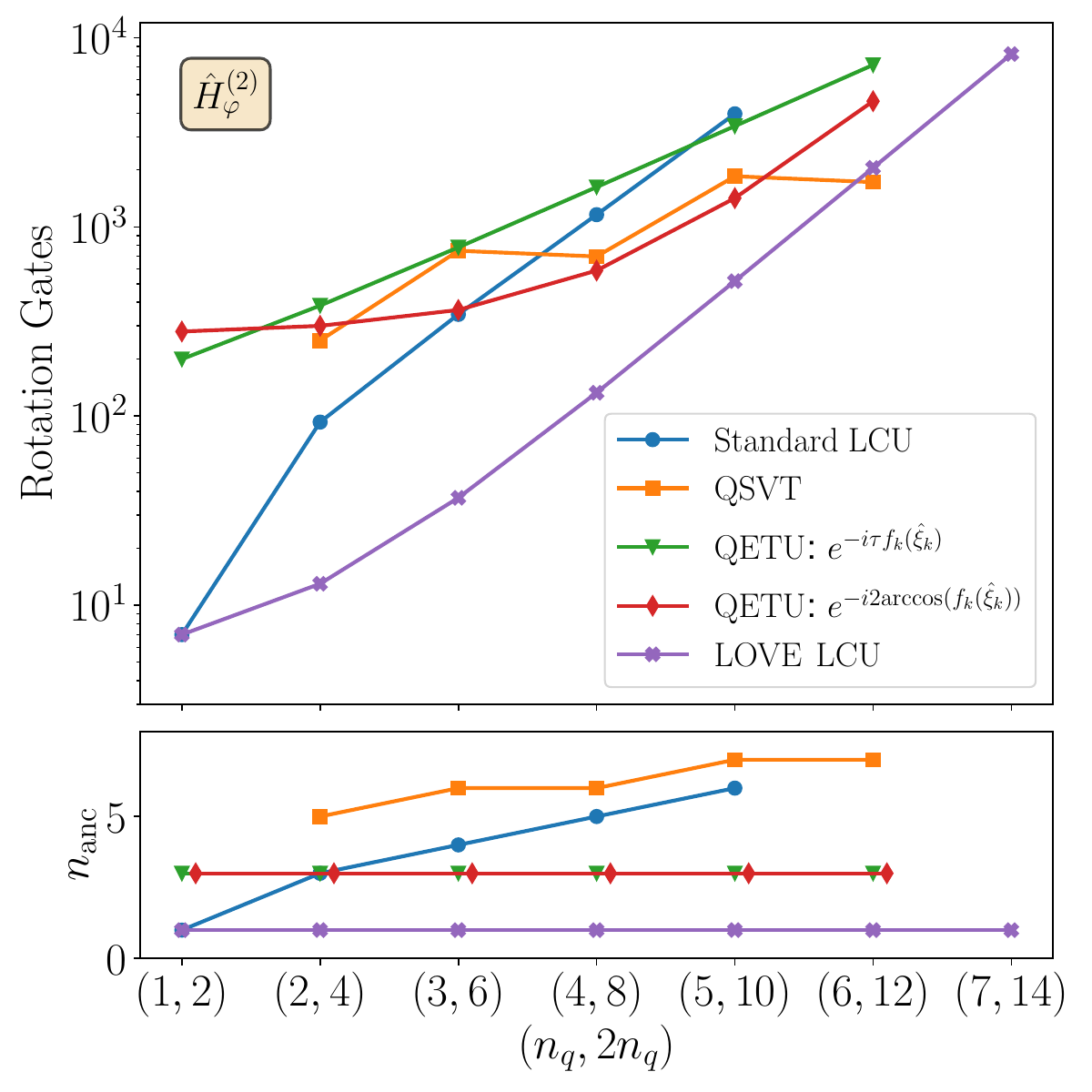}
    \caption{Rotation gate count and number of ancillary qubits required to block encode the two site Hamiltonian $\hat{H}^{(2)}_\varphi$ in Eq.~\eqref{eq:Hphi_two_site}.
    Different colored and shaped data points correspond to different methods to prepare the BE.
    Points with the same number of ancillary qubits have been shifted slightly for clarity.
    While QSVT and QETU based methods require using a final layer of LCU to add the local BEs $f_0^{(0)}(\hat\varphi_1), f_0^{(0)}(\hat\varphi_2), f_2(\hat\varphi_1-\hat\varphi_2)$, \LOVELCU does not.
    For this reason, \LOVELCU requires only a single ancillary qubit.
    This advantage also leads to \LOVELCU outperforming all other methods for $\nq \lesssim 6$; for $\nq>6$, QSVT outperforms other methods due to the relative exponentially improved asymptotic scaling.}
    \label{fig:rot_anc_be_Hphi_two_site}
\end{figure}

We now turn to preparing a BE of $f_2(\hat\varphi_i-\hat\varphi_j)$.
Because this operator acts on two lattice sites and is of dimension $2^{2\nq}$, the relative cost of the different methods are fundamentally different than the single qubit operators.
In particular, the asymptotic gate cost of \LOVELCU is now $\OO(2^{2\nq})$, which is the worst of all methods.
Despite this fact, \LOVELCU requires the fewest rotation gates for $\nq \leq 5$.
For $\nq > 5$, using either QSVT or QETU with $\me^{-i \tau \hat\varphi}$ as a building block requires fewer gates than \LOVELCU, with QSVT being the cheapest due to the $\OO(\nq \log \nq)$ asymptotic gate count scaling.

Because using $\nq \sim 5$ is a much more realistic value for quantum simulations, it appears at first glance that using QSVT is the best method to prepare a BE of $f_2(\hat\varphi_i-\hat\varphi_j)$.
However, \LOVELCU has the advantage that it can prepare BEs of arbitrary $2 \nq$-qubit Hermitian operators for the same computational cost.
This implies that one could use \LOVELCU to prepare a BE of, e.g., $f_0^{(0)}(\hat\varphi_1) + f_0^{(0)}(\hat\varphi_2) + f_2(\hat\varphi_1-\hat\varphi_2)$, for the same cost as just $f_2(\hat\varphi_1-\hat\varphi_2)$.
The same is not true for QSVT or the QETU based methods, which would require one to first prepare a BE of $f_0^{(0)}(\hat\varphi_1), f_0^{(0)}(\hat\varphi_2)$, and $f_2(\hat\varphi_1-\hat\varphi_2)$, and then add these operators using LCU.
A comparison using different methods to prepare BE of the two-site Hamiltonian
\begin{equation}
\begin{split}
\label{eq:Hphi_two_site}
    \hat H_\varphi^{(2)} &= f_0^{(0)}(\hat\varphi_1) + f_0^{(0)}(\hat\varphi_2) + f_2(\hat\varphi_1-\hat\varphi_2)
    \\
    &=\sum_{j=1}^2\left(\frac{1}{2}m^2 \hat\varphi_j^2 + \frac{\lambda}{4!}\hat\varphi_j^2\right) + \frac{1}{2}\left(\hat\varphi_1-\hat\varphi_2 \right)^2 
\end{split}
\end{equation}
is shown in Fig.~\ref{fig:rot_anc_be_Hphi_two_site}.
Due to the need for an additional layer of LCU, we see that the cost of QSVT and QETU based methods has increased relative to standard LCU and \LOVELCU methods; \LOVELCU is now the best method for $\nq \lesssim 6$, with QSVT requiring almost identical rotation gate counts for $\nq=6$.

The final operator we study is $\cos(\hat\varphi)$.
Using $\hat A = \cos(\hat\varphi)$ is a special case of \LOVELCU; the $\SELECT$ oracle reduces to 
\begin{equation}
    \SELECT = \sket{0}\sbra{0} \otimes e^{-i \hat\varphi} + \sket{1}\sbra{1} \otimes e^{i \hat\varphi},
\end{equation}
which can be implemented using the control-free procedure using $2\nq$ of CNOT gates and $\nq$ of $R_z$ gates.
While using QETU with $\me^{-i \tau \hat\varphi}$ can also achieve an asymptotic gate count of $\OO(\nq)$, it will have a larger overall prefactor in the cost.
Because $\cos(x)$ is an infinite degree polynomial, QSVT (or QETU using $\me^{-i 2 \arccos(\hat\varphi/\alpha)}$ as the building block), can only prepare approximate BEs of $\cos(\hat\varphi)$;  
while one can use the efficient Jacobi-Anger expansion to do this, one must use a polynomial of degree $d = \OO(\|\hat \varphi\| + \log(1/\epsilon)) = \OO(2^{\nq/2} + \log(1/\epsilon))$ to approximate to error $\OO(\epsilon)$.
Since $\cos(\hat\varphi)$ is a sum of $\OO(2^{\nq})$ Pauli strings, the asymptotic gate cost of using LCU to BE $\cos(\hat\varphi)$ requires $\OO(\nq 2^{\nq})$ gates.
From this discussion, we see that using \LOVELCU has both the best asymptotic gate scaling as well as the smallest overall prefactor due to the simplicity of the $\SELECT$ oracle.
Despite the clear superiority of \LOVELCU when preparing BEs of $\cos(\hat\varphi)$, a comparison the gate count using \LOVELCU and naive LCU to prepare a BE of $\cos(\hat\varphi)$ is shown in the bottom right plot in~\cref{fig:be_cost_indiv_terms}; \LOVELCU can prepare a BE of $\cos(\hat\varphi)$ for $\nq=12$ using 18 rotation gates and a single ancillary qubit.

We now turn to the problem of simulating time evolution, considering both 2nd and 4th order PFs, and GQSP where the BE was constructed using \LOVELCU.
We study the gate cost to construct the time evolution operator as a function of the evolution time $t$ and error $\epsilon$; the error is defined as $\epsilon = \| U(t) - U_\epsilon(t)\|$, where $U(t)$ is the exact time evolution operator calculated using numerical exponentiation, and $U_\epsilon(t)$ is the approximate time evolution operator calculated using GQSP or PFs.
Here we only show the rotation gate counts; CNOT gate counts are given in Appendix~\ref{app:cnot_count}.

The first system we study is the singe-site Hamiltonian
\begin{equation}
\begin{split}
\label{eq:H_single_site}
    \hat{H}^{(1)} &= \mathcal{F}^\dagger f_1(\hat\pi^{(D)}) \mathcal{F} + f_0^{(0)}(\hat\varphi)
    \\
    &=\frac{1}{2} \hat\pi^2 + \frac{1}{2}m^2 \hat\varphi^2 + \frac{\lambda}{4!} \hat\varphi^4.
\end{split}
\end{equation}
\Cref{fig:single_site_rot_vs_eps_t} shows the number of rotation gates as a function of $t$ and $\epsilon$ for $\nq=3$.
We find that GQSP requires less rotation gates compared to using PFs for errors as large as $\epsilon \lesssim 10^{-2}$, which improves upon the work in Ref.~\cite{Hariprakash:2023tla} by $\sim 5$ orders of magnitude; the majority of the savings leading to this significant improvement are a factor of $\sim 13$ gate reduction in the construction of controlled calls to $W_H$, and a factor of two gate reduction from using GQSP as opposed to standard QSP.
Another notable observation is that, in the region of error $\epsilon$ where using GQSP outperforms PFs, the 2\textsuperscript{nd} order PF outperforms the 4\textsuperscript{th} order PF.
This is due to the fact that the 4\textsuperscript{th} order PF has not yet reached a small enough value of $\epsilon$ for the expected $(1/\epsilon)^{1/4}$ scaling to take effect.
While GQSP algorithms are typically thought of as appropriate in the era of full fault-tolerance, our result for the single-site system suggests that for certain systems post-Trotter methods can outperform the PF-based ones sooner than expected.

To better understand how the relative costs change with the number of lattice sites, we now consider the two-site Hamiltonian
\begin{equation}
\begin{split}
\label{eq:H_two_site}
    \hat H^{(2)} &= \sum_{j=1}^2\left[\mathcal{F}_j^\dagger f_1(\hat\pi^{(D)}_j) \mathcal{F}_j + f_0^{(0)}(\hat\varphi_j) \right]
    + f_2(\hat\varphi_1-\hat\varphi_2)
    \\
    &=\sum_{j=1}^2\left(\frac{1}{2}\hat\pi_j^2 + \frac{1}{2}m^2 \hat\varphi_j^2 + \frac{\lambda}{4!}\hat\varphi_j^2\right) + \frac{1}{2}\left(\hat\varphi_1-\hat\varphi_2 \right)^2. 
\end{split}
\end{equation}
\Cref{fig:two_site_rot_vs_eps_t} shows the number of rotation gates as a function of $t$ and $\epsilon$ for $\nq=3$.
We see that GQSP now outperforms 2nd and 4th order PFs for $\epsilon \lesssim 5 \times 10^{-5}$ and $\epsilon \lesssim 10^{-7}$, respectively.
While this result suggests one should use PFs for larger lattice sizes, it is likely that significant reductions in the GQSP gate count can be achieved by applying system specific optimizations for compiling the final LCU $\SELECT$ and $\PREPARE$ oracles.
Such gate reductions could make GQSP competitive with PFs for larger errors, and is an interesting direction for future work.

\begin{figure*}
    \centering
    \includegraphics[width=0.6\textwidth]{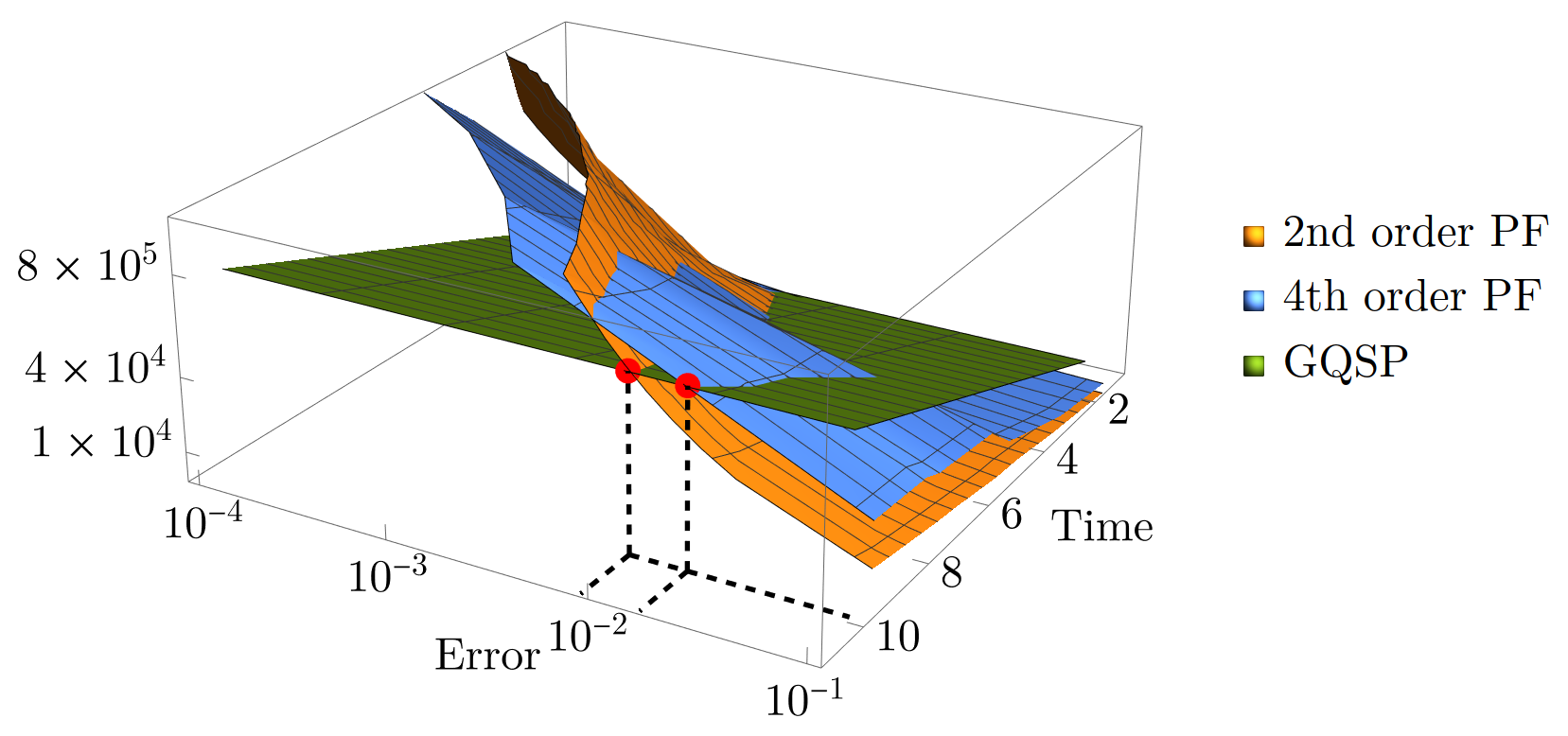}
    \caption{Rotation gate count as a function of error $\epsilon$ and simulation time $t$ for simulating time-evolution of the single-site Hamiltonian in Eq.~\eqref{eq:H_single_site}.
    The blue, orange, and green surfaces show results calculated using a $2^{\rm nd}$ order PF, a $4^{\rm th}$ order PF, and GQSP where the BE was constructing using \LOVELCU.
    For $t=10$, GQSP outperforms $2^{\rm nd}$ and $4^{\rm th}$ order PFs for $\epsilon \lesssim 10^{-2}$ and $\epsilon \lesssim 2\times 10^{-2}$, respectively.}
    \label{fig:single_site_rot_vs_eps_t}
\end{figure*}

\begin{figure*}
    \centering
    \includegraphics[width=0.6\textwidth]{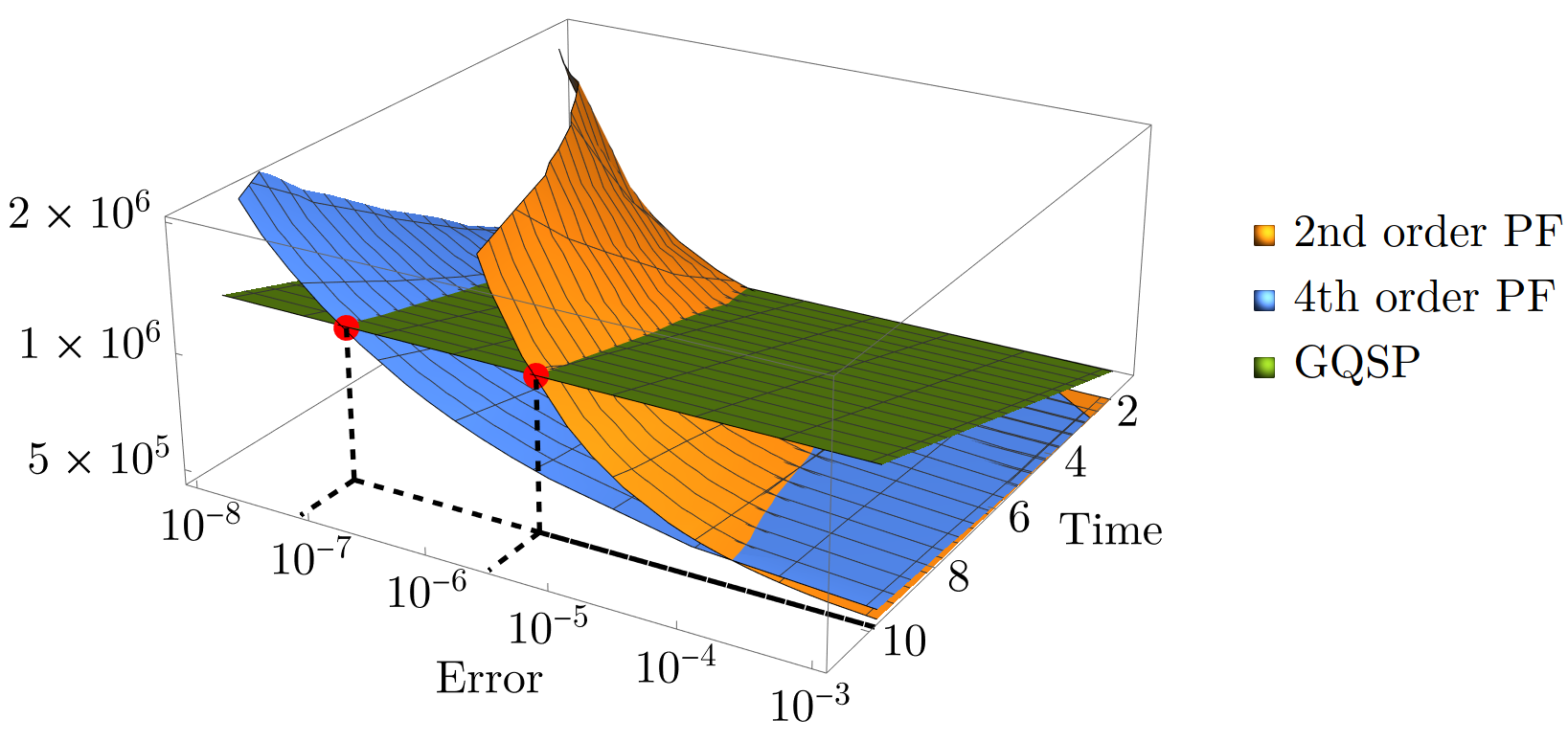}
    \caption{Rotation gate count as a function of error $\epsilon$ and simulation time $t$ for simulating time-evolution of the two-site in Eq.~\eqref{eq:H_two_site}.
    The blue, orange, and green surfaces show results calculated using a $2^{\rm nd}$ order PF, a $4^{\rm th}$ order PF, and GQSP where the BE was constructing using \LOVELCU.
    For $t=10$, GQSP outperforms $2^{\rm nd}$ and $4^{\rm th}$ order PFs for $\epsilon \lesssim 5\times 10^{-6}$ and $\epsilon \lesssim 1\times 10^{-7}$.}
    \label{fig:two_site_rot_vs_eps_t}
\end{figure*}

\begin{table}
% \begin{ruledtabular}
\centering
\begin{tabular}{ccc}
Operator &\makecell{Complexity\\from Ref.~\cite{Hardy:2024ric}} & \makecell{Complexity\\from QSVT}\\
% \hline
\hline
$\hat{\varphi}_i\hat{\varphi}_j$ & $\OO(\nq)$ & $\OO(\nq\log \nq)$\\
\hline
$\hat{\varphi}^2$ & $\OO(\nq^2)$ & $\OO\left(\nq\log \nq\right)$\\
\hline
$\hat{\pi}^2$ & $\OO(\nq^2)$ & $\OO\left(\nq\log \nq\right)$\\
\hline
$\hat{\varphi}^4$ & $\OO(\nq^2)$ & $\OO\left(\nq\log \nq\right)$
\end{tabular}
% \end{ruledtabular}
\caption{Comparison of complexities to block-encode terms for $\NLat = \OO(1)$ sites in scalar field theory Hamiltonian. The second column are complexities based on \cite{Hardy:2024ric} which uses a ``comparator'' operator to reduce the complexity of LCU operations; The third column is the relevant QSVT method described in~\cref{sec:be} (taking polynomials of degree $d \le 4$).}\label{table:hardy_result}
\end{table}

\section{Discussion and Conclusion\label{sec:discussion}}

In this work, we presented several new methods for preparing BEs for particular formulations of bosonic lattice field theories.
The methods presented in this work relied on two properties of such digitizations, namely the simplicity of the local operators $\hat\varphi, \hat\pi^{(D)}$, and the fact that only a small number of qubits per site $\nq$ are required to achieve the digitization errors necessary for realistic simulations.

Exploiting the simplicity of local operators $\hat\varphi$ and $\hat\pi^{(D)}$, we showed that QSVT can prepare BEs of degree $d$ functions of local terms in the Hamiltonian using $\OO(d \nq \log \nq)$ gates.
For polynomial functions with fixed degree $d$ (relevant for simulating $\hat{\varphi}^4$ scalar field theories), this QSVT based method scales asymptotically as $\OO(\nq \log \nq)$ and improves upon the methods in Ref.~\cite{Hardy:2024ric} (see~\cref{table:hardy_result} for a comparison), and is identical to the scaling of preparing BEs of operators of the form $\sim(\hat\varphi_i-\hat\varphi_j)^2$ in Ref.~\cite{Kharazi:2024fit}.

Next, we considered construction of BEs with the aid of the QETU algorithm.
While the simple operator $\me^{-i \tau \hat\varphi}$ is a natural building block for QETU circuits, because QETU generally requires approximating functions with discontinuous derivatives, we found that high degree polynomials are required to prepare BEs to a high precision.
This poor convergence was overcome by exploiting the fact that the spectrum of the local operators are known exactly, and one can use QETU to prepare exact BEs by reproducing the function only at those values, which requires $\OO(\nq 2^{\nq})$ gates.
Despite this technically inefficient scaling with $\nq$, through explicit circuit constructions, we found that, due to QSVT having a relatively large prefactor in the gate cost, QETU outperforms QSVT for small values of~$\nq$.
The comparison of different BEs considered in this work is summarized in~\cref{table:local_BEs}.

Motivated by the observation that algorithms with inefficient scaling in $\nq$ can require the fewest gates, we developed the conceptually simple \LOVELCU approach.
\LOVELCU can prepare BEs of arbitrary Hermitian operators $\hat{A}$ acting on $\nq$ qubits using only two controlled calls to $\me^{-i \arccos(\hat{A})}$, requiring a single ancillary qubit and $\OO(2^{\nq})$ ($\OO(4^{\nq})$) gates for diagonal (non-diagonal) operators.
Our numerical investigations demonstrated that \LOVELCU outperforms all other methods for preparing BEs of operators acting on $\lesssim 11$ qubits, at which point QSVT wins due to its relative exponentially improved gate scaling, albeit requiring more ancillary qubits than \LOVELCU.

Our findings indicate that, while understanding asymptotic gate complexities can serve as a useful guide, explicit numerical investigations are essential to compare different BE techniques for a specific application as their performance depends heavily on the considered Hamiltonian, problem size, and desired precision. We note that one can view the different approaches to BEs considered in this work as belonging to the same family of QSP-based methods but with different building blocks, which thus change the functions of interest and their implementations. We observe that methods utilizing asymptotically inefficient building blocks can potentially outperform the asymptotically optimal ones in practically relevant regimes.

In addition to developing new methods for constructing BEs using these techniques, we also developed methods for efficiently constructing the walk operator $W_{\hat{H}}$, required for using QSP-based methods for Hamiltonian simulation.
In particular, when BEs of local terms are constructed using any of the new methods presented in this work, and the full Hamiltonian is then constructed using a final LCU procedure to add these local BEs, the walk operator can be constructed as in~\cref{eq:whrsu} where $S$ is a single $\Zgate$-gate.
Using GQSP combined with \LOVELCU, we found that, for a single site anharmonic oscillator, GQSP outperforms PF methods for errors as small as $\epsilon \sim 10^{-2}$, which is $\sim 5$ orders of magnitude improved relative to previous work~\cite{Hariprakash:2023tla}.
Note that because PF methods have better asymptotic scaling with the number of lattice sites than GQSP~\cite{Hariprakash:2023tla}, the error threshold at which GQSP outperforms PFs is expected to be lower for larger lattice sizes.
To better understand how the relative cost changes with the number of lattice sites, we studied the two-site system and found that GQSP now outperforms PF methods for errors $\epsilon \sim 10^{-7}$.
This dramatic change, however, likely stems (at least in part) from using a compiler that is agnostic to the structure of the system studied, emphasizing the need for system specific circuit optimizations.

While we performed numerical gate count studies for a scalar field theory Hamiltonian, our methods can be directly applied to other bosonic theories formulated in terms of conjugate operators, including compact formulations of U(1) LGTs~\cite{Bender:2020ztu, Bauer:2021gek} and the mixed-basis formulation of SU(2) LGT in Ref.~\cite{DAndrea:2023qnr}. 
Our methods can also be applied to LGTs that are not formulated in terms of conjugate variables, including the standard Kogut-Susskind Hamiltonian and Loop-String-Hadron formulations~\cite{Raychowdhury:2018osk, Raychowdhury:2019iki, Kadam:2022ipf, Kadam:2024zkj}.
Notably, this includes situations in which the Hamiltonian operator is non-local, i.e., contains an exponential number of Pauli operators.
For example, the Hamiltonian operator of the gauge-fixed U(1) LGT~\cite{Haase:2020kaj, Bender:2020ztu, Bauer:2021gek, Grabowska:2022uos, Kane:2022ejm} includes a term of the form $\cos\bigl(\sum_j^N B_j\bigr)$, where the magnetic field operator at site $j$, $B_j$, is analogous to the field operator $\hat \varphi$.
While this operator is a sum of a number of Pauli strings exponential in the number of lattice sites, it can be readily encoded via the \LOVELCU construction in~\cref{fig:lcu_arccos} with a cost linear in the number of sites; this exponential reduction in cost allows one to avoid the usage of costly sparse oracle routines.
More generally, the methods discussed in this work could be useful for block encoding Hamiltonians involving operator functions of multiple variables.
One potential application is the first-quantized chemistry Hamiltonians, which can expressed in terms of position and derivative operators, such as $1/r$ and $\partial_x^2$, acting on fermionic wavefunctions.

Our work leads naturally to several interesting research directions.
Regarding using QETU with $\me^{-i \tau \hat\varphi_{\rm sh}}$ as a building block, it was found in Ref.~\cite{Kane:2023jdo} that limiting the spectrum of $\hat\varphi_{\rm sh}$ to be in the range $[\eta, \pi-\eta]$ and varying $\eta$ allowed one to approximate Gaussian functions $\sim e^{-\hat\varphi^2}$ to a precision $\epsilon$ using only $\log(\nq \log(1/\epsilon))$ gates.
The same is not true, however, for polynomial operators of the form $\hat\varphi^n$ considered in this work; as $\eta \to \pi/2$ the scale factor of the BE goes to zero, leading to a significant increase in the cost of Hamiltonian simulation.
It would be interesting to see if the method of varying $\eta$ can achieve a $\OO(\nq \log(1/\epsilon))$ scaling for other functions that do not diverge for large argument, some relevant physical examples being central potentials of the form $1/r$, or $1/r^6$ and $1/r^{12}$ which appear in the Lennard-Jones potential, see, e.g., Ref.~\cite{Kharazi:2024fit}.

% {\color{red}Write a paragraph about applications of our techniques to other models.
% }

% In this work the developed approaches to block encoding have been exclusively applied to construction of block encodings in scalar bosonic field theories.
% We now discuss several generalizartions

With an eye towards fault tolerance, it is essential to compile the circuits down to metrics suitable for the error correction protocols (e.g., T-gate counts) to assess the relative cost of the methods in this work.
The QETU-based methods and \LOVELCU inherently require the use of rotation gates, which implies they likely come with a large prefactor in the T-gate count.
These methods, however, may serve useful for partial-fault tolerant implementations where only Clifford gates are implemented in a fault-tolerant manner, and rotation gates are implemented in a non-fault-tolerant way~\cite{PRXQuantum.5.010337}.
The QSVT-based methods are generally expected to require less T-gates, as one can use unary-LCU methods to reduce the cost of the $\SELECT$ oracle, and the methods in Ref.~\cite{Hardy:2024ric} to construct a BE of $\hat\varphi$ avoiding the need for T-gates in the $\PREPARE$ oracle. 
It would be interesting to do a dedicated T-gate study using a combination the QSVT method in this work and the efficient T-gate constructions in Refs.~\cite{Babbush:2018ywg, Hardy:2024ric}.

While the total T-gate count is usually considered for the total cost of a fault-tolerant simulation (see, e.g., Refs.~\cite{Davoudi:2022xmb,Rhodes:2024zbr}), there have been suggestions that another possible metric is instead T-gate \emph{depth}~\cite{10.1109/TCAD.2013.2244643}.
This is similar in spirit to Amdahl's Law in classical computing, which states that the maximum possible speedup from parallelizing one's code is limited by components of the code that must be executed sequentially.
In other words, Amdahl's Law describes the time a code takes to run in the limit of access to infinitely many classical computing nodes (barring communication time).
By asking the same question in the context of using several fault-tolerant quantum computing nodes, T-gate depth is the metric for how parallelizable a circuit is.
In scenarios where T-gate depth is the relevant cost, reductions in the depth of preparing BEs can be achieved using the modified LCU procedure in Ref.~\cite{Boyd:2023mph}.
Using this or similar procedures, it is possible to imagine that the depth of constructing a BE of geometrically local Hamiltonians can be reduced to constant in the number of lattice sites.
Such a method would result in the depth of QSP-based time evolution methods scaling linearly in the volume, which outperforms the volume scaling of the depth of PF based methods for geometrically local systems~\cite{Childs:2019hts}.

We conclude by pointing out that, while further improvements can likely be made within the paradigm of preparing BEs via QSP-based methods, dramatic reductions in the gate count will likely require considering new paradigms. Although our focus has been on constructing BEs using LCU for QSP-based time evolution, which typically requires long coherent circuits and full fault tolerance, our methods can be adapted for simulations relying solely on LCU with approaches like the Truncated Taylor Series~\cite{Berry:2014ivo}. To make these simulations more suitable for noisy intermediate-scale quantum (NISQ) devices, one promising direction is to explore recently developed randomized near-term implementations of LCU~\cite{Zeng:2022pim,Chakraborty:2023vtr}, which could be combined with our constructions to enable more practical quantum simulations in the near term. Another promising avenue is to combine Trotter-based methods with near-optimal time evolution techniques~\cite{Faehrmann:2021mqc,Zeng:2022pim,Chakraborty:2025sry}, potentially leveraging the strengths of both approaches for improved efficiency. Additionally, the developing field of multi-variable QSP methods~\cite{Rossi:2022sfw, Rossi:2023jgh, Mori:2023fat, Nemeth:2023enj, Laneve:2024kwq, Gomes:2024tup} (this general class of methods includes multi-variable QETU), which leverages the structure of lattice field theories written as sums of many commuting variables, could result in significant cost reductions.

\section*{Acknowledgements}

SH acknowledges support from the Berkeley Center for Theoretical Physics and the National Energy Research Scientific Computing Center (NERSC), a U.S. Department of Energy Office of Science User Facility located at Lawrence Berkeley National Laboratory, operated under Contract No. DE-AC02- 05CH11231.
CWB and MK were supported by the DOE, Office of Science under contract DE-AC02-05CH11231, partially through Quantum Information Science Enabled Discovery (QuantISED) for High Energy Physics (KA2401032). 
This material is based upon work supported by the U.S. Department of
Energy, Office of Science, Office of Advanced Scientific Computing Research, Department of Energy Computational Science Graduate Fellowship under Award Number DE-SC0020347.

\bibliographystyle{plainnat}
\bibliography{Quantum_Journal/main_quantum_journal_final}

\clearpage
\newpage
\onecolumngrid

\appendix

\section{Constructing the walk operator}
\label{app:WH}

In this appendix, we provide detailed proofs for the selection of the operator $S$ (see \cref{sec:be}) used in constructing the walk operator when the associated BE is implemented using QSVT (\cref{lemma:QSVT_walk_op}) or QETU (\cref{lemma:QETU_walk_op}) for even polynomials. Specifically, we demonstrate that for both QSVT and QETU, the operator $S$ can be chosen as a single $\Zgate$-gate acting on the signal and control qubit, respectively, to satisfy the conditions in~\cref{eq:S_cond}.
We then consider the scenario where the BE for the full Hamiltonian $U_{\hat{H}}$ is constructed using LCU to combine several local $U_{\hat{H}_j}$ BEs, with each $U_{\hat{H}_j}$ constructed using either QSVT, QETU, or \LOVELCU. We demonstrate that, as long as the circuit organization discussed in detail in~\cref{ssec:UH_WH_general} is used, the associated walk operator $W_{\hat{H}}$ can be constructed by choosing $S$ to be a single $\Zgate$-gate.

The following Lemma shows how to construct the walk operator given that that the BE was constructed using QSVT for even polynomials.

\begin{lemma}
    \label{lemma:QSVT_walk_op}
    Let $\hat{A}$ and $f(\hat{A})$ be Hermitian operators acting on the Hilbert space $\mathcal{H}_s$, and, without loss of generality, assume $\|\hat{A}\| \leq 1$. Let $U_A$ be a $(1, m)$ block encoding of $\hat{A}$, acting on the Hilbert space $\mathcal{H}_a \otimes \mathcal{H}_s$, where $\mathcal{H}_a$ is an $m$-qubit ancillary register. Let $U_{f(\hat{A})}$ be an $(\alpha, m+1)$ block encoding of $f(\hat{A})$, constructed using the even QSVT circuit (see~\cref{fig:qsvt_circ}) with the building block $U_A$, acting on the Hilbert space $\mathcal{H}_c \otimes \mathcal{H}_a \otimes \mathcal{H}_s$, where $\mathcal{H}_c$ is the single-qubit signal qubit. Then, the operator $S' = Z_c \otimes \unit_a \otimes \unit_s$ satisfies
    \begin{align}
        \left(\bra{0}_c \otimes \bra{0}_a \otimes \unit_s\right) S' U_{f(\hat{A})} \left(\ket{0}_c \otimes \ket{0}_a \otimes \unit_s\right) &= f(\hat{A})/\alpha\,,
        \\
        \label{eq:a3}
        \left(\bra{0}_c \otimes \bra{0}_a \otimes \unit_s\right) S' U_{f(\hat{A})} S' U_{f(\hat{A})} \left(\ket{0}_c \otimes \ket{0}_a \otimes \unit_s\right) &= \unit_s\,.
    \end{align}
\end{lemma}

\begin{proof}
Using the fact that $S' \ket{0}_c \otimes \unit_{as} = \ket{0}_c \otimes \unit_{as}$, we see that
\begin{equation}
    \left(\bra{0}_c \otimes \bra{0}_a \otimes \unit_s\right) S' U_{f(\hat{A})} \left(\ket{0}_c \otimes \ket{0}_a \otimes \unit_s\right) = \left(\bra{0}_c \otimes \bra{0}_a \otimes \unit_s\right) U_{f(\hat{A})} \left(\ket{0}_c \otimes \ket{0}_a \otimes \unit_s\right) = f(\hat{A})/\alpha,
\end{equation}
where the final equality in the above equation is the definition of the block encoding $U_{f(\hat{A})}$. The first relation is therefore satisfied.

To prove the second relation is satisfied, we show that $S' U_{f(\hat{A})} = U_{f(\hat{A})}^\dagger S'$.
To see this, we use two circuit identities: $Z H = H X$ and $X_c \otimes \unit_a \text{CR}_{\tilde{\phi}} = \text{CR}_{\tilde{\phi}}^\dagger X_c \otimes \unit_a$ (see Fig.~\ref{fig:crphi_circ} for the circuit for $\text{CR}_{\tilde\phi}$). The latter follows from the facts that $X$ gates commute with the target of multi-controlled NOT gate and $X R_z(\theta) = R_z(\theta)^\dagger X$.
Repeated use of these identities imply that commuting the $Z$ gate in $S'$ past the QSVT circuit flips the sign of the phases in each $\text{CR}_{\tilde{\phi}}$.
However, because the phases $\tilde{\phi}$ are symmetric, combined with the fact that one alternates calls to $U_A$ and $U_A^\dagger$, flipping sign of the phases in each $\text{CR}_{\tilde{\phi}}$ is equivalent to taking the Hermitian conjugate and thus $S' U_{f(\hat{A})} = U_{f(\hat{A})}^\dagger S'$. This implies that $S' U_{f(\hat{A})} S' U_{f(\hat{A})} = S' U_{f(\hat{A})} U_{f(\hat{A})}^\dagger S' = (S')^2 = \unit_{cas}$ and so the second relation is satisfied.

\end{proof}

The following Lemma shows how to construct the walk operator given that that the BE was constructed using QETU for even polynomials.

\begin{lemma}
    \label{lemma:QETU_walk_op}
    Let $\hat{A}$ and $f(\hat{A})$ be Hermitian operators acting on the Hilbert space $\mathcal{H}_s$. Let $U_{f(\hat{A})}$ be an $(\alpha, 1)$ block-encoding of $f(\hat{A})$, constructed using the QETU circuit for even polynomials (see Fig.~\ref{fig:qetu_circ}) with the building block $\me^{-i \tau \hat{A}}$, acting on the Hilbert space $\mathcal{H}_c \otimes \mathcal{H}_s$ where $\mathcal{H}_c$ is the Hilbert space of the control qubit in the QETU circuit.
    Then, the operator $S' = Z_c \otimes \unit_s$ satisfies
    \begin{align}
    \left(\bra{0}_c \otimes \unit_s\right) S' U_{f(\hat{A})} \left(\ket{0}_c \otimes \unit_s\right) &= f(\hat{A})\,,
    \\
    \left(\bra{0}_c \otimes \unit_s\right) S' U_{f(\hat{A})} S' U_{f(\hat{A})} \left(\ket{0}_c \otimes \unit_s\right) &= \unit_s\,.
    \end{align}
\end{lemma}
\begin{proof}

Using the fact that $S' \ket{0}_c \otimes \unit_s =\ket{0}_c \otimes \unit_s$, we see that
\begin{equation}
    \left(\bra{0}_c \otimes \unit_s\right) S' U_{f(\hat{A})} \left(\ket{0}_c \otimes \unit_s\right) = \left(\bra{0}_c \otimes \unit_s\right) U_{f(\hat{A})} \left(\ket{0}_c \otimes \unit_s\right) = f(\hat{A})\,,
\end{equation}
where the final equality in the equation above is the definition of the block encoding $U_{\hat{H}}$. The first relation is therefore satisfied.

To show that the second relation is satisfied, we will demonstrate that $S' U_{f(\hat{A})} = U_{f(\hat{A})}^\dagger S'$. We denote by $C_U$ the controlled call to the building block $\me^{-i \tau \hat{A}}$. Note that, due to the structure of alternating calls to $C_U$ and $C_U^\dagger$, combined with the fact that the rotation angles $\{\tilde{\varphi}_j\}$ are symmetric, the circuit for $U_{f(\hat{A})}^\dagger$ is equivalent to negating the phases in the single-qubit $R_x$ gates in the original QETU circuit. We will now show that commuting $Z_c \otimes \unit_s$ past $U_{\hat{H}}$ has the effect of negating the signs of the phases in the $R_x$ gates. We note the following circuit identities: commuting a $Z$ gate past an $R_x(\theta)$ gate flips the sign of the rotation angle, \emph{i.e.}, $Z e^{i \theta X} = e^{-i \theta X} Z$, and $Z$ gates acting on the control qubit commute with controlled calls to $e^{\pm i \tau \hat{A}}$.
This can be seen by noting that diagonal matrices commute, which implies $[Z, \sket{0}\sbra{0}] = [Z, \sket{1}\sbra{1}] = 0$.
These identities imply that commuting the $Z$ gate past the QETU circuit flips the sign of all rotation angles, and therefore $S' U_{f(\hat{A})} = U_{f(\hat{A})}^\dagger S'$.
Using this result, we see that $S' U_{f(\hat{A})} S' U_{f(\hat{A})} = S' U_{f(\hat{A})} U_{f(\hat{A})}^\dagger S' = (S')^2 = \unit_{cs}$, and so the second relation is satisfied.
\end{proof}

The following lemma shows how to construct the walk operator $W_{\hat{H}}$ for the full Hamiltonian from $U_{\hat{H}}$ when $U_{\hat{H}}$ is constructed using LCU to combine local $U_{\hat{H}_j}$'s, where each $U_{\hat{H}_j}$ shares a common $S$ operator for constructing $W_{\hat{H}_j}$.

\begin{lemma}
\label{lemma:LCU_WH}
    Let $\hat{H} = \sum_{j=0}^{M-1} \beta_j \hat{H}_j$ be a Hamiltonian acting on a Hilbert space $\mathcal{H}_s$ with $\beta_j \in \mathbb{R}^+$.
    Without loss of generality, assume $\|\hat{H}_j\| \leq 1, \forall j$.
    Let $U_{\hat{H}_j}$ be a $(1, m_j)$ block-encoding of $\hat{H}_j$, acting on the joint Hilbert space $\mathcal{H}_a \otimes \mathcal{H}_s$, where $\mathcal{H}_a$ is the Hilbert space of ancillary qubit register containing $\max_j(m_j)$ qubits. 
    Furthermore, let $U_{\hat{H}} = (\PREPARE^\dagger \otimes \unit_{as}) \SELECT (\PREPARE \otimes \unit_{as})$ be a $\left(\|\vec{\beta}\|_1, \lceil \log_2 M \rceil + \max_j(m_j)\right)$ block-encoding of $\hat{H}$ acting on the joint Hilbert space $\mathcal{H}_{a_\Lambda} \otimes \mathcal{H}_a \otimes \mathcal{H}$, where $\mathcal{H}_{a_\Lambda}$ is the Hilbert space of the ancillary register containing $ \lceil \log_2 M \rceil$ qubits, $\|\vec{\beta}\|_1 \equiv \sum_{j=0}^{M-1}|\beta_j|$, $\SELECT = \sum_{j=0}^{M-1} (\sket{j} \sbra{j})_{a_\Lambda} \otimes U_{\hat{H}_j}$ and $\PREPARE\sket{0}_{a_\Lambda} = \frac{1}{\sqrt{\|\vec{\beta}\|_1}} \sum_{j=0}^{M-1} \sqrt{\beta_j} \sket{j}_{a_\Lambda}$.
    Define $S' = \unit_{a_\Lambda} \otimes S_a \otimes \unit_s$.
    If for each $U_{\hat{H}_j}$ the operator $S_a$ satisfies
    \begin{align}
        \left(\sbra{0}_a \otimes \unit_s \right) (S_a \otimes \unit_s) U_{\hat{H}_j} \left(\sket{0}_a \otimes \unit_s \right) &= \hat{H}_j\,,
        \\
        \left(\sbra{0}_a \otimes \unit_s \right) \left[S_a \otimes \unit_s) U_{\hat{H}_j}\right]^2 \left(\sket{0}_a \otimes \unit_s \right) &= \unit_s\,,
    \end{align}
    then
    \begin{align}
        \left(\sbra{0}_{a_\Lambda} \otimes \sbra{0}_a \otimes \unit_{s}\right) S' U_{\hat{H}} \left(\sket{0}_{a_\Lambda} \otimes \ket{0}_a \otimes \unit_{s}\right) &= \frac{\hat{H}}{\|\vec{\beta}\|_1}\,, 
        \\
        \left(\sbra{0}_{a_\Lambda} \otimes \bra{0}_a \otimes \unit_{s}\right) S' U_{\hat{H}} S' U_{\hat{H}} \left(\sket{0}_{a_\Lambda} \otimes \ket{0}_a \otimes \unit_{s}\right) &= \unit_s\,.
    \end{align}
\end{lemma}
\begin{proof}

First, note that 
\begin{equation}
\label{eq:suh11}
\begin{split}
    S' U_{\hat{H}} &= \left(\unit_{a_\Lambda} \otimes S_a \otimes \unit_s \right) (\PREPARE^\dagger \otimes \unit_{as}) \left(\sum_{j=0}^{M-1} (\sket{j} \sbra{j})_{a_\Lambda} \otimes U_{\hat{H}_j} \right) \left(\PREPARE \otimes \unit_{as} \right) 
    \\
    &= \sum_{j=0}^{M-1} \left[\PREPARE^\dagger(\sket{j} \sbra{j})_{a_\Lambda}\PREPARE \right] \otimes \left[(S_a \otimes \unit_s)U_{\hat{H}_j}\right].
\end{split}
\end{equation}
From~\cref{eq:suh11} it immediately follows that the first relation is satisfied:
\begin{equation}
\begin{split}
    (\sbra{0}_{a_\Lambda}& \otimes \sbra{0}_a \otimes \unit_{s}) S' U_{\hat{H}} \left(\sket{0}_{a_\Lambda} \otimes \ket{0}_a \otimes \unit_{s}\right)
    \\
    &= \left(\sbra{0}_{a_\Lambda} \otimes \sbra{0}_a \otimes \unit_{s}\right)\sum_{j=0}^{M-1} \left[\PREPARE^\dagger(\sket{j} \sbra{j})_{a_\Lambda}\PREPARE \right] \otimes \left[(S_a \otimes \unit_s)U_{\hat{H}_j}\right] \left(\sket{0}_{a_\Lambda} \otimes \ket{0}_a \otimes \unit_{s}\right)
    \\
    &= \sum_{j=0}^{M-1} \sbra{0}_{a_\Lambda}\left[\PREPARE^\dagger(\sket{j} \sbra{j})_{a_\Lambda}\PREPARE \right]\sket{0}_{a_\Lambda}  \otimes \left(\sbra{0}_a \otimes \unit_a\right) \left[(S_a \otimes \unit_s)U_{\hat{H}_j}\right]\left(\sket{0}_a \otimes \unit_a\right)
    \\
    &= \sum_{j=0}^{M-1} \sbra{0}_{a_\Lambda}\left[\PREPARE^\dagger(\sket{j} \sbra{j})_{a_\Lambda}\PREPARE \right]\sket{0}_{a_\Lambda}  \hat{H}_j 
    \\
    &= \frac{1}{\|\vec{\beta}\|_1}\sum_{j,k,l=0}^{M-1} \sqrt{\beta_k \beta_\ell} \langle k | j\rangle \langle j | l\rangle \hat{H}_j 
    \\
    &= \frac{1}{\|\vec{\beta}\|_1}\sum_{j=0}^{M-1}\beta_j \hat{H}_j 
    \\
    &= \frac{\hat{H}}{\|\vec{\beta}\|_1}\,.
\end{split}
\end{equation}

The second relation also follows in a straightforward way:
\begin{equation}
\begin{split}
    (&\sbra{0}_{a_\Lambda} \otimes \bra{0}_a  \otimes \unit_{s}) S' U_{\hat{H}} S' U_{\hat{H}} \left(\sket{0}_{a_\Lambda} \otimes \ket{0}_a \otimes \unit_{s}\right) 
    \\
    &= (\sbra{0}_{a_\Lambda} \otimes \bra{0}_a  \otimes \unit_{s}) \left[\sum_{j=0}^{M-1} \left[\PREPARE^\dagger(\sket{j} \sbra{j})_{a_\Lambda}\PREPARE \right] \otimes \left[(S_a \otimes \unit_s)U_{\hat{H}_j}\right]\right]^2\left(\sket{0}_{a_\Lambda} \otimes \ket{0}_a \otimes \unit_{s}\right)
    \\
    &= (\sbra{0}_{a_\Lambda} \otimes \bra{0}_a  \otimes \unit_{s}) \left[\sum_{j=0}^{M-1} \left[\PREPARE^\dagger(\sket{j} \sbra{j})_{a_\Lambda}\PREPARE \right] \otimes \left[(S_a \otimes \unit_s)U_{\hat{H}_j}\right]^2\right]\left(\sket{0}_{a_\Lambda} \otimes \ket{0}_a \otimes \unit_{s}\right)
    \\
    &= \sum_{j=0}^{M-1}\sbra{0}_{a_\Lambda}\PREPARE^\dagger(\sket{j} \sbra{j})_{a_\Lambda}\PREPARE\sket{0}_{a_\Lambda} \left(\bra{0}_a  \otimes \unit_{s}) \left[(S_a \otimes \unit_s )U_{\hat{H}_j}\right]^2\left( \ket{0}_a \otimes \unit_{s}\right)\right)
    \\
    &= \frac{1}{\|\vec{\beta}\|_1}\sum_{j,k,l=0}^{M-1} \sqrt{\beta_k \beta_\ell} \langle k | j\rangle \langle j | l\rangle
    \\
    &= \frac{1}{\|\vec{\beta}\|_1}\sum_{j=0}^{M-1}\beta_j 
    \\
    &=  \unit_s.
\end{split}
\end{equation}
\end{proof}

\section{CNOT gate counts for scalar field theory simulation \label{app:cnot_count}}

In this appendix, we present the associated CNOT gate counts from the numerical studies conducted in~\cref{sec:scalar_field_theory_numerics}. All results are qualitatively similar to the rotation gate counts discussed in the main text.

Figure~\ref{fig:be_cost_indiv_terms_cnot} shows the CNOT gate counts, associated with the rotation gate counts in Fig.~\ref{fig:be_cost_indiv_terms}, for BEs of the local terms $\frac{1}{2}\hat\pi^2$, $\frac{m}{2}\hat\varphi^2+\frac{\lambda}{4!}\hat\varphi^4$, $\frac{1}{2}(\hat\varphi_1-\hat\varphi_2)$, and $g\cos(\hat\varphi)$.

Figure~\ref{fig:cnot_be_Hphi_two_site} shows the CNOT gate counts, associated with the rotation gate counts in Fig.~\ref{fig:rot_anc_be_Hphi_two_site}, for BEs of the two-site $\hat\varphi$ Hamiltonian in Eq.~\eqref{eq:Hphi_two_site}.

Figure~\ref{fig:single_site_cnot_vs_eps_t} shows the CNOT gate counts, associated with the rotation gate counts in Fig.~\ref{fig:single_site_rot_vs_eps_t}, for simulating time evolution of the single site Hamiltonian in Eq.~\eqref{eq:H_single_site} using a 2nd order PF, a 4th order PF, and GQSP where $W_H$ was constructed using \LOVELCU.

Figure~\ref{fig:two_site_cnot_vs_eps_t_nq3} shows the CNOT gate counts, associated with the rotation gate counts in Fig.~\ref{fig:two_site_rot_vs_eps_t}, for simulating time evolution of the two-site Hamiltonian in Eq.~\eqref{eq:H_single_site} using a 2nd order PF, a 4th order PF, and GQSP where $W_H$ was constructed using \LOVELCU.

\begin{figure*}[h!]
    \centering
    \includegraphics[width=0.48\textwidth]{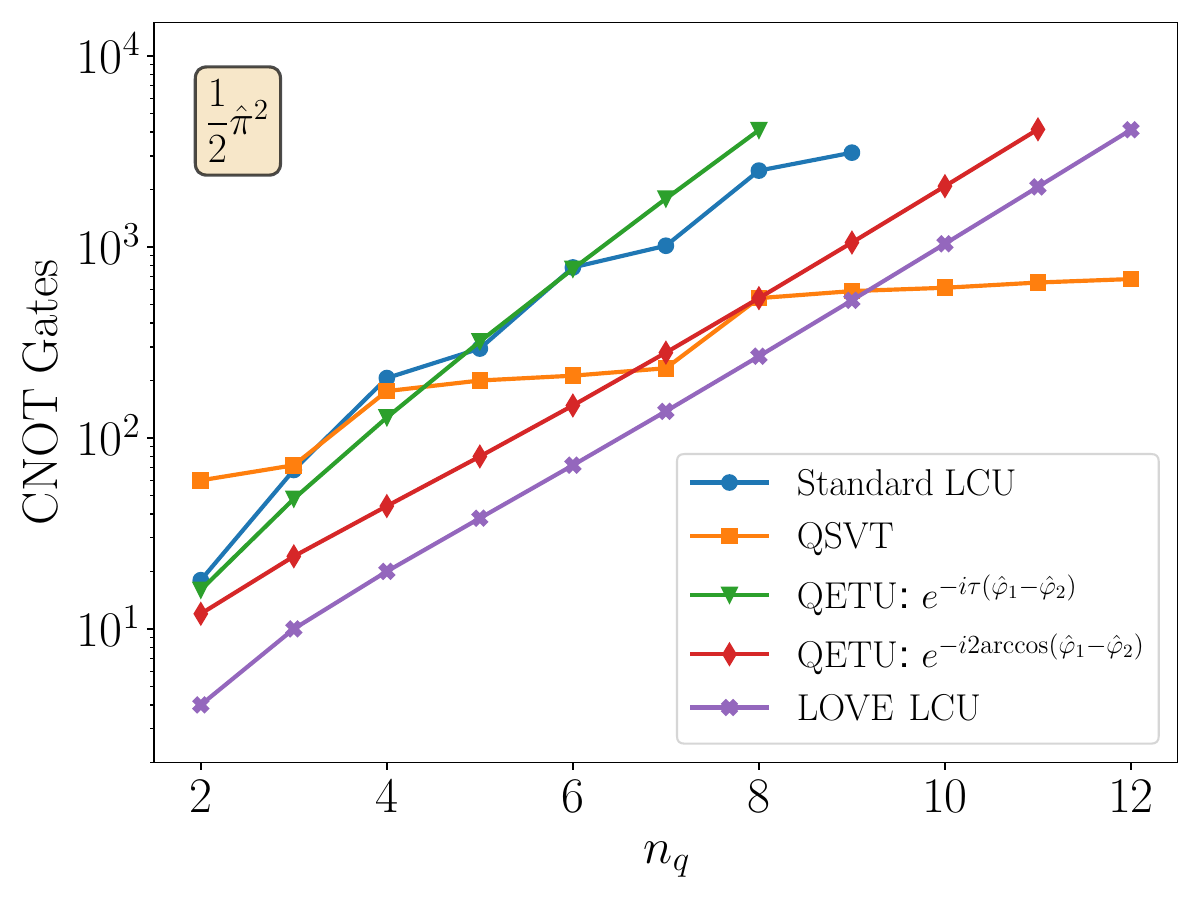}
    \includegraphics[width=0.48\textwidth]{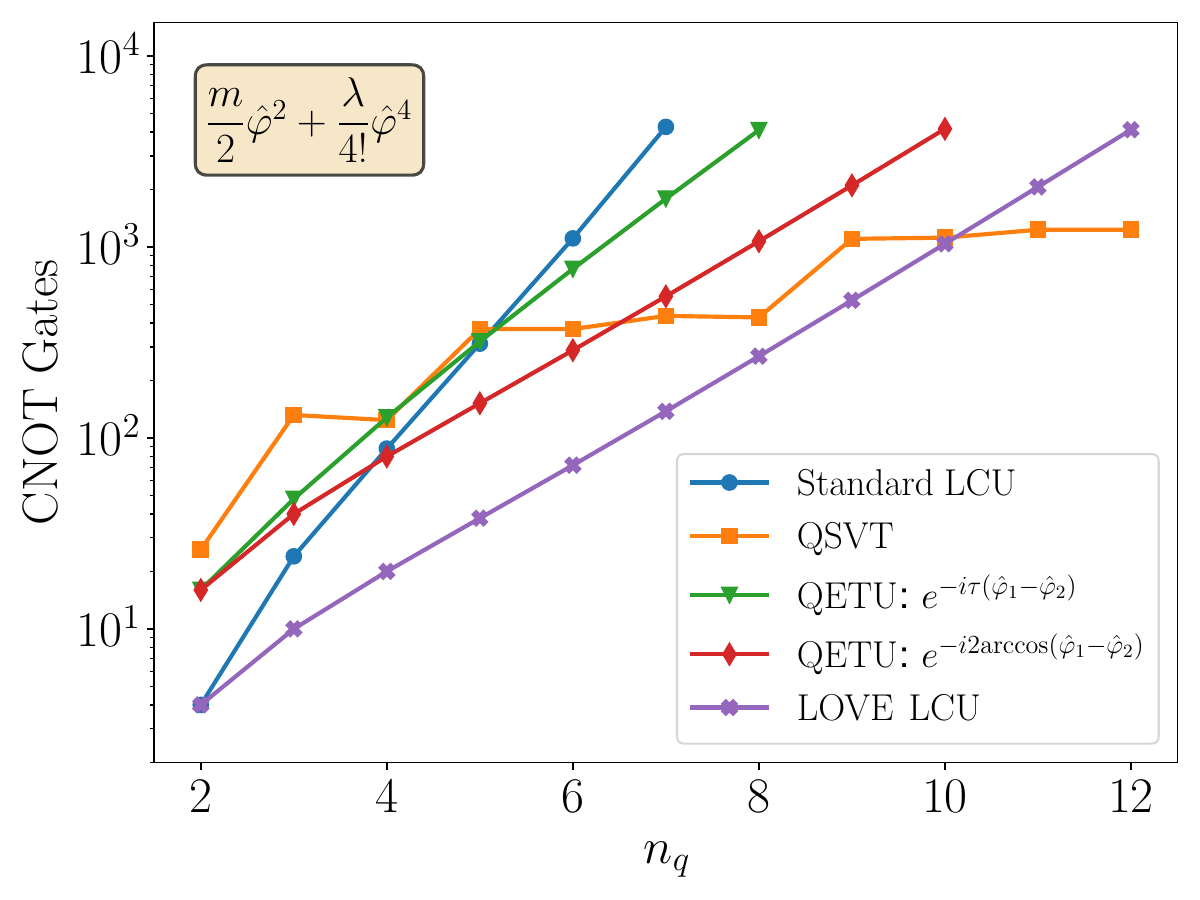}
    \includegraphics[width=0.48\textwidth]{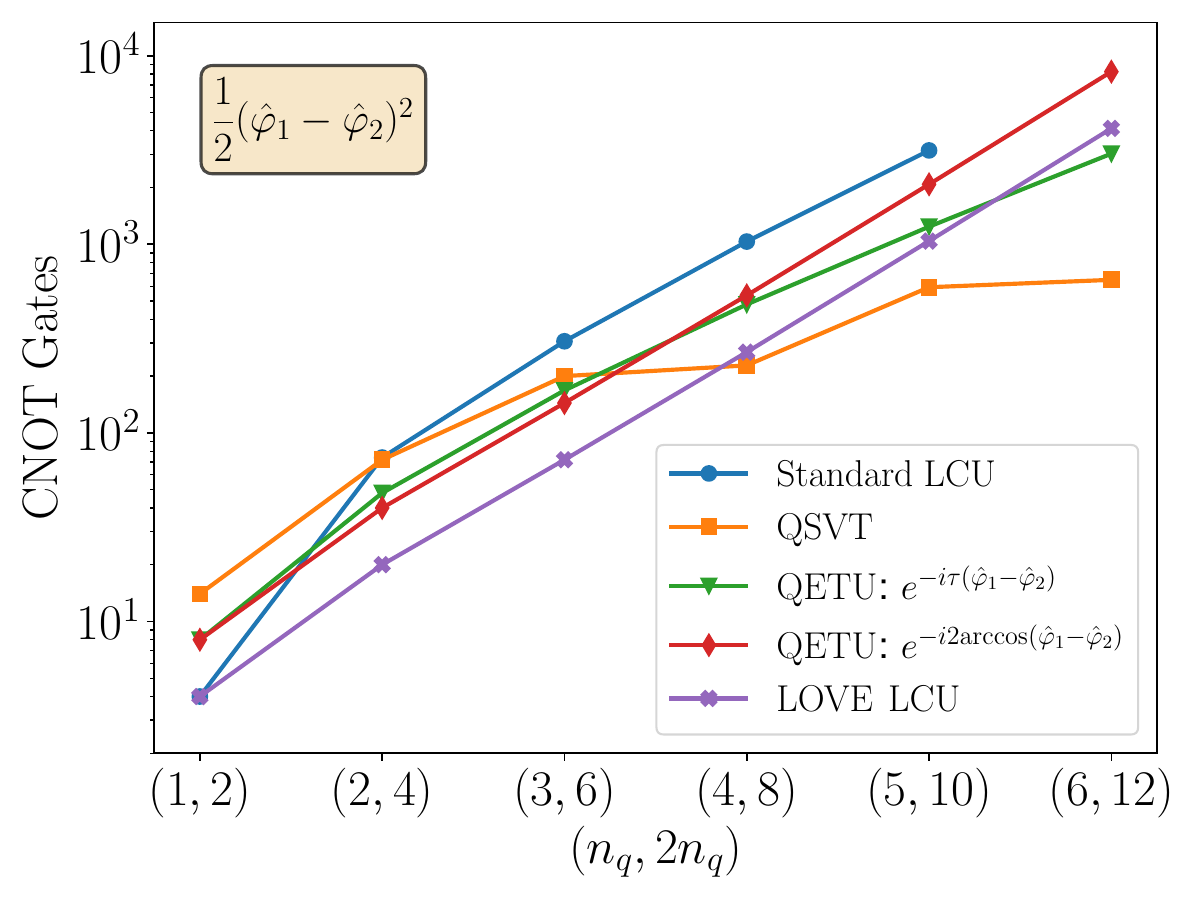}
    \includegraphics[width=0.48\textwidth]{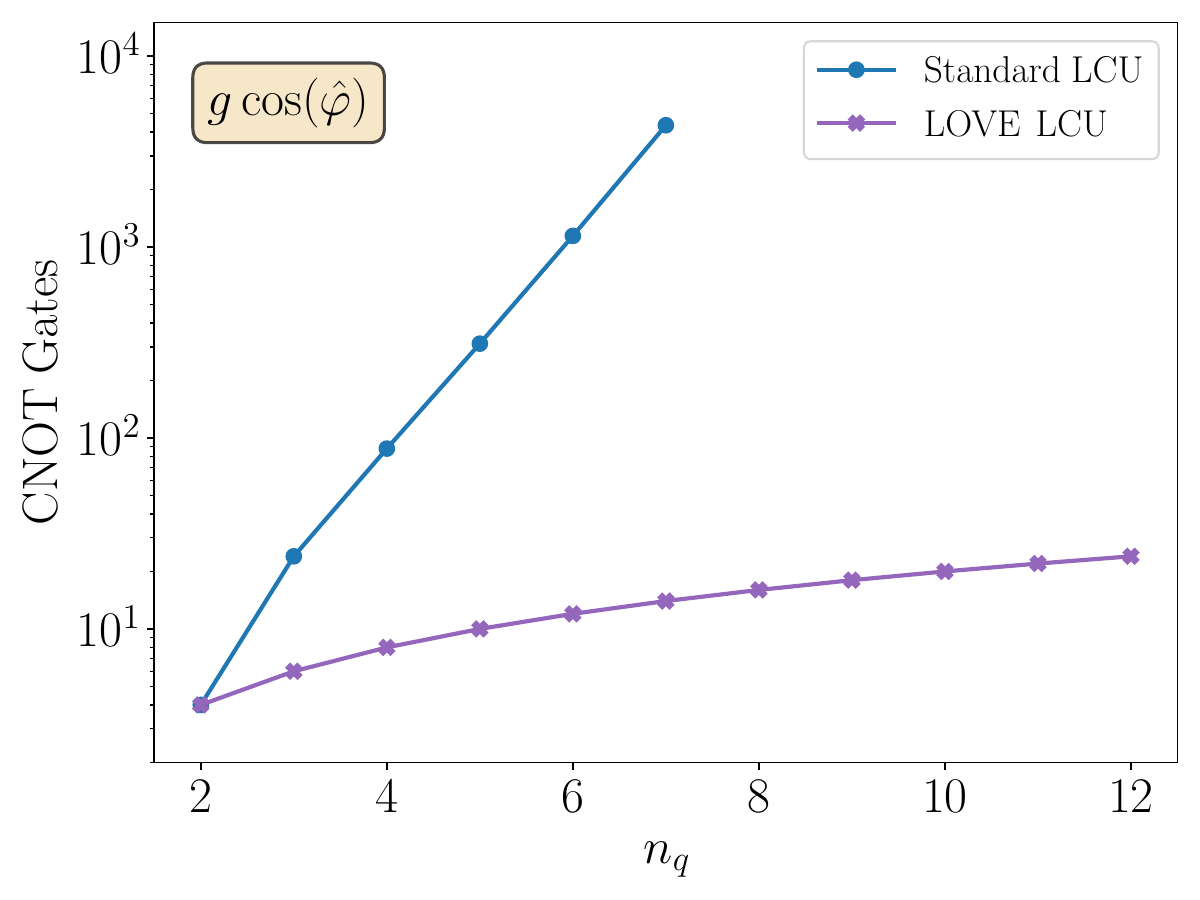}
    \caption{CNOT gate count and number of ancillary qubits required to block encode local bosonic operators. 
    The top left, top right, bottom left, and bottom right plots show resource requirements to block encode $\frac{1}{2}\hat\pi^2$, $\frac{m}{2}\hat\varphi^2+\frac{\lambda}{4!} \hat\varphi^4$, $\frac{1}{2}(\hat\varphi_1-\hat\varphi_2)^2, g\cos(\hat\varphi)$, respectively, for $m=1, \lambda=32, g=1$.
    Different colored and shaped data points correspond to different methods to prepare the BE.
    For the single-site operators $\hat\pi^2$ and $\frac{m}{2}\hat\varphi^2+\frac{\lambda}{4!} \hat\varphi^4$, \LOVELCU requires the fewest rotation gates for $\nq \lesssim 9$.
    For the two-site operator $\frac{1}{2}(\hat\varphi_1-\hat\varphi_2)^2$, \LOVELCU requires the fewest rotation gates for $\nq \lesssim 4$.
    \LOVELCU constructs a BE of $g\cos(\hat\varphi)$ using the fewest gates for all $\nq$.}
    \label{fig:be_cost_indiv_terms_cnot}
\end{figure*}

\begin{figure}[h!]
    \centering
    \includegraphics[width=0.5\textwidth]{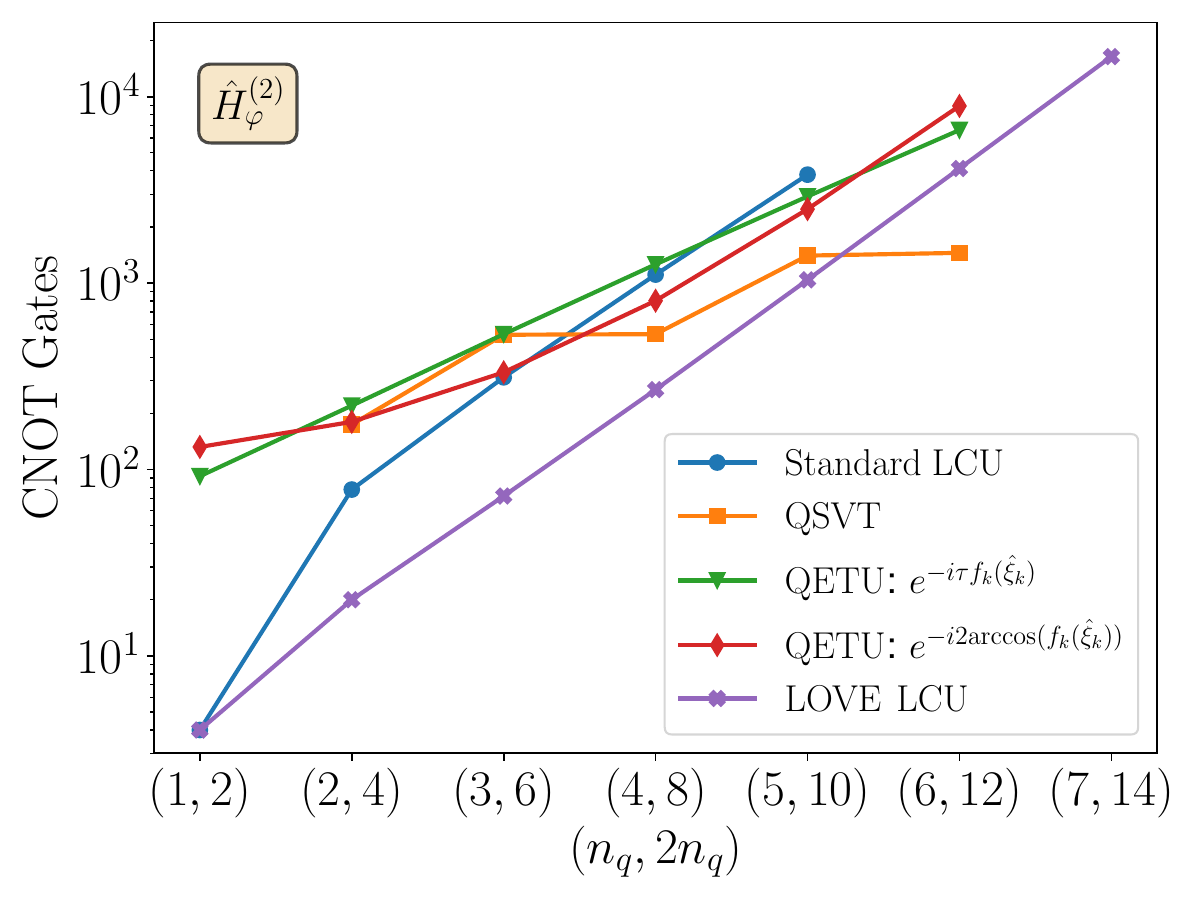}
    \caption{CNOT gate count required to block encode the two site Hamiltonian $\hat{H}^{(2)}_\varphi$ in Eq.~\eqref{eq:Hphi_two_site}.
    Different colored and shaped data points correspond to different methods to prepare the BE.
    Points with the same number of ancillary qubits have been shifted slightly for clarity.
    While QSVT- and QETU-based methods require using a final layer of LCU to add the local BEs $f_0^{(0)}(\hat\varphi_1), f_0^{(0)}(\hat\varphi_2), f_2(\hat\varphi_1-\hat\varphi_2)$, \LOVELCU does not.
    This advantage also leads to \LOVELCU outperforming all other methods for $\nq \lesssim 5$; for $\nq>5$, QSVT outperforms other methods due to the relative exponentially improved asymptotic scaling. }
    \label{fig:cnot_be_Hphi_two_site}
\end{figure}

\begin{figure*}[h!]
    \centering
    \includegraphics[width=0.6\textwidth]{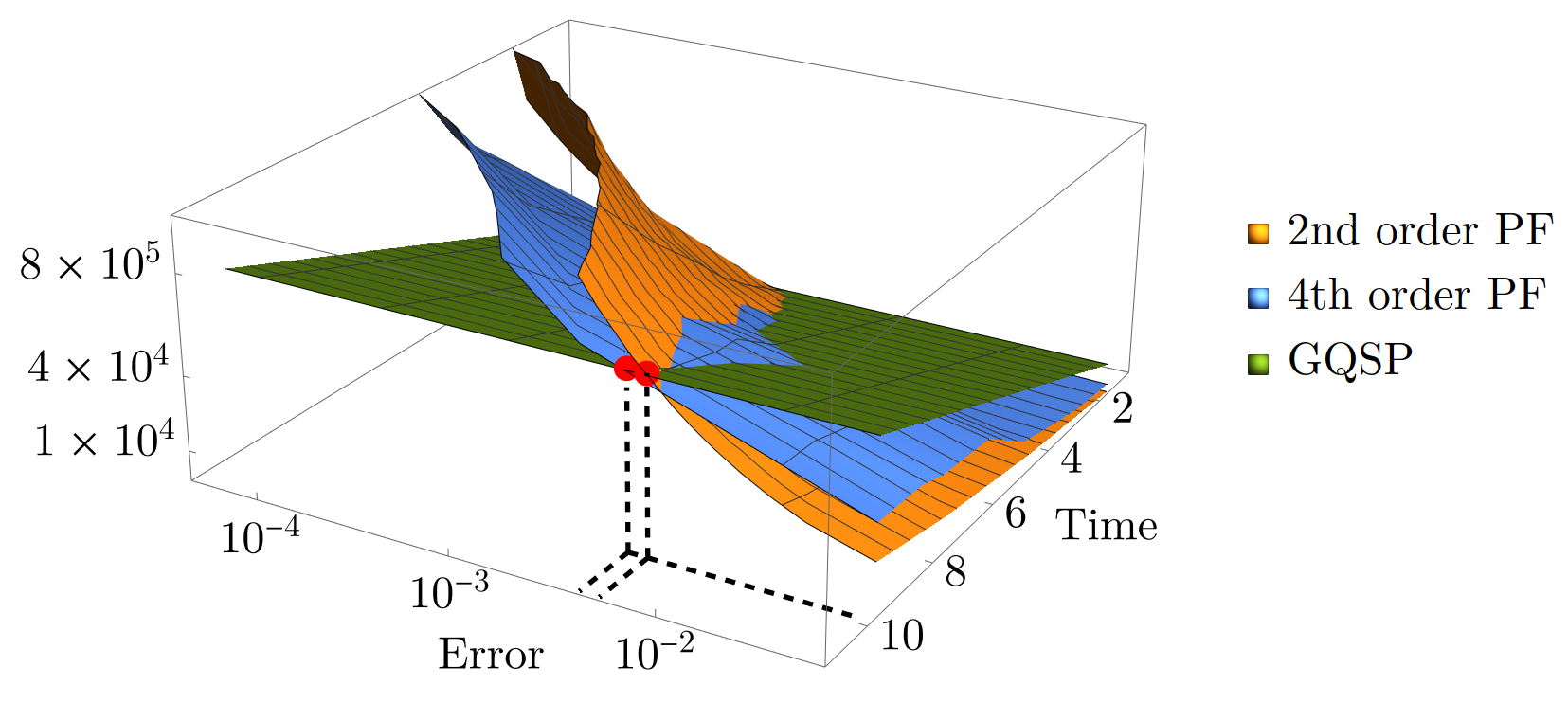}
    \caption{CNOT gate count as a function of error $\epsilon$ and simulation time $t$ for simulating time-evolution of the single-site and two-site Hamiltonian in Eq.~\eqref{eq:H_single_site}.
    The blue, orange, and green surfaces show results calculated using a $2^{\rm nd}$ order PF, a $4^{\rm th}$ order PF, and GQSP where the BE was constructing using \LOVELCU.
    For $t=10$, GQSP outperforms both $2^{\rm nd}$ and $4^{\rm th}$ order PFs for $\epsilon \lesssim 5\times 10^{-3}$.}
    \label{fig:single_site_cnot_vs_eps_t}
\end{figure*}

\begin{figure*}[h!]
    \centering
    \includegraphics[width=0.6\textwidth]{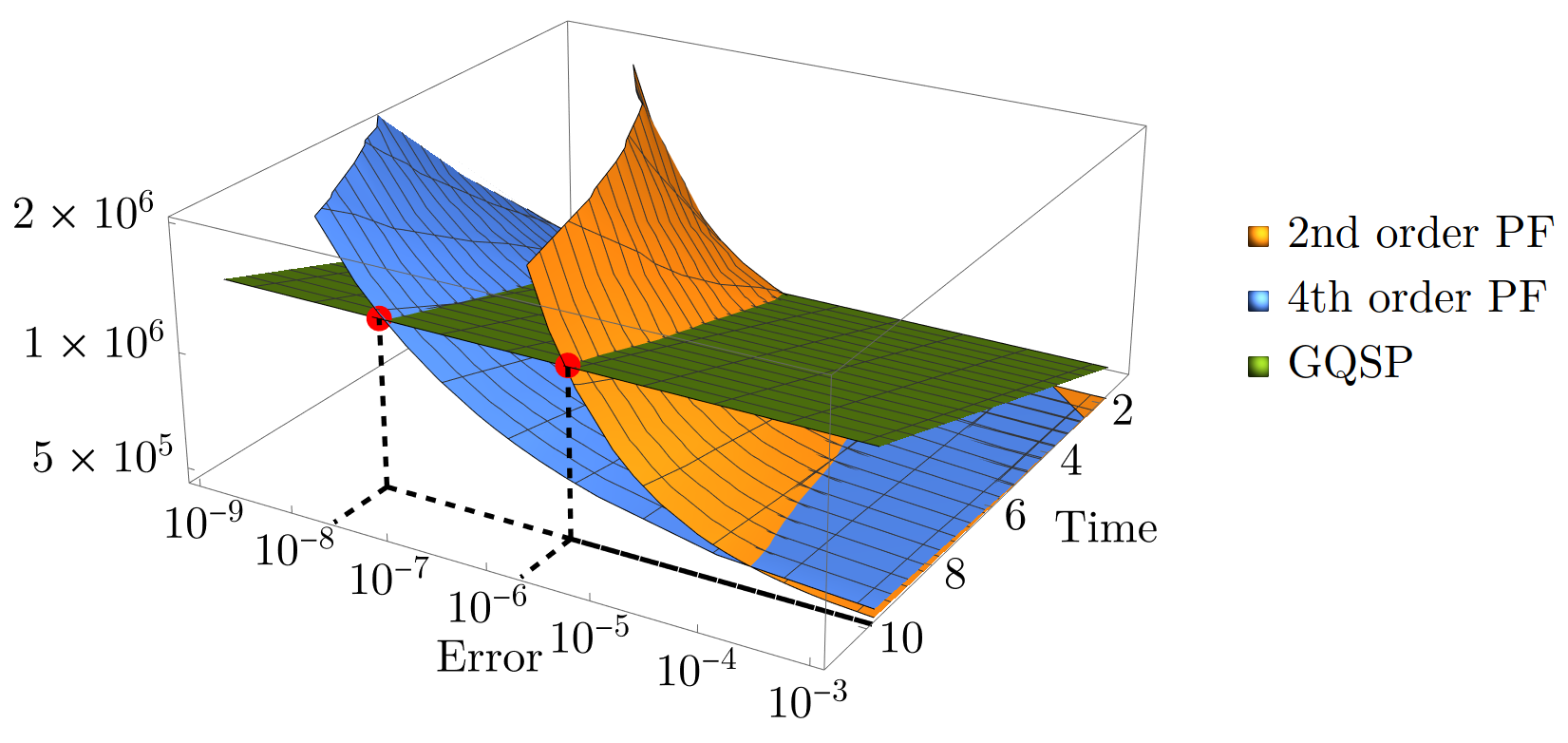}
    \caption{CNOT gate count as a function of error $\epsilon$ and simulation time $t$ for simulating time-evolution of the two-site Hamiltonian in Eq.~\eqref{eq:H_two_site}.
    The blue, orange, and green surfaces show results calculated using a $2^{\rm nd}$ order PF, a $4^{\rm th}$ order PF, and GQSP where the BE was constructing using \LOVELCU.
    For $t=10$, GQSP outperforms $2^{\rm nd}$ and $4^{\rm th}$ order PFs for $\epsilon \lesssim 1\times 10^{-6}$ and $\epsilon \lesssim 5\times 10^{-8}$.}
    \label{fig:two_site_cnot_vs_eps_t_nq3}
\end{figure*}

\FloatBarrier
\section{Scale factor for Block Encodings constructed using LCU}
\label{app:lcu_scale_factor}

In this appendix, we derive the scale factors for block encodings of $f_k(\hat{\xi}_k)$, given in \cref{eq:local_f_scalar_ft} in the main text, when prepared using standard LCU techniques.

To begin, we consider the scale factor for $\hat \varphi^d$; the result is analogous for $\bigl(\hat \pi^{(D)} \bigr)^d$.
Using \cref{eq:pauli_decomp_phi} we have
\begin{equation}
\begin{split}
    \hat \varphi^d &= (-\varphi_{\rm max})^d \left(\frac{2^{\nq-1}}{2^{\nq}-1}\right)^d \left(\sum_{m=0}^{\nq-1} 2^{-m} Z_m \right)^d
    \\
    &= (-\varphi_{\rm max})^d \left(\frac{2^{\nq-1}}{2^{\nq}-1}\right)^d \sum_{m_1, m_2, \dots, m_d=0}^{\nq-1} 2^{-(m_1+m_2+\dots + m_d)} Z_{m_1} Z_{m_2} \dots Z_{m_d}\,.
\end{split}
\end{equation}
Since $Z_{m_j}^2 = \unit$, the sum over $m_j$ includes multiple instances of the same Pauli string.
Constructing a BE using LCU requires grouping these repeated terms and determining the total coefficient for each unique Pauli string.
However, because the coefficient of each Pauli string in the sum has the same sign, no cancellations in the final values of the coefficients occur.
This fact implies that, even if there are repeated Pauli strings in the sum, the scale factor can be found by simply adding the coefficients of each Pauli string in the sum.
Using this fact, the scale factor is given by
\begin{equation}
\begin{split}
    \beta_{\hat\varphi^d} &=  \varphi_{\rm max}^d \left(\frac{2^{\nq-1}}{2^{\nq}-1}\right)^d \sum_{m_1, m_2, \dots, m_d=0}^{\nq-1} 2^{-(m_1+m_2+\dots + m_d)} 
    \\
    &= \varphi_{\rm max}^d \left(\frac{2^{\nq-1}}{2^{\nq}-1}\right)^d \left(\sum_{m=0}^{\nq-1} 2^{-m}\right)^d
    \\
    &= \varphi_{\rm max}^d \left(\frac{2^{\nq-1}}{2^{\nq}-1}\right)^d \left(\frac{2^{\nq}-1}{2^{\nq-1}}\right)^d
    \\
    &= \varphi_{\rm max}^d\,,
\end{split}
\end{equation}
where, going from the second to the third line, we used the identity $\sum_{m=0}^{\nq-1} 2^{-m} = (2^{\nq}-1)/2^{\nq-1}$.
From this result we can directly see the scale factor for $f_1(\hat \pi^{(D)}) = \frac{1}{2}\bigl(\hat \pi^{(D)}\bigr)^2$ is
\begin{equation}
    \beta_{f_1(\hat \pi^{(D)})} = \frac{1}{2} \pi_{\rm max}^2\,,
\end{equation}
which is the smallest possible scale factor to BE this operator.

Next, we turn to $f^{(1)}_0(\hat \varphi) = \frac{m^2}{2}\hat \varphi^2 + \frac{\lambda}{4!}\hat \varphi^4$, whose Pauli decomposition is
\begin{equation}
\begin{split}
    f^{(1)}_0(\hat \varphi) &= \frac{m^2}{2}\hat \varphi^2 + \frac{\lambda}{4!}\hat \varphi^4
    \\
    &= (-\varphi_{\rm max})^2 \left(\frac{2^{\nq-1}}{2^{\nq}-1}\right)^2\frac{m^2}{2} \sum_{m_1,m_2=0}^{\nq-1}2^{-(m_1+m_2)} Z_{m_1} Z_{m_2} + 
    \\
    &\hspace{0.5in} +(-\varphi_{\rm max})^4 \left(\frac{2^{\nq-1}}{2^{\nq}-1}\right)^4\frac{\lambda}{4!} \sum_{l_1,l_2,l_3,l_4=0}^{\nq-1}2^{-(l_1+l_2+l_3+l_4)} Z_{l_1} Z_{l_2} Z_{l_3} Z_{l_4}\,.
\end{split}
\end{equation}
Assuming $\lambda > 0$, we see that, similar to the previous case, the coefficient of each Pauli string in the decomposition of $f(\hat \varphi)$ has the same sign, and the scale factor is simply the sum of the coefficients
\begin{equation}
\begin{split}
    \beta_{f_0^{(1)}(\hat \varphi)} &= \frac{m^2}{2}\varphi_{\rm max}^2\left(\frac{2^{\nq-1}}{2^{\nq}-1}\right)^2 \left(\sum_{m=0}^{\nq-1}2^{-m}\right)^2 + \frac{\lambda}{4!}\varphi_{\rm max}^4\left(\frac{2^{\nq-1}}{2^{\nq}-1}\right)^4 \left(\sum_{m=0}^{\nq-1}2^{-m}\right)^4
    \\
    &= \frac{m^2}{2}\varphi_{\rm max}^2 + \frac{\lambda}{4!}\varphi_{\rm max}^4\,,
\end{split}
\end{equation}
which saturates the lower bound given by the operator norm.

The situation is more complicated for $f^{(2)}_0(\hat \varphi) = g \cos(\hat \varphi)$ as its Pauli decomposition consists of terms with opposite signs, requiring careful accounting for cancellations.
For this reason, we study the scale factor numerically.
Figure~\ref{fig:scale_factor_cos} presents the scale factor of the BE of $\cos(\hat \varphi)$ constructed using both standard LCU and \LOVELCU.
For this comparison, we set $\varphi_{\rm max} = \pi$.
With \LOVELCU, the scale factor approaches the optimal value of 1, whereas with standard LCU, it remains generally larger and asymptotes to $\sim 1.3$.
\begin{figure}
    \centering
    \includegraphics[width=0.5\linewidth]{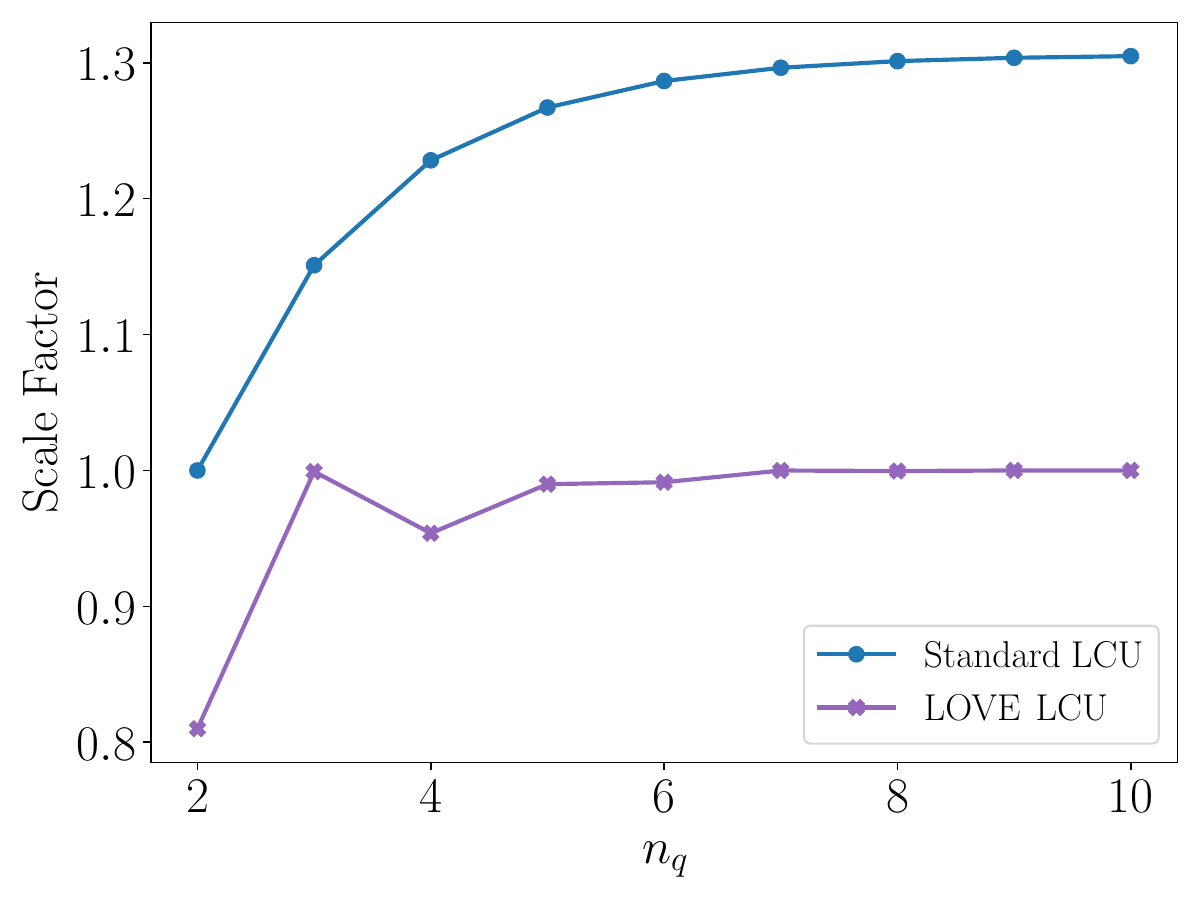}
    \caption{Scale factor of the BE of $\cos(\hat \varphi)$ as a function of $\nq$.
    The blue circles and purple crosses correspond to using standard LCU and \LOVELCU to construct the BE, respectively.
    The value of $\varphi_{\rm max}$ was set to $\varphi_{\rm max}=\pi$ for all values of $\nq$.}
    \label{fig:scale_factor_cos}
\end{figure}

Lastly, we consider $f_2(\hat{\varphi}_1 - \hat{\varphi}_2) = \frac{1}{2} (\hat{\varphi}_1 - \hat{\varphi}_2)^2$.  
To study the scale factor, we first rewrite it as $f_2(\hat{\varphi}_1 - \hat{\varphi}_2) = \frac{1}{2} (\hat{\varphi}_1^2 + \hat{\varphi}_2^2 - 2 \hat{\varphi}_1 \hat{\varphi}_2)$.
While this operator is a sum of terms with different signs, this does not affect the determination of the scale factor since the terms with differing signs have no Pauli strings in common.  
This can be verified by directly examining the decomposition of $f_2(\hat{\varphi}_1 - \hat{\varphi}_2)$.  
If we assign $\hat{\varphi}_1$ to act on qubits $0, 1, \dots, \nq - 1$ and $\hat{\varphi}_2$ to act on qubits $\nq, \nq + 1, \dots, 2\nq - 1$, the decomposition is
\begin{align}
    f_2(\hat \varphi_1-\hat\varphi_2) &= \frac{1}{2}(\hat \varphi_1^2 +\hat\varphi_2^2 - 2 \hat\varphi_1 \hat \varphi_2)
    \\
    &= \frac{1}{2}\varphi_{\rm max}^2\left(\frac{2^{\nq-1}}{2^{\nq}-1}\right)^2 \nonumber
    \\
    &\hspace{0.2in} \times \Big(\sum_{m_1,m_2=0}^{\nq-1} 2^{-(m_1+m_2)} Z_{m_1}Z_{m_2} + \sum_{l_1,l_2=0}^{\nq-1} 2^{-(l_1+l_2)} Z_{\nq+l_1}Z_{\nq+l_2} - 2\sum_{m,l=0}^{\nq-1} 2^{-(m+l)} Z_m Z_{\nq+l}\Big) \nonumber\,.
\end{align}
From this expression, we see that all Pauli strings with opposite signs are distinct. 
Thus, the scale factor can be determined by summing the coefficients.
Doing so gives
\begin{equation}
\begin{split}
    \beta_{f_2(\hat \varphi_1-\hat\varphi_2)} &= \frac{1}{2}\varphi_{\rm max}^2\left(\frac{2^{\nq-1}}{2^{\nq}-1}\right)^2 \Big(\sum_{m_1,m_2=0}^{\nq-1} 2^{-(m_1+m_2)} + \sum_{l_1,l_2=0}^{\nq-1} 2^{-(l_1+l_2)} + 2\sum_{m,l=0}^{\nq-1} 2^{-(m+l)} \Big)
    \\
    &= \frac{1}{2}\varphi_{\rm max}^2 \left(\frac{2^{\nq-1}}{2^{\nq}-1}\right)^2 4 \Big(\sum_{m=0}^{\nq-1} 2^{-m}\Big)^2
    \\
    &= 2\varphi_{\rm max}^2.
\end{split}
\end{equation}
Since the operator $\hat\varphi$ is sampled symmetrically from $-\varphi_{\rm max}$ to $\varphi_{\rm max}$, this scale factor attains its smallest possible value.
\end{document}